\documentclass[a4paper,11pt]{article}

\usepackage[margin=2.5cm]{geometry}

\usepackage{lmodern} 
\usepackage{blindtext}
\usepackage{float}
\usepackage[T1]{fontenc} 
\usepackage[utf8]{inputenc}

\usepackage[backend=biber,date=year,doi=true,isbn=false,url=false,maxbibnames=5,sorting=none,style=numeric-comp]{biblatex}
\defbibfilter{appendixOnlyFilter}{
	segment=1 
	and not segment=0 
}
\newbibmacro{string+doiurlisbn}[1]{%
  \iffieldundef{doi}{%
    \iffieldundef{url}{%
      \iffieldundef{isbn}{%
        \iffieldundef{issn}{%
          #1%
        }{%
          \href{http://books.google.com/books?vid=ISSN\thefield{issn}}{#1}%
        }%
      }{%
        \href{http://books.google.com/books?vid=ISBN\thefield{isbn}}{#1}%
      }%
    }{%
      \href{\thefield{url}}{#1}%
    }%
  }{%
    \href{http://dx.doi.org/\thefield{doi}}{#1}%
  }%
}
\DeclareFieldFormat{title}{\usebibmacro{string+doiurlisbn}{\mkbibemph{#1}}}
\DeclareFieldFormat[article,thesis,incollection,inproceedings]{title}{\usebibmacro{string+doiurlisbn}{\mkbibquote{#1}}}

\usepackage{mathtools} 
\usepackage{amsfonts}
\usepackage{amsmath} 
\usepackage{dsfont}
\usepackage{physics}
\usepackage{enumerate}
\usepackage[shortlabels]{enumitem}
\usepackage[font=small,labelfont=bf]{caption}
\usepackage{tikz}
\usetikzlibrary{hobby} 
\usepackage{subcaption}
\usepackage{graphicx}
\graphicspath{{plots/}} 
\usepackage{placeins} 
\usepackage[normalem]{ulem} 
\usepackage{array} 
\newcolumntype{C}[1]{>{\centering\let\newline\\\arraybackslash\hspace{0pt}}m{#1}}
\usepackage{color}
\newcolumntype{L}[1]{>{\raggedright\let\newline\\\arraybackslash\hspace{0pt}}m{#1}}

\usepackage{multirow}
\allowdisplaybreaks
\usepackage[colorlinks=true,linkcolor=blue,citecolor=citegreen]{hyperref}
\usepackage{enumerate}
\usepackage{algorithm}
\usepackage{algpseudocode}

\usepackage{amsthm}

\usepackage[affil-it]{authblk}
\usepackage[nameinlink,capitalise]{cleveref}
\usepackage{aliascnt}

\newtheorem{theorem}{Theorem}
\crefname{theorem}{Theorem}{Theorems}
\newaliascnt{lemma}{theorem}
\newtheorem{lemma}[lemma]{Lemma}
\aliascntresetthe{lemma}
\crefname{lemma}{Lemma}{Lemmas}
\newaliascnt{corollary}{theorem}
\newtheorem{corollary}[corollary]{Corollary}
\aliascntresetthe{corollary}
\crefname{corollary}{Corollary}{Corollaries}
\newaliascnt{definition}{theorem}
\newtheorem{definition}[definition]{Definition}
\aliascntresetthe{definition}
\crefname{definition}{Definition}{Definitions}
\newaliascnt{conjecture}{theorem}

\aliascntresetthe{conjecture}
\crefname{conjecture}{Conjecture}{Conjectures}
\newaliascnt{remark}{theorem}

\aliascntresetthe{remark}
\crefname{remark}{Remark}{Remarks}
\newaliascnt{proposition}{theorem}

\aliascntresetthe{proposition}
\crefname{proposition}{Proposition}{Propositions}
\crefname{section}{Section}{Sections}
\crefname{subsection}{Subsection}{Subsections}
\crefname{appendix}{Appendix}{Appendices}
\crefname{equation}{Eq.}{Eqs.}

\renewcommand{\norm}[1]{\|#1\|}

\renewcommand{\i}{\mathrm{i}}

\newcommand{\identity}{\mathds{1}}

\newcommand{\id}{\identity}
\newcommand{\chan}{\mathcal{E}}
\newcommand{\curlU}{\mathcal{U}}
\newcommand{\curlA}{\mathcal{A}}
\newcommand{\arr}[1]{\overrightarrow{#1}}
\newcommand{\sumvec}[1]{\sum_{\arr{#1}}}
\newcommand{\sumvecp}[1]{\sum_{\arr{#1}'}}

\newcommand{\psim}{\psi_{\arr{m}}}

\newcommand{\phimv}[2]{\phi_{\arr{#1},\arr{#2}}}

\newcommand{\phimpv}[2]{\phi_{\arr{#1}',\arr{#2}}}

\newcommand{\cte}{\widetilde{C}}

\DeclareMathOperator{\poly}{poly}

\DeclareMathOperator{\gen}{Gen}

\definecolor{citegreen}{RGB}{0,165,0}

\renewcommand{\paragraph}[1]{\textbf{\textit{#1}}}

\usepackage{todonotes}
\usepackage{xcolor}
\usepackage{booktabs}
\usepackage{multirow}
\usepackage{multicol}
\usepackage{qcircuit}

\renewcommand{\rm}[1]{\mathrm{#1}}
\renewcommand{\cal}[1]{\mathcal{#1}}
\newcommand{\bb}[1]{\mathbb{#1}}
\renewcommand{\ketbra}[1]{| #1  \rangle \!  \langle  #1 |}
\newcommand{\SWAP}{\mathrm{SWAP}}
\newcommand{\red}[1]{{\color{red}#1}}

\DeclareMathOperator{\Var}{Var}
\definecolor{orange-red}{rgb}{1.0, 0.27, 0.0}
\newcommand{\Xiaug}{\Xi^*}
\newcommand{\FNL}{\Phi}

\algrenewcommand\algorithmicrequire{\textbf{Input:}}
\algrenewcommand\algorithmicensure{\textbf{Output:}}

\title{Optimal and improved gate decompositions for accelerated classical simulation of near-Gaussian fermionic circuits}
\author[1,2,3]{Beatriz Dias}
\author[1]{Jan Lukas Bosse}
\author[1]{James R. Seddon}
\affil[1]{Phasecraft Ltd, London, UK}
\affil[2]{Department of Mathematics, School of Computation, Information and Technology, Technical~University~of~Munich, Germany}
\affil[3]{Munich Center for Quantum Science and Technology, Munich, Germany}

\begin{document}
\date{\today}
\maketitle

\begin{abstract}
Fermionic Gaussian circuits can be simulated efficiently on a classical computer, but become universal when supplemented with non-Gaussian operations. Similar to stabilizer circuits augmented with non-stabilizer resources, these non-Gaussian circuits can be simulated classically using rank- or extent-based methods. These methods decompose non-Gaussian states or operations into Gaussian ones, with runtimes that scale polynomially with measures of non-Gaussianity such as the rank and the extent -- quantities that typically grow exponentially with the number of non-Gaussian resources. Current fermionic rank- and extent-based simulators have mostly been limited to Gaussian circuits with magic-state injection. Extending them to mixed states and non-unitary channels has been hindered by the lack of known extent-optimized decompositions for physically relevant gates and noisy channels. In this work, we address this gap. First, we derive analytic decompositions for key non-Gaussian gates and channels, including decompositions for arbitrary two-qubit fermionic gates which are provably optimal for diagonal gates or those acting on Jordan-Wigner-adjacent qubit pairs.  Second, we show that stochastic Pauli noise can reduce the effective extent of non-Gaussian rotation gates, but that fermionic magic is substantially more robust to such noise than stabilizer magic. Finally, we demonstrate how these decompositions can accelerate classical sampling from the output distribution of a quantum circuit. This involves a generalization of existing sparsification methods, previously limited to convex-unitary channels, to circuits involving intermediate measurements and feed-forward. Our decompositions also yield speedups for emulating noisy Pauli rotations with quasiprobability simulators in the large-angle/arbitrary-strength-noise and small-angle/low-noise parameter regimes.
\end{abstract}

\newpage
\tableofcontents

\section{Introduction}

Quantum computers are expected to eventually exceed the capabilities of classical computers in performing certain specialized tasks. Problems such as factoring and unstructured search famously admit quantum algorithms~\cite{10.1109/SFCS.1994.365700,10.1145/237814.237866} offering  superpolynomial and polynomial speedups respectively compared to the fastest known classical algorithms.  Quantum computers also provide a natural platform for simulating the behaviour of other quantum systems~\cite{Lloyd2026univQSim}. It has been shown that there exist quantum simulation tasks that can be solved in polynomial time on a quantum computer, but cannot be solved efficiently classically unless widely-held complexity-theoretic conjectures are violated~\cite{Jordan2018bqpcomplete}. Qubit-based devices are naturally suited to studying spin systems, while fermionic systems can be simulated via fermion-to-qubit mappings such as the Jordan-Wigner transform~\cite{jordan_uber_1928}. The simulation of fermionic systems is of key interest in fields such as materials science~\cite{clinton_towards_2024} and quantum chemistry~\cite{cao_quantum_2019,McArdle2020qchem}. Recent years have seen increasingly large-scale demonstrations of spin~\cite{IBM2023utility,quantinuum2025digitalquantummagnetism,zhang2025googleNMR} and fermion~\cite{quantinuum2025FH,evered_probing_2025,phasecraft2025quantinuum,phasecraft2025google} models simulated on real quantum hardware, and this has in turn helped to drive parallel progress in the simulation of quantum computers by classical methods~\cite{PhysRevLett.69.2863,PhysRevLett.91.147902,Verstraete01032008,SCHOLLWOCK201196,10.1145/3564246.3585234, Tindall2024kickedIsingTN,fontana2023classicalsimulationsnoisyvariational,PhysRevA.99.062337,rudolph2023classicalsurrogatesimulationquantum,schuster2024polynomialtimeclassicalalgorithmnoisy,miller2025simulationfermioniccircuitsusing,PhysRevX.6.021043,Bravyi_2016_fastnorm,Bravyi_2019,Dias_2024,reardonsmith2024improved,PhysRevA.110.042402,hahn2024classicalsimulationquantumresource,gottesman1997stabilizercodesquantumerror,gidney_stim_2021,PhysRevA.70.052328,howard2017application,heinrich2018robustness,seddon2019quantifying,bourassaFastSimulationBosonic2021,10.1145/380752.380785,PhysRevA.65.032325,knill2001fermionic,bravyiLagrangianRepresentationFermionic2004,WANG20071}.

General quantum circuits are believed to be hard to simulate classically, requiring time and space exponential in the system size, but any claims of genuine quantum advantage for contemporary noisy intermediate-scale quantum (NISQ) devices are contingent on the limitations of currently known classical simulation techniques. Indeed tensor network methods~\cite{PhysRevLett.69.2863,PhysRevLett.91.147902,Verstraete01032008,SCHOLLWOCK201196}
have proved remarkably successful in simulating circuit sizes comparable to those currently available on hardware, and have recently been used to refute particular quantum advantage claims~\cite{Tindall2024kickedIsingTN}. The main barrier to practical quantum computation is the susceptibility of quantum devices to physical noise, so that competing classical simulation techniques need only simulate noisy, rather than perfect, quantum evolution in order to match the capabilities of current quantum hardware. Indeed, it is well known in several settings that the presence of noise reduces the overhead for classical simulation of quantum circuits, and sufficiently strong noise can wash out any possibility of quantum advantage. Classical simulation also plays a vital role in benchmarking, verification, and error mitigation~\cite{Czarnik2021errormitigation,TFLO} in current quantum hardware demonstrations~\cite{IBM2023utility,quantinuum2025digitalquantummagnetism,zhang2025googleNMR,quantinuum2025FH,evered_probing_2025,phasecraft2025quantinuum,phasecraft2025google}. Aside from verifying and aiding near term quantum  computations of fermionic systems, being able to simulate a wider class of fermionic systems on classical computers also has applications in computational physics. Being able to go beyond the mean-field description of free fermions is a crucial ingredient in many post-Hartree-Fock methods such as the Coupled Cluster~\cite{Kuemmel1991,ek1991} or the Configuration Interaction~\cite{DavidSherrill1999} methods.

The question of which classical simulation method is best suited to a given task depends on the particular setting. The most direct strategy is to explicitly store the state vector or density matrix, which requires memory constant in the circuit depth but scaling exponentially with the system size. This is feasible for small systems, but quickly becomes infeasible for a few tens of qubits. More sophisticated methods exploit circuit structure. Tensor network approaches~\cite{PhysRevLett.69.2863,PhysRevLett.91.147902,Verstraete01032008,SCHOLLWOCK201196}
leverage low entanglement to compress the description of a quantum state.  The more recent Pauli~\cite{10.1145/3564246.3585234,fontana2023classicalsimulationsnoisyvariational,PhysRevA.99.062337,rudolph2023classicalsurrogatesimulationquantum,schuster2024polynomialtimeclassicalalgorithmnoisy} and Majorana~\cite{miller2025simulationfermioniccircuitsusing} path propagation methods rely on the fact that it is often cheaper to backpropagate (Heisenberg-evolve) an observable through a circuit in a suitable basis, truncating negligible terms, rather than track the forward evolution of a quantum state. Randomness in circuit structure and the presence of incoherent noise both prove important in making these path propagation methods competitive.

Rank-based approaches are another class of method that simulate circuits built from classically tractable models (e.g., stabilizer or Gaussian operations) augmented with magic gates (or states)~\cite{PhysRevX.6.021043,Bravyi_2016_fastnorm,Bravyi_2019,Dias_2024,reardonsmith2024improved,PhysRevA.110.042402,hahn2024classicalsimulationquantumresource}. Stabilizer operations, encompassing Clifford circuits with stabilizer initial states, Pauli measurements and classical feed-forward, can be simulated classically in an efficient way using the stabilizer formalism~\cite{gottesman1997stabilizercodesquantumerror,PhysRevA.70.052328}. State-of-the-art stabilizer simulators have been demonstrated for Clifford circuits involving tens of thousands of qubits and millions of operations~\cite{gidney_stim_2021}.  Despite not being universal, stabilizer computations can be promoted to universality by consuming additional resources: either by supplementing the set of stabilizer operations with magic gates (e.g., the~$T$ gate) or equivalently by adding ancilla qubits initialized in magic states~\cite{PhysRevA.71.022316}. Rank-based methods simulate these augmented stabilizer circuits by describing the state vector as a linear combination of stabilizer states~\cite{PhysRevX.6.021043,Bravyi_2016_fastnorm,Bravyi_2019}. The stabilizer rank is an example of a magic monotone, meaning it is non-increasing under stabilizer operations, and is defined as the minimum number of stabilizer terms needed to represent the state vector. It generally scales exponentially in the number of magic gates implemented or magic states consumed within the circuit, and it is this rank that governs the runtime and space requirements of exact classical simulation. Admitting approximate simulation does not allow evasion of asymptotic exponential scaling, but does buy significant reductions in runtime, by replacing an exact high-rank state decomposition with a sparsified approximation with greatly reduced rank. It turns out that the approximate rank has an upper bound proportional to a measure of magic known as stabilizer extent, which is the minimal squared $L^1$-norm over all possible decompositions of a given non-stabilizer state into stabilizer terms. Despite the unavoidable exponential scaling, these methods are feasible for the simulation of circuits with little injected magic resource, even on large system sizes. Additionally, recent results show that considering different types of noise can either increase or decrease the hardness of simulation for this type of algorithm. For example, stochastic Pauli noise channels, such as the commonly considered depolarizing and dephasing channels, are stabilizer-preserving and tend to reduce simulation overhead, whereas non-stabilizer-preserving error channels such as amplitude-damping noise can increase runtime~\cite{seddon2019quantifying}.

A classically simulable quantum model that is more natural to the setting of fermionic quantum simulation is the set of (fermionic) Gaussian circuits~\cite{10.1145/380752.380785,PhysRevA.65.032325,knill2001fermionic,bravyiLagrangianRepresentationFermionic2004}. Unlike stabilizer circuits, fermionic Gaussian circuits are motivated by a well-known physical model which is integrable: free fermion dynamics, i.e., the dynamics of non-interacting fermions. Simultaneously, the bosonic flavour of these free models known as (bosonic) linear optics also admits efficient classical simulation algorithms~\cite{WANG20071}. Analogously to the stabilizer setting, supplementing fermionic Gaussian operations with appropriate (magic) gates or states leads to universality~\cite{doi:10.1098/rspa.2008.0189,PhysRevA.84.022310}.  Recent works based on minimal extensions of the stabilizer formalism have been adapted to the Gaussian context~\cite{Dias_2024,reardonsmith2024improved}, where fermionic magic states are decomposed into linear combinations of Gaussian states. This adaptation is primarily based on a novel phase-sensitive Gaussian simulator which is a minimal extension of the standard efficient simulation algorithms for Gaussian circuits (which do not keep track of phases). Similar algorithms have also been developed for non-Gaussian bosonic circuits~\cite{bourassaFastSimulationBosonic2021,PhysRevA.110.042402,hahn2024classicalsimulationquantumresource}. 

While in the context of near-stabilizer circuits, different algorithms have been developed based on decomposing either states, gates or channels into linear combinations of stabilizer components~\cite{Bravyi_2016_fastnorm,PhysRevX.6.021043,Bravyi_2019,Seddon_2021,Seddon_2021,howard2017application,heinrich2018robustness}, most work in the fermionic Gaussian setting has focused on decompositions of (magic) states~\cite{Brod2016,Dias_2024,cudby2024gaussiandecompositionmagicstates,reardonsmith2024improved}. This means that existing classical algorithms in the context of near-Gaussian circuits do not enable the application of magic gates directly. Instead, magic gates are implemented by introducing ancilla qubits initialized in the corresponding magic state and by applying an appropriate gadget of Gaussian operations, a technique known as gadgetization or magic state injection. 

Another line of research has considered the decomposition of channels into a linear combination of fermionic Gaussian maps~\cite{Hakkaku_2022} with weights that can be non-positive, forming a quasiprobability distribution. This leads to a Monte Carlo type  simulator, operating at the level of density matrices rather than state vectors, where observables are estimated by sampling Gaussian maps based on their weights and averaging over many trials.  The sampling cost is governed by the $L^1$-norm of the quasiprobability distribution, known as the fermionic nonlinearity~\cite{Hakkaku_2022}, and is again expected to scale exponentially with the number of magic gates. However here the sampling cost scales quadratically with the monotone rather than linearly in the case of Born rule probability estimation using extent-based techniques (see e.g.,~\cite[Section 6]{reardonsmith2024improved}). 

In this work, we obtain speedups for approximate rank-based methods for classically simulating fermionic non-Gaussian circuits, extending the class of circuits for which simulation is practically possible given finite classical computational resources. Our results hinge on the fact that the runtime and memory requirement for simulating circuit evolution and measurement using these methods scales linearly with the $L^1$-norm of some known decomposition of the evolved state vector into Gaussian terms. It is crucial to minimize this quantity as far as possible at the gate level, as modest improvements in the decomposition of individual gates can lead to reduction in runtime by a factor exponential in the number of resourceful gates. For example, for a circuit with $t$ instances of some non-Gaussian gate $U$, the runtime scales as $O(c^{2t})$ if the gate decomposition used has $L^1$-norm $c$. Finding an improved decomposition of $U$ with reduced norm $c/\alpha$ for some $\alpha>1$ will decrease the runtime for simulating the circuit by a factor $\alpha^{-2t}$. The optimal squared $L^1$-norm achievable for a given magic state or gate is called its Gaussian extent.

While it is straightforward to find a feasible Gaussian decomposition, it is nontrivial to prove that a given decomposition is optimal. Finding an optimal decomposition numerically is hampered by the fact that, unlike in the stabilizer case, where the unitary extent could be found by convex optimization over the finite set of Clifford gates, the set of Gaussian unitary operations is continuous. Previous works on the Gaussian extent have focused on proving optimality for decompositions of pure resource states used for magic state injection \cite{Dias_2024,reardonsmith2024improved}. Pure state and unitary extent can be extended to mixed states and channels, via so-called convex roof extensions \cite{Uhlmann1998entropy,Uhlmann2010convexroof,Regula_2018,Seddon_2021,seddonThesis}. Computing these generalized resource measures numerically is harder still, as it involves non-convex optimization over the infinite set of all possible ensemble decompositions into pure states or unitary gates \cite{Adriazola2025nonconvex}.  Our strategy in this work is to obtain explicit analytic decompositions for unitary gates and channels commonly used in fermionic circuits, finding extent-optimal decompositions wherever possible, and improved decompositions in other cases.
 
Our most practically significant contribution is to provide analytic decompositions for non-Gaussian gates subject to stochastic Pauli noise, such that the effective extent (and therefore classical simulation cost) is reduced relative to that of the noiseless gate. In particular, we show that our decompositions are optimal for particular gate/noise combinations, such as two-qubit $ZZ$-rotations subject to Pauli noise of the same type. The results for noisy channels are facilitated by a detailed study of unitary gate decompositions. We provide analytic decompositions for arbitrary two-qubit fermionic gates that are provably optimal with respect to the Gaussian unitary extent for diagonal gates or those acting on on Jordan-Wigner-adjacent pairs. By providing direct decompositions of unitary non-Gaussian gates, we avoid the need to recompile the circuit as a Gaussian circuit with injected magic states, and sidestep overhead associated with classically simulating the intermediate measurements needed for state injection. This also allows us to simulate non-parity-preserving single-qubit gates that cannot be implemented by magic state injection, and we provide optimal decompositions for gates such as the Hadamard, and $X$ and $Y$ rotations along with their noisy counterparts. We study the link between Gaussian extent and fermionic nonlinearity, showing that our circuit decompositions can be used both for bit-string sampling using rank-based classical simulation, and for observable estimation using quasiprobability sampling~\cite{Hakkaku_2022}. Moreover, we show that our analytic decompositions of noisy $ZZ$ rotations improve on those determined numerically in~\cite{Hakkaku_2022} over a large range of parameters, namely for all noise strengths in the large-angle regime, and for low noise strength in the small-angle regime. Our work illustrates how the presence of noise shifts the boundary between circuits that are tractable for current classical computing hardware and those that are not.

The rest of this manuscript is structured as follows. In~\cref{sec:priorwork} we survey key measures of magic introduced in prior work that quantify classical simulation cost, before detailing our main technical contributions in~\cref{sec:contributions}. In~\cref{sec:background} we provide background on fermionic Gaussian states and operations, their classical simulability and fermionic magic states and gates. In~\cref{sec:magicmonotones} we survey magic monotones, including the unitary rank, convex-unitary channel extent and augmented channel extent. \cref{sec:extentplus} gives extent-optimal decompositions for a family of non-fermionic states, closely related to a construction given by Brod \cite{Brod2016}. We derive extent-optimal and improved decompositions for unitary gates in~\cref{sec:Udecomp} and for noisy channels in~\cref{sec:noisydecomps}. \cref{sec:FNL} develops connections between Gaussian extent and fermionic nonlinearity~\cite{Hakkaku_2022}, showing that our optimal and improved gate decompositions can be repurposed for obtaining speedups in observable estimation via quasiprobability sampling.  In~\cref{sec:algorithm} we describe a classical simulation algorithm for drawing bit strings from the output distribution of a non-Gaussian circuit. In~\cref{sec:SPARSE} we discuss sparsification of non-Gaussian convex-unitary and adaptive circuits, which reduces the effective Gaussian rank, establishing the technical link between our resource-theoretic results and the cost of classical simulation.  In the appendices, we prove some technical results used in the main text. We conclude in~\cref{sec:conclusion}.

\subsection{Prior work \label{sec:priorwork}}

Rank-based methods for simulating Clifford circuits with few non-Clifford resources scale polynomially with the system size, the number of Clifford gates in the circuit, and certain magic monotones which quantify the amount of magic added to the circuit~\cite{PhysRevX.6.021043,Bravyi_2016_fastnorm,Bravyi_2019}. The latter is usually the bottleneck, since magic monotones typically scale exponentially with the number of magic resources added. 
In general terms, these methods work as follows: the magic states are decomposed into superpositions of stabilizer states, and each term in the superposition is evolved under a stabilizer circuit using an extension of the Gottesmann-Knill theorem that is capable of recording and updating the phase in each term~\cite{Bravyi_2019}. The latter is a functionality that is not present in the standard Gottesman-Knill theorem~\cite{gottesman1997stabilizercodesquantumerror} and that is required here in order to retain information about the relative phases in superpositions of stabilizer states. Because of this construction, the added simulation runtime due to including magic states is naturally quantified by the rank of the initial magic state, first introduced in~\cite{PhysRevX.6.021043}:

\begin{quote}
	The \emph{(state) rank}~$\chi(\ket{\psi})$ of a state~$\psi$ is the minimum number of terms when writing the (possibly non-stabilizer) state~$\psi$ as a superposition of stabilizer states.
\end{quote}

In~\cite{PhysRevX.6.021043}, the rank-based algorithms proposed scale quadratically with the rank~$\chi(\psi)$ of the initial state~$\psi$. For approximate simulation, the runtime can be improved to linear in the rank~$\chi(\psi)$~\cite{Bravyi_2016_fastnorm}. This is due to a method proposed in~\cite{Bravyi_2016_fastnorm} for estimating norms (and hence computing measurement probabilities) of linear combinations of stabilizer states in time linear in the number of terms, which contrasts with the quadratic cost of computing these norms exactly.

Although the rank is an operationally relevant quantity, there are no known efficient methods to compute it. Moreover, since the rank does not take into account the weight of different terms in superposition, it does not cleanly evaluate how close a given state is to being a free (stabilizer) state, and does not capture fine-grained differences in the resourcefulness of different states. This motivated the definition of the extent $\xi(\ket{\psi})$ of a state $\ket{\psi}$ in \cite{Bravyi_2019}.

\begin{quote}
	The \emph{(pure state) extent}~$\xi(\ket{\psi})$ of a state~$\ket{\psi}$ is the minimum~$L^1$-norm squared of the vector of coefficients when writing the (possibly non-stabilizer) state~$\ket{\psi}$ as a superposition of stabilizer states.
\end{quote}

The extent $\xi(\ket{\psi})$ is related to an approximate version of the rank $\chi(\ket{\psi})$ (see~\cite{Bravyi_2016_fastnorm,Bravyi_2019} for a definition of the ``approximate stabilizer rank''), and it quantifies the cost of approximate simulation when initializing a stabilizer circuit with (a state close to) the magic state $\ket{\psi}$ -- the optimal runtime is linear in the extent $\xi(\ket{\psi})$. Contrary to the rank, the extent has the advantage that it can be formulated as a convex optimization program, specifically a second-order cone program. Hence, it can be computed by polynomial time algorithms in the size of the problem (see e.g.~\cite{Boyd_Vandenberghe_2004}), which here is the total number of stabilizer states which is super-exponential in the number of qubits. Nevertheless, the extent can be computed numerically for small system sizes. In addition, convex optimization techniques have been used to prove that the extent of tensor products $\bigotimes_j \ket{\psi_j}$ of up to 3-qubit states $\ket{\psi_j}$ is multiplicative, i.e., $\xi(\bigotimes_j \ket{\psi_j}) = \prod_j \xi(\ket{\psi_j})$~\cite{Bravyi_2019}. In contrast, the extent of tensor products of sufficiently high-dimensional states is strictly submultiplicative, which in principle enables performance gains~\cite{Heimendahl2021stabilizerextentis}, but makes the convex optimization problem harder. Recent works have focused on obtaining upper bounds for the stabilizer extent in cases where it cannot be computed efficiently: the approach in~\cite{Qassim2021improvedupperbounds} uses the contraction of certain cat-states; more recently~\cite{Sutcliffe_2024,wan2025simulatemagicstatecultivation} have improved the known upper bounds using ZX-calculus.  

An alternative approach in rank-based simulation is to decompose gates instead of states. The authors in Ref.~\cite{Bravyi_2019} take this route and develop a ``sum-over-Cliffords'' method which avoids gadgetization by decomposing the magic gates directly into linear combinations of Clifford gates. Here, the runtime is quantified by quantities similar to the rank and extent, but for gates:
\begin{quote}
	The \emph{unitary rank}~$\chi(U)$ of a (unitary) gate~$U$ is the minimum number of terms when writing the (possibly non-Clifford) gate~$U$ as a linear combination of Clifford gates.
\end{quote}
\begin{quote}
	The \emph{unitary extent}~$\xi(U)$ of a (unitary) gate~$U$ is the minimum~$L^1$-norm squared of the vector of coefficients when writing the (possibly non-Clifford) gate~$U$ as a linear combination of Clifford gates.
\end{quote}

More recent works have sought to extend rank-based methods, previously confined to pure states and unitary evolution, to density operators and more general channels. The work~\cite{Seddon_2021} introduced a simulator for sampling from a stabilizer circuit applied to a mixed state, with performance quantified by a generalization of the (state) extent for mixed states. Reference~\cite{seddonThesis} analyzes similar methods, in addition to algorithms based on decomposing channels instead of states, with runtime scaling with the so-called channel extent:
\begin{quote}
	The \emph{convex-unitary channel extent}~$\Xi(\cal{E})$ of a channel~$\cal{E}$ is the convex-roof extension~\cite{Regula_2018} of the unitary extent (see~\cref{def:channelextent} for a formal definition).
\end{quote}

Another set of methods available are quasiprobability-based simulators~\cite{Pashayan2015,howard2017application,Bennink2017,heinrich2018robustness,seddon2019quantifying,Seddon_2021,denzler2026frames}. These decompose density matrices or channels into sums of stabilizer density operators or channels, unlike rank-based methods that work at the level of pure states or use convex-roof extensions. Quasiprobability algorithms typically scale quadratically with the associated magic monotone, leading to slower runtimes compared to rank- and extent-based methods, which in the case of approximate simulation scale linearly with the rank or extent. This speedup comes from efficiently computing probabilities as norms of superpositions, which quasiprobability methods cannot exploit. Examples of magic monotones for quasiprobability simulation include the robustness of magic~\cite{howard2017application,heinrich2018robustness,seddon2019quantifying}, generalized robustness~\cite{seddon2019quantifying,Seddon_2021,Saxena2022}, and dyadic negativity~\cite{Seddon_2021,seddonThesis}. Parallel with the setting of rank/extent-based methods, one can either focus on gadgetized Clifford circuits consuming magic states as a resource~\cite{howard2017application,Seddon_2021}, or decompose gates and noisy channels in non-Clifford circuits directly~\cite{Bennink2017,seddon2019quantifying,seddonThesis,Saxena2022}.

While most efforts have been devoted to simulating stabilizer circuits with non-stabilizer resources, similar methods apply to (fermionic) Gaussian computations with non-Gaussian resources. Recent works~\cite{Dias_2024,reardonsmith2024improved} have extended rank-based simulation to the fermionic setting: states are decomposed as superpositions of (fermionic) Gaussian states, and algorithms rely on the efficient classical simulability of  Gaussian dynamics augmented with the ability to keep track of phases in superpositions.  These algorithms scale polynomially with the Gaussian rank/extent, and the Gaussian extent is multiplicative for tensor products of up to four-qubit states~\cite{cudby2024gaussiandecompositionmagicstates,reardonsmith2024fermioniclinearopticalextent}. Similarly to the stabilizer setting, a fast-norm algorithm enables a quadratic-to-linear speedup for estimating norms of superpositions of Gaussian states~\cite{bravyiComplexityQuantumImpurity2017a}.

In contrast, Ref.~\cite{Hakkaku_2022} uses quasiprobability methods, decomposing non-Gaussian channels into non-positive linear combinations of Gaussian maps and seeking to optimize decompositions with respect to the so-called fermionic nonlinearity:
\begin{quote}
	The \emph{fermionic nonlinearity}~$\FNL(\cal{E})$ of a fermionic channel~$\cal{E}$ is the minimum~$L^1$-norm of the vector of coefficients when decomposing~$\cal{E}$ into a sum of terms of the form~$L(\cdot)R$ with~$L,R$ Gaussian unitaries. 
\end{quote}
Rank methods for non-Gaussian simulation based on gate/channel decompositions via convex-roof extensions, which can be more efficient than quasiprobability methods (see e.g.~\cite[Section 6]{reardonsmith2024improved}), remain largely unexplored.

\subsection{Our contribution}\label{sec:contributions}

We summarize our novel technical results in four key contributions. 

\subsubsection*{Result 1: Optimal decompositions of non-Gaussian gates into Gaussian gates}

We prove optimality with respect to the unitary extent for the decompositions given in~\cref{tab:optimaldecomp} for
\begin{enumerate}[1)]
	\item \label{it:gatedecomp1} any two-qubit fermionic gate which a) acts on nearest neighbour qubits or b) is diagonal,
	\item the (single-qubit) Hadamard gate~$H$,
	\item the single-qubit rotation~$R_{Y}(\theta)$ around the~$Y$-axis,
	\item the single-qubit rotation~$R_{X}(\theta)$ around the~$X$-axis.
\end{enumerate}
As concrete examples of~\cref{it:gatedecomp1}, we provide explicit optimal decompositions for 
\begin{enumerate}[i)]
	\item the controlled-phase gate~$C(\theta)$,
	\item the two-qubit $ZZ$ rotation~$R_{ZZ}(\theta)$, 
	\item the~$\SWAP$ gate acting on nearest neighbours.
\end{enumerate}
These results -- including the optimal decomposition and the unitary extent -- are compiled in~\cref{tab:optimaldecomp}, where the gates are also defined.

\begin{table}[h]
\centering
\resizebox{\textwidth}{!}{  
\begin{tabular}{c|c|c|c}
\toprule
gate & unitary extent~$\xi$ & optimal decomposition & established \\ \midrule
2-qubit n.n. fermionic gate 
& \multirow{2}{*}{$1 + |\sin c|$} 
& $(R_Z(t_1) \otimes R_Z(t_2)) \cdot R_{XX}(a) \cdot R_{YY}(b) \, \cdot$ 
& \multirow{2}{*}{\cref{lem:extent2qubitgate}} \\ \cline{1-1}
2-qubit diagonal fermionic gate 
&  
& $\cdot \left( \cos(\frac{c}{2}) I - i \sin(\frac{c}{2}) Z \otimes Z \right) \cdot (R_Z(t_3) \otimes R_Z(t_4))$ 
&  \\ \hline 
$C(\theta) = \left(\begin{smallmatrix} 1 & & & \\ & 1 & & \\ & & 1 & \\ & & & e^{i\theta} \end{smallmatrix}\right)$ 
& $1+\left|\sin(\frac{\theta}{2})\right|$ 
& $e^{i\theta/4} \left( \cos(\frac{\theta}{4}) I + i \sin(\frac{\theta}{4}) Z \otimes Z \right) \left( R_Z(\frac{\theta}{2}) \otimes R_Z(\frac{\theta}{2}) \right)$ 
& \multirow{4}{*}{\cref{sec:decomp_C_RZZ_SWAP}} \\ \cline{1-3}
$R_{ZZ}(\theta) = e^{-i\frac{\theta}{2}ZZ}$ 
& $1+|\sin\theta|$ 
& $\cos(\frac{\theta}{2}) I - i \sin(\frac{\theta}{2}) Z \otimes Z$ 
&  \\ \cline{1-3}
n.n. $\SWAP = \left(\begin{smallmatrix}  1 & & & \\ & & 1 & \\ & 1 & & \\ & & & 1 \end{smallmatrix}\right)$
& $2$ 
& $\frac{e^{i\pi/4}}{\sqrt{2}}f\SWAP \left( I + i Z \otimes Z \right) ( R_{Z}(\frac{\pi}{2}) \otimes R_{Z}(\frac{\pi}{2}) )$ 
&  \\ \hline
$H = \frac{1}{\sqrt{2}} \left(\begin{smallmatrix} 1 & 1 \\ 1 & -1 \end{smallmatrix}\right)$
& $2$ 
& $\frac{1}{\sqrt{2}} (Z + X)$ 
&~\cref{sec:hadamard} \\ \hline
$R_Y(\theta) = e^{-i\frac{\theta}{2}Y}$ 
& \multirow{2}{*}{$1+|\sin\theta|$} 
& $\cos(\frac{\theta}{2}) I - i \sin(\frac{\theta}{2}) Y$ 
& \multirow{2}{*}{\cref{sec:RY}} \\ \cline{1-1}\cline{3-3}
$R_X(\theta) = e^{-i\frac{\theta}{2}X}$ 
&  
& $\cos(\frac{\theta}{2}) I - i \sin(\frac{\theta}{2}) X$ 
&  \\ \bottomrule
\end{tabular}}
\caption{The unitary extent and the respective optimal decomposition for any nearest neighbour (n.n.) or diagonal two-qubit fermionic gate -- in particular, the controlled phase~$C(\theta)$, the~$ZZ$-rotation~$R_{ZZ}(\theta)$ and the n.n.~$\SWAP$ gates -- and for the single-qubit Hadamard~$H$, $Y$-rotation~$R_{Y}(\theta)$ and $X$-rotation~$R_{X}(\theta)$ gates. Finding the optimal decomposition and extent for a general two-qubit fermionic gate requires solving a system of 16 equations to determine the parameters~$a,b,c,t_1, t_2, t_3,t_4 \in \bb{R}$.}
\label{tab:optimaldecomp}
\end{table}

The extent is multiplicative for the gates~$C(\theta)$,~$R_{ZZ}(\theta)$ and nearest neighbour~$\SWAP$ applied in parallel, i.e., 
\begin{align}
	\xi(C(\theta)^{\otimes m}) = \xi(C(\theta))^m
\end{align}
and the optimal decomposition of~$C(\theta)^{\otimes m}$ is the~$m$-fold tensor product of the optimal decomposition of~$C(\theta)$ given in~\cref{tab:optimaldecomp}, and similarly for~$R_{ZZ}(\theta)$ and the nearest-neighbour~$\SWAP$. This means that, when applying~$m$ copies of one such gate in parallel, there is no better strategy than to optimize the decomposition of each gate individually. 

All the non-Gaussian single-qubit and two-qubit gates we consider have rank two.
Notice that for the gates~$\SWAP$ and~$H$ the unitary extent coincides with the unitary rank, namely,
\begin{align}
	\chi(\SWAP) = \xi(\SWAP) = 2 = \chi(H) = \xi(H) \ . 
\end{align}
The same is true for the gate $C(\theta)$ with $\theta=\pi+ 2\pi j$ with $j\in\bb{Z}$, for two-qubit nearest neighbour or diagonal fermionic gates with $c = \pi/2 + \pi j,j\in\bb{Z}$ (see~\cref{tab:optimaldecomp}) and for the gates $R_{ZZ}(\theta)$, $R_{X}(\theta)$, $R_{Y}(\theta)$ with $\theta = \pi/2 + \pi j,j\in\bb{Z}$.
Consequently, in these cases simulating with a sparsified decomposition offers no runtime improvement over the exact decomposition, unless the circuit involves other gates with extent smaller than two.

Since the set of free unitaries is continuous in the Gaussian setting, finding optimal decompositions numerically is more challenging. Hence our focus on proving optimal decompositions analytically. This avoids having to discretize the set of Gaussian unitaries, which is required for finding optimal decompositions via convex optimization. To this end, we establish the so-called lifting lemma for Gaussian unitaries -- see~\cref{lem:lifting} -- analogous to the lifting lemma~\cite[Lemma 1]{Bravyi_2019} for stabilizer-based computations. This lemma allows us to derive the optimal decompositions presented in~\cref{tab:optimaldecomp} for two-qubit gates using the known decomposition of the corresponding magic states, previously derived in~\cite{reardonsmith2024improved,cudby2024gaussiandecompositionmagicstates,reardonsmith2024fermioniclinearopticalextent}.

We note that, while the two-qubit gates~$C(\theta)$,~$R_{ZZ}(\theta)$ and the~$\SWAP$ are magic gates for Gaussian computations, the non-Gaussian single-qubit gates~$H$, $R_{Y}(\theta)$ and~$R_{X}(\theta)$ are often not considered magic gates because they are not fermionic in the sense that they do not preserve parity (see the related discussion in~\cref{sec:intro_magic_gates}). Nevertheless, they may find application in settings such as the simulation of spin systems that can be mapped to a near-Gaussian model.

\subsubsection*{Result 2: Optimal (and improved) decompositions of non-Gaussian noisy channels into (possibly adaptive) Gaussian channels}

We give improved/optimal decompositions for rotation gates subject to Pauli noise. 
Our results demonstrate that 1) stochastic Pauli noise and 2) allowing for adaptive occupation number measurements in the decomposition lowers the cost of classically simulating rotation channels. 
 
The stochastic Pauli channel~$\cal{E}_{\cal{P}} = (1-p) \cal{I} + p \cal{P}$  leaves the state invariant with probability~$1-p$ and applies a Pauli error~$P$ with probability~$p$.
The rotation channel~$\cal{R_P}(\theta) (\cdot) = R_P(\theta) (\cdot) R_P(\theta)^\dagger$, with~$R_P(\theta) = \exp(-i \theta P / 2 )$ the rotation around the~$P$-axis, subject to Pauli noise gives the noisy rotation channel 
\begin{align}
	\label{eq:decomp0}
	\cal{N_{P,P'}} = \cal{E}_{\cal{P}} \circ \cal{R_P}(\theta) =  (1-p) \cal{R}_P(\theta) + p \cal{P'} \circ \cal{R}_P(\theta) \ .
\end{align}

In~\cref{sec:noisydecomps}, when $\cal{P} = \cal{P'}$ we show the channel~$\cal{N_{P,P}}$ has improved decomposition (compared to~\cref{eq:decomp0})
\begin{align}
	\label{eq:decomp1}
	\cal{N_{P,P}} &= s \cal{R}_P(\varphi) + (1-s) \cal{R}_P(\pi-\varphi) \ ,
\end{align}
where~$\varphi = \arcsin\left[ (1-2p) \sin\theta \right]$ and~$s = \frac{1}{2} \left( 1 + (1-2p) \cos\theta / \cos\varphi \right)$, with respect to the channel extent. 
In~\cref{sec:optimal_decomp_NY_NZZ}, we further show that for 
\begin{enumerate}[1)]
	\item the channel~$\cal{N_Y}$ which is a~$Y$-rotation subject to~$Y$-noise (i.e.~$P=P'=Y$),
	\item the channel~$\cal{N_{ZZ}}$ which is~$ZZ$-rotation subject to~$ZZ$-noise (i.e.~$P=P'=ZZ$),
\end{enumerate}
the decomposition in~\cref{eq:decomp1} is optimal with respect to the channel extent which is
\begin{align}
	\Xi(\cal{N_Y}) = \Xi(\cal{N_{ZZ}}) = 1 + (1-2p) |\sin\theta| \ .
\end{align}

In addition, we define a generalization of the channel extent~$\Xi$ called augmented channel extent~$\Xiaug$ -- see~\cref{def:channelextentaugmented} -- which allows decomposing in terms of  both unitary and adaptive Gaussian channels. For the channel~$\cal{N}'_\cal{ZZ} = \cal{N_{ZZ,Z}}$ that consists of a~$ZZ$-rotation subject to single qubit~$Z$-noise, allowing for adaptive Gaussian channels admits the decomposition 
\begin{align}
	\label{eq:decomp3}
	\cal{N'_{ZZ}} = (1-2p) \cal{R_{ZZ}}(\theta) + 2 p \cal{E} \ ,
\end{align}
where~$\cal{E} (\cdot) = K_0 (\cdot) K_0^\dagger + K_1 (\cdot) K_1^\dagger~$ is a convex-Gaussian channel with~$K_0 = \Pi_0 \otimes R_{Z}(\theta)$,~$K_1 = \Pi_1 \otimes R_{Z}(-\theta)$ where~$\Pi_0=\ketbra{0}$ and~$\Pi_1=\ketbra{1}$, see~\cref{sec:improved_decomp_NZZ}. The decomposition in~\cref{eq:decomp3} achieves~$\Xiaug(\cal{N}'_\cal{ZZ}) \leq 1 + (1-2p) |\sin\theta|$, an improvement over the upper bound~$\Xi(\cal{N}'_\cal{ZZ}) \leq 1 + |\sin\theta|$ for the channel extent.

The optimal/improved decompositions for rotation channels subject to~$Z_1$,~$Z_2$ or~$Z_1 Z_2$ depolarising noise can be combined to obtain improved decompositions when the noise channel is a combination of these components i.e., of~$Z_1$,~$Z_2$ and~$Z_1 Z_2$. As an example, in~\cref{sec:decompgeneraldephasing} we give an improved decomposition of the~$\cal{R}_{\cal{Z}\cal{Z}}(\theta)$ subject to the noise channel~$\cal{E} = (1-p)\id + p/3 (\cal{Z}_1 + \cal{Z}_2 + \cal{Z}_1 \cal{Z}_2)$.

 
\subsubsection*{Result 3: Optimal and improved decompositions for the fermionic non-linearity}

In~\cref{lem:FNLoptimal} we show that an optimal decomposition of a gate~$U$ with respect to the unitary extent gives an optimal decomposition of the channel~$\cal{U}(\cdot) = U (\cdot) U^\dagger$ with respect to the fermionic non-linearity and 
\begin{align}
	\Phi(\cal{U}) = \xi(U) \ .
\end{align}
Hence, the optimal unitary extent decompositions given in~\cref{tab:optimaldecomp} for the gates~$C(\theta),R_{ZZ}(\theta)$ and~$\SWAP$ provide optimal decompositions for the respective channels~$\cal{C}(\theta),\cal{R_{ZZ}}(\theta)$ and~$\cal{SWAP}$ with respect to the fermionic nonlinearity, and the fermionic non-linearity of the latter is equal to the unitary extent of the respective gates, given in~\cref{tab:optimaldecomp}. 
This result improves on the decompositions established in~\cite{Hakkaku_2022} for the noiseless channel~$\cal{R}_{\cal{ZZ}}(\theta)$ and, in some parameter regimes (large $\theta$, or small $\theta$ given sufficiently weak noise), for the noisy channel~$\cal{E}(p) \circ \cal{R}_{\cal{ZZ}}(\theta)$ with~$\cal{E}(p) = (1-p)\id + p/3 (\cal{Z}_1 + \cal{Z}_2 + \cal{Z}_1 \cal{Z}_2)$.

\subsubsection*{Result 4: Ensemble sampling lemma for adaptive circuits}

We give a procedure for sparsifying circuits involving intermediate non-unitary elements. In particular, we consider circuits of the form $\chan = \mathcal{U}_T \circ \mathcal{A}_T \circ \ldots \mathcal{U}_1 \circ \mathcal{A}_1 \circ \mathcal{U}_0 $ comprising interleaved layers of Gaussian adaptive channels $\mathcal{A}_t$ and non-Gaussian gates $\mathcal{U}_t$ with known decomposition. We show that given a Gaussian initial state $\ket{g}$, we
can generate a sparsification $\Omega$  drawn at random from an ensemble $\rho_1 = \sum_\Omega \Pr(\Omega) \Omega/\Tr[\Omega]$ that is close in trace-norm to the exact final state $\rho = \chan(\op{g})$. We will see that the exact final state can be expressed
\begin{equation}
\rho = \chan(\op{g}) = \sum_{\arr{m}} \op{\psi_{\arr{m}}}, \quad \ket{\psi_{\arr{m}}}  =  \sum_v c_v \ket*{\phi_{\arr{m},v}}
\end{equation}
where $\ket{\psi_{\arr{m}}}$ are the non-normalized non-Gaussian states obtained by applying some sequence of Kraus operators labelled by $\arr{m}$, and $\ket*{\phi_{\arr{m},v}}$ are non-normalized Gaussian states. The complex coefficients $c_v$ derive only from the known decompositions of the unitary gates $\mathcal{U}_j$ and do not depend on the adaptive trajectory $\arr{m}$. We then give approximation guarantees by proving the following lemmas.
\begin{lemma}[Ensemble sampling lemma for adaptive circuits]\label{lem:ensembleAdapt}
The procedure defined in~\cref{it:sparsification1,it:sparsification2,it:sparsification3} in~\cref{sec:SPARSE_adapt}, which has runtime $O(kT)$, returns a description of a sparsified operator $\Omega =\sum_{\arr{m}} \op{\Omega_{\arr{m}}}$, with probability $Pr(\Omega)$, where each $\ket{\Omega_{\arr{m}}}$ is a non-normalized superposition over $k$ terms, such that the expected normalized state defined by
\begin{equation}
\rho_1 = \mathbb{E}\left( \frac{\Omega}{\Tr[\Omega]} \right) = \sum_\Omega \Pr(\Omega) \frac{\Omega}{\Tr[\Omega]}
\end{equation}
satisfies
\begin{equation}
\norm{\rho_1 - \chan(\op{\phi_0})}_1 \leq  2  \frac{\norm{c}_1^2}{k} + \sqrt{\Var(\Tr[\Omega])}\ . \label{eq:ensembleresult}
\end{equation}
where $\norm{c}_1 = \sum_v \abs{c_v}$.
\end{lemma}

\begin{lemma}[Sparsification variance bound]\label{lem:varianceAdapt}
The variance in the normalization of the randomly drawn operator $\Omega$ from~\cref{lem:ensembleAdapt} can be bounded as
\begin{equation}
\Var(\Tr \Omega ) \leq \frac{4(\cte - 1)}{k} + 2 \left( \frac{\norm{c}_1^2}{k} \right)^2 + \frac{10 - 12\cte}{k^2} +  O\left( \frac{\cte}{k^3}\right),
\end{equation}
where
\begin{equation}
\cte =  \norm{c}_1\sum_{\arr{w}} \abs{c_{\arr{w}}} \cdot \left| \sumvec{m}   \bra*{\psim}\ket*{\phimv{m}{w}}\right|^2 \ .
\end{equation}
\end{lemma}
Notice that normalizing $\cte$ by $\norm{c}^2_1$ gives the average over $\arr{w}$ of the overlap
\begin{equation}\left| \sumvec{m}   \bra*{\psim}\ket*{\phimv{m}{w}}\right|^2 \leq  \left( \sumvec{m}    \left|\bra*{\psim}\ket*{\phimv{m}{w}}\right|\right)^2 \end{equation} so can loosely be thought of as the average overlap of the final state with each of its Gaussian terms. We can expect this quantity 
to be small for non-trivial circuits containing significant magic. The novelty of our result compared to previous sparsification routines established in~\cite[Lemma 6]{Bravyi_2019}, \cite[Theorem 16]{Seddon_2021} and~\cite{seddonThesis} is the incorporation of intermediate adaptive channels, whereas the prior works dealt with sparsifying either initial magic states~\cite{Bravyi_2019,Seddon_2021} or convex-unitary circuits~\cite{seddonThesis}.
It thus allows adapting existing stabilizer-based simulation techniques~\cite{Bravyi_2019,Seddon_2021} to adaptive circuits. These algorithms leverage the optimal and improved decompositions of non-Gaussian gates and channels we propose here.

\section{Background \label{sec:background}}

In~\cref{sec:gaussianoperations} we discuss Gaussian computations. We define Gaussian states and operations. We briefly discuss the efficient classical simulability of Gaussian computations and a recent phase-sensitive simulator for Gaussian computations. \cref{sec:intro_magic_gates} introduces magic gates and magic states that augmented Gaussian computations.

Before proceeding, we establish some notation used throughout this work:

\begin{enumerate}[1)]
	\item We define~$[j] := \{1, \ldots, j\}$,~$j\in\bb{N}$.
	\item We denote the single-qubit Pauli gates by
	\begin{align}
		X = 
		\begin{pmatrix}
		0 & 1 \\
		1 & 0	
		\end{pmatrix} \ , \qquad
		Y = 
		\begin{pmatrix}
		0 & -i \\
		i & 0	
		\end{pmatrix} \ , \qquad
		Z = 
		\begin{pmatrix}
		1 & 0 \\
		0 & -1	
		\end{pmatrix} \ .
	\end{align}
	\item Let~$U$ be a single-qubit unitary operator. For~$j\in[n]$, we define the~$n$-qubit unitary 
	\begin{align}
		U_j = I^{\otimes(j-1)} \otimes U \otimes I^{\otimes(n-j)} \ ,
	\end{align}
	where~$U$ acts on the~$j$th qubit and the identity operator~$I$ acts on all other qubits.
	\item Let~$U$ be a two-qubit unitary. For~$j,k\in [n]$, we define the~$n$-qubit unitary~$U_{j,k}$ as the operator which acts with~$U$ on qubits~$j$ and~$k$, and trivially (with the identity operator) elsewhere. 
	\item We denote an~$n$-qubit unitary channel by~$\cal{U}(\cdot) = U (\cdot) U^\dagger$ where~$U$ is an~$n$-qubit unitary. A unitary channel~$\cal{U}$ can be represented by its transfer matrix~$U \otimes \overline{U}$. 
	\item We use~$\ket{\psi^+} = (\ket{00} + \ket{11})/\sqrt{2}$.
	\item We denote by~$|s|$ the number of elements in the set~$s$.
	\item The Pauli rotation gate around the axis defined by the Pauli string~$P$ by an angle~$\theta$ is~$R_P(\theta) = \exp(-i \theta P /2)$. For two-qubit Pauli strings $P\otimes P'$, we use the notation~$R_{PP'}(\theta) = \exp(-i \theta P \otimes P' /2)$.
\end{enumerate}

\subsection{Gaussian computations \label{sec:gaussianoperations}} 

Here we give a brief characterization of Gaussian states, unitaries and measurements. For a more thorough description of Gaussian computations see e.g.,~\cite{bravyiLagrangianRepresentationFermionic2004,knill2001fermionic,PhysRevA.65.032325,10.1145/380752.380785}.
We consider a~$n$-qubit system with Hilbert space~$\cal{H}_n \cong (\bb{C}^2)^{\otimes n}$.
We define the so-called Jordan-Wigner operators as~$c_1 =  X \otimes \id^{\otimes(n-1)}$,~$c_2 = Y \otimes \id^{\otimes(n-1)}$ and
\begin{align}
	\label{eq:JW1}
	c_{2j-1} = Z^{\otimes (j-1)} \otimes X \otimes \id^{\otimes(n-j)} \\
	\label{eq:JW2} c_{2j} = Z^{\otimes (j-1)} \otimes Y \otimes \id^{\otimes(n-j)}
\end{align}
for~$j\in\{2,\ldots,n\}$. These are related to the so-called Jordan-Wigner mapping~$\gamma_{2j-1} \rightarrow c_{2j-1}$ and~$\gamma_{2j} \rightarrow c_{2j}$ between Majorana operators~$\{\gamma_j\}_{j=1}^{2n}$ which act on fermionic modes~\cite{jordan_uber_1928} (see also e.g.,~\cite{doi:10.1098/rspa.2008.0189}) and Pauli operators~$\{c_j\}_{j=1}^{2n}$ which act on qubits. Throughout this work we will work in the qubit picture. The results carry over to the Majorana picture by applying the Jordan Wigner transformation. 

A (pure)~$n$-qubit state~$\ket{\phi} \in \cal{H}_n$ is a fermionic state if it is an eigenstate of the parity operator~$Z^{\otimes n}$. Eigenstates of~$Z^{\otimes n}$ with eigenvalue~$+1$ ($-1$) are said to have even (odd) parity. Similarly, a unitary~$V$ is a fermionic unitary if it preserves parity, i.e., if~$[V, Z^{\otimes n}] = 0$.

\subsubsection*{Gaussian unitaries}

A two-qubit matchgate~\cite{10.1145/380752.380785} is a unitary gate defined as 
\begin{align}
	\label{eq:matchgate2}
	G(A,B) :=
	\begin{pmatrix}
		A_{1,1} & 0 & 0 & A_{1,2} \\
		0 & B_{1,1} & B_{1,2} & 0 \\
		0 & B_{2,1} & B_{2,2} & 0 \\
		A_{2,1} & 0 & 0 & A_{2,2} \\
	\end{pmatrix}
	\text{ where } A,B \in U(2) \text{ and } \det A = \det B \ . 
\end{align}
We define a~$n$-qubit nearest neighbour (n.n.) matchgate as~$G(A,B)_{j,j+1}$ for~$j\in[j-1]$.
We denote by~$G_n$ the set of~$n$-qubit unitaries generated by n.n. matchgates and the gates~$X_j$ and~$Y_j$ for~$j\in\bb{N}$. Concretely,~$G_n$ is generated by
\begin{align}
	\gen_1(G_n) = 
	\{ G(A,B)_{j,j+1} : A,B \in U(2) , \det A = \det B, j\in[n-1] \} 
	\cup \{ X_j, Y_j : j \in [n]\} \ .
\end{align}
Oftentimes it is useful to consider the alternative set of generators
\begin{align}
	\label{eq:gen2}
  \gen_2(G_n) = \{ U_{j,k}(\theta) := \exp(\tfrac{\theta}{2} c_j c_k) : j , k \in [2n] ; \theta \in [0,2\pi) \} \cup \{ U_j := c_j : j\in [2n] \} \ .
\end{align}
We call an element of~$G_n$ a ($n$-qubit) Gaussian unitary or circuit. 
A unitary gate~$U_R$ is Gaussian if and only if (see e.g.,~\cite{bravyiLagrangianRepresentationFermionic2004,bravyiComplexityQuantumImpurity2017a})
\begin{align}
	\label{eq:UR_R_relation}
	U_R c_j U_R^\dagger = \sum_{j=1}^{2n} R_{j,k} c_k \text{ where } R \in O(2n)
\end{align}
for all~$j\in[2n]$ and~$U_R \in G_n$. 
Note that our definition of Gaussian unitaries extends the most standard definition, which typically includes only the two-qubit matchgates defined in~\cref{eq:matchgate2} (and the circuits they generate), excluding~$X_j$ and~$Y_j$.
This extension still renders matchgate computations efficiently classically simulable, as we show later in this section. 

It will be useful to keep in mind that under the Jordan-Wigner mapping, quadratic majorana generators $c_j c_k$ either correspond to pairs of $X$ or $Y$ operators on distinct qubits, connected by an unbroken $Z$-string if separated in the Jordan-Wigner ordering, or, if $j$ is odd and $k=j+1$, to the single-qubit operator $Z_{k/2}$. Hence the single-qubit Pauli rotation $R_{Z_j}(\theta)$ is always a free operation in our setting, as are nearest-neighbour two-qubit rotations $R_{PQ}(\theta)$ for $P$ and $Q$ in $\{X,Y\}$.

\subsubsection*{Gaussian states}

A pure state~$\ket{\psi} \in \cal{H}_n$ is an~$n$-qubit (pure fermionic) Gaussian state if it can be obtained by the action of a Gaussian circuit on the all-zero state~$\ket{0}^{\otimes n}$. We denote the set of~$n$-qubit (pure) fermionic Gaussian states by 
\begin{align}
	\cal{G}_n = \{ U \ket{0}^{\otimes n} : U \in G_n \} \ .
\end{align}
A Gaussian state~$\ketbra{\psi}$ is described by its covariance matrix, a~$2n\times 2n$ antisymmetric real matrix~$\Gamma(\psi)$ with entries
\begin{align}
	\label{eq:covmatrix}
	\Gamma(\psi)_{j,k} = \langle \psi, i c_jc_k\psi \rangle \qquad \text{ for } j,k \in[2n] \ .
\end{align}
Expectation values of higher order products of Jordan-Wigner operators can be obtained as a function of the covariance matrix~$\Gamma$ using Wick's theorem~\cite{PhysRev.80.268}. 

It follows from the definition of a Gaussian state that matchgate circuits map (convex mixtures of) Gaussian states to (convex mixtures of) Gaussian states. A Gaussian unitary channel is a channel~$\cal{U}(\cdot) = U(\cdot)U^\dagger$ where~$U\in G_n$ is a Gaussian circuit.
We work with the following definition of a convex-Gaussian channel.

\begin{definition}[Convex-Gaussian channels]
	A convex-Gaussian channel is a completely positive trace-preserving map which maps convex mixtures of pure Gaussian states to convex mixtures of pure Gaussian states.
\end{definition}
A Gaussian channel is a channel which maps Gaussian states to Gaussian states, for more details see e.g.,~\cite{doi:10.1142/9781860948169_0002}. 
Notice that convex mixtures of Gaussian states are not necessarily Gaussian states~\cite{Melo_2013,vershyninaCompleteCriterionConvexGaussianstate2014}, and not all convex-Gaussian channels are Gaussian channels. 
\subsubsection*{Gaussian measurements}

Gaussian measurements are defined as measurements which map a Gaussian state onto a Gaussian state on the remaining qubits (i.e., the qubits that are not measured). 

The most commonly considered Gaussian measurements are computational basis measurements, or equivalently~$Z$-measurements. A computational basis measurement of the~$j$th qubit has possible outcomes~$\{0,1\}$ and POVM given by~$\{\Pi_k(0),\Pi_k(1)\}$ with 
\begin{align}
	\Pi_k(0) = I^{\otimes (k-1)} \otimes \ketbra{0} \otimes I^{\otimes (n-k)}
	\quad\text{and}\quad
	\Pi_k(1) = I^{\otimes (k-1)} \otimes  \ketbra{1} \otimes I^{\otimes (n-k)} \ .	
\end{align}

\subsubsection*{Classical simulability of Gaussian computations \label{sec:classical_simulation_matchgates}}

Gaussian computations --~$\poly(n)$ size Gaussian circuits acting on~$n$-qubit Gaussian states followed by computational basis measurements -- can be simulated classically in an efficient way (i.e., in polynomial time and space) using the covariance matrix formalism which we describe here~\cite{knill2001fermionic,PhysRevA.65.032325}. 
This is because Gaussian states and operations admit a description in space~$\poly(n)$, and Gaussian operations map Gaussian states to Gaussian states in a tractable way: there are~$\poly(n)$ time rules to update the covariance of a Gaussian state subject to a Gaussian operation. 

Gaussian states and operations admit the following~$\poly(n)$ classical descriptions:
\begin{enumerate}[1)]
	\item A Gaussian state~$\ket{\psi}$ is specified by its covariance matrix, see~\cref{eq:covmatrix}.
	\item A Gaussian unitary~$U_R$ is specified by a~$2n \times 2n$ orthogonal matrix~$R\in O(2n)$, see~\cref{eq:UR_R_relation}.
	\item A single-qubit computational basis measurement is described by a tuple~$(j,m) \in [n] \times \{0,1\}$ where the integer~$j$ identifies the measured mode, and~$m$ is the measurement outcome.
\end{enumerate}

Consider a Gaussian state~$\ket{\psi} \in \cal{G}_n$ with covariance matrix~$\Gamma(\psi)$. The state~$U_R \ket{\psi}$ obtained by evolution under a Gaussian unitary~$U_R$ has covariance matrix
\begin{align}
	\label{eq:cov_evolved_psi}
	\Gamma(U_R \psi) = R \Gamma(\psi) R^T \ ,
\end{align}
where~$R$ is the~$2n\times 2n$ orthogonal matrix that specifies~$U_R$ by~\cref{eq:UR_R_relation}. \cref{eq:cov_evolved_psi} takes time~$O(n^3)$ to compute.

Let~$\ket{\psi}$ be a Gaussian state. The probability of obtaining outcome~$m\in\{0,1\}$ when performing a computational basis measurement of the~$k$th qubit is
\begin{align}
	\label{eq:meas_prob}
	p(m) = \frac{1}{2}\left( 1 + (-1)^m \Gamma_{2j-1,2j} \right) \ , 
\end{align}
and the post-measurement state
\begin{align}
	\ket{\psi(m)} = \frac{1}{\sqrt{p(m)}} \Pi(m) \ket{\psi}
\end{align}
is a Gaussian state whose covariance matrix is the antisymmetric matrix~$\Gamma(\psi(m))$ with entries~\cite{bravyi2012disorder}
\begin{align}
	\label{eq:cov_mat_post_measure}
	\Gamma(\psi(m))_{k,\ell} =  
	\begin{cases}
		0 & \hspace{-2cm} \text{if } k = \ell \\
		(-1)^m & \hspace{-2cm} \text{if }(k,\ell) = (2j-1,2j) \\
		-(-1)^m & \hspace{-2cm} \text{if }(k,\ell) = (2j,2j-1) \\
		\frac{1+(-1)^{\delta_{\lfloor k/2 \rfloor, j}+\delta_{\lfloor \ell/2 \rfloor, j}}}{2}
		\left(\Gamma_{k,\ell} + \frac{(-1)^m}{2 p(m)} \left( \Gamma_{2j-1,\ell} \Gamma_{2j,k} - \Gamma_{2j-1,k} \Gamma_{2j,\ell} \right)\right) & \text{otherwise} 
	\end{cases} \ .
\end{align}
The probability~$p(m)$ in~\cref{eq:meas_prob} takes time~$O(1)$ to compute, and the covariance matrix~$\Gamma(\psi(m))$ in~\cref{eq:cov_mat_post_measure} takes time~$O(n^2)$.
We collect the runtimes for the different tasks in~\cref{tab:FLOsimulation_runtimes}.

We note that including intermediate computational basis measurements on a subset of qubits and allowing the choice of subsequent gates conditioned on measurement outcomes still results in a Gaussian computation that is efficiently classically simulable.

\begin{table}[b]
\centering
\begin{tabular}{@{}llcc@{}}
\toprule
\multicolumn{2}{c}{operation} & \begin{tabular}[c]{@{}c@{}}runtime \\ (without phase)\end{tabular} & \begin{tabular}[c]{@{}c@{}}runtime \\ (with phase)\end{tabular} \\ \midrule
\multicolumn{2}{l}{ evolution by Gaussian unitary~$U_R\in\gen_2(G_n)$} & $O(n^3)$ & $O(n^3)$ \\ \midrule
 \hspace{0.9mm} single-qubit computational & measurement probability & $O(1)$ & $O(1)$ \\ \cmidrule(l){2-4} 
 \hspace{0.9mm} basis measurement & postmeasurement state & $O(n^2)$ & $O(n^3)$ \\ \bottomrule
\end{tabular}
\caption{Time complexity of classical algorithms for Gaussian unitary evolution and computational basis measurements of a Gaussian state. We consider unitary evolution by one of the generators of Gaussian dynamics in~$\gen_2(G_n)$. The task of measuring is split into two: computing the measurement probability (done similarly for phase-sensitive and phase-agnostic computations), and computing the post-measurement state. The column ``without phase'' gives the runtimes for phase-agnostic simulation, i.e., the time it takes to evaluate~\cref{eq:meas_prob,eq:cov_evolved_psi,eq:cov_mat_post_measure}. The column ``with phase'' gives the runtimes for phase-sensitive simulation: see~\cite{Dias_2024,reardonsmith2024improved}.
\label{tab:FLOsimulation_runtimes}}
\end{table}

\subsubsection*{Phase sensitive simulator for Gaussian computations}

The simulator given in the previous section is not phase-sensitive: it describes the system up to a global phase. For example, the covariance matrix~$\Gamma(\psi)$ specifies the density operator~$\ketbra{\psi}$, not the state vector~$\ket{\psi}$.

In~\cite{Dias_2024,reardonsmith2024improved} the authors introduce an augmented covariance matrix formalism that keeps track of the phase of pure states in Gaussian computations. 
We refer to these works for details, and simply highlight that each of the tasks
\begin{enumerate}[1)]
	\item evolution of a Gaussian state by a Gaussian unitary in~$\gen_2(G_n)$ (see~\cref{eq:gen2}),
	\item a computational basis measurement of a single-qubit
\end{enumerate}
takes time~$O(n^3)$.
We note that the work in~\cite{Dias_2024} formulates Gaussian computations for fermions and not qubits: to translate~\cite{Dias_2024} to our setting it suffices to take the Jordan-Wigner transformation (note the similarity with~\cref{eq:JW1,eq:JW2})
\begin{align}
	\gamma_{2j-1} \rightarrow Z^{\otimes (j-1)} \otimes X \otimes \id^{\otimes(n-j)} \\
	\gamma_{2j} \rightarrow Z^{\otimes (j-1)} \otimes Y \otimes \id^{\otimes(n-j)}
\end{align}
for~$j\in[n]$ from Majorana operators~$\{\gamma_j\}_{j=1}^{2n}$ to Pauli operators.
In~\cref{tab:FLOsimulation_runtimes} we summarize the runtime for the relevant tasks in Gaussian computations when taking into account phases.

\subsection{Magic gates and magic states \label{sec:intro_magic_gates}}

Gaussian computations are not universal. Magic gates are defined as gates that, when combined with the set~$G_n$ of Gaussian circuits, form a universal set of gates, which can be used to perform universal quantum computation. 
Examples of magic gates for Gaussian computations are the controlled-phase~$C(\theta)$, the~$ZZ$-rotation~$R_{ZZ}(\theta)$ and the~$\SWAP$ gates (see~\cref{tab:optimaldecomp} for definitions of these gates). The~$\SWAP$ gate is not to be confused with the matchgate 
\begin{align}
	\label{eq:fswap}
		f\SWAP = 
		\begin{pmatrix}
			1 & 0 & 0 & 0 \\
			0 & 0 & 1 & 0 \\
			0 & 1 & 0 & 0 \\
			0 & 0 & 0 & -1 
		\end{pmatrix} 
\end{align}
which performs a swap of the two-qubits but, unlike the~$\SWAP$ gate, adds a~$-1$ phase to the state~$\ket{11}$. 

In Ref.~\cite{PhysRevA.84.022310} the authors give the following decomposition of any two-qubit fermionic (i.e., parity preserving) gate using single- and two-qubit rotation gates. (This will be useful to prove~\cref{lem:extent2qubitgate}.)
\begin{lemma}[Lemma III.1 in Ref.~\cite{PhysRevA.84.022310}]
\label{lem:fermionicgatedecomposition}
Any two-qubit fermionic gate~$U$ can be written as
\begin{align}
	\label{eq:gateUdecomp}
	U = (R_Z(t_1) \otimes R_Z(t_2)) \cdot R_{XX}(a) \cdot R_{YY}(b) \cdot R_{ZZ}(c) \cdot  (R_Z(t_3) \otimes R_Z(t_4)) \ ,
\end{align}
for~$t_1,t_2,t_3,t_4,a,b,c\in\bb{R}$.
\end{lemma}
Notice that the decomposition \eqref{eq:gateUdecomp} of~$U$ involves only one gate that is not Gaussian: The rotation~$R_{ZZ}(c)$ around the~$ZZ$-axis. All remaining gates in the decomposition are Gaussian. 

An alternative approach to augmenting Gaussian computations is to include ancilla qubits initialized in so-called magic states -- this procedure is commonly referred to as (magic) state injection. A magic state for a magic gate~$M$ is a state which can be injected anywhere in the circuit and that, combined with a Gaussian circuit that may use adaptive Gaussian measurements, implements the magic gate~$M$. Such a construction is called a gadget. 

In~\cite{Hebenstreit_2019} the authors establish that all pure fermionic non-Gaussian states are magic states for Gaussian computations. Magic states are required to be fermionic so that they can be injected at any point in the circuit. Precisely, fermionic states can be swapped through arbitrary states using Gaussian unitaries only (see~\cite[Lemma 1]{Hebenstreit_2019}). Analogously, magic gates are typically required to be fermionic gates (i.e., gates which map fermionic states to fermionic states) so that they can be swapped through any unitary using Gaussian unitaries only. 

For~$\theta \in (0, 2\pi)$, the Choi–Jamiołkowski state~\cite{Jamiolkowski1972,Choi1975}
\begin{align}
	\ket{m(\theta)} = ( I\otimes C(\theta) \otimes I ) \ket{\psi_+}^{\otimes 2}
\end{align}
corresponding to the controlled phase gate~$C(\theta)$ is a magic state for Gaussian computations. Indeed, the state~$\ket{m(\theta)}$ can be used in the gadget given in~\cref{fig:gadget} to implement the gate~$C(\theta)$. Notice that the gadget in~\cref{fig:gadget} used postselection. Instead, one can use the gadget proposed in~\cite{Hebenstreit_2019} (see also~\cite{bampounis2024matchgatehierarchycliffordlikehierarchy}) which includes a correction circuit composed of Gaussian operations alone, avoiding postselection.
In~\cite{Hebenstreit_2019} the authors show that any 4-qubit fermionic state is Gaussian equivalent to~$\ket{m(\theta)}$. 

The lowest dimensional magic states are 4-qubit fermionic states. This is because fermionic states of up to three qubits are Gaussian~\cite{Melo_2013}. There exist states that are non-Gaussian on fewer than 4 qubits, but these are not fermionic in the sense that they are states without fixed parity. Indeed, it has been shown that the action of a Gaussian circuit on any $t$-qubit tensor product of arbitrary single-qubit states can be classically simulated with constant overhead \cite{Brod2016}, and in \cref{sec:extentplus} we will consider the implications for Gaussian rank and extent.  In contrast, in the fixed-parity setting, the state extent was recently shown to be multiplicative for tensor products of fermionic magic states of up to four qubits. 
\begin{lemma}[see~\cite{cudby2024gaussiandecompositionmagicstates,reardonsmith2024fermioniclinearopticalextent}]
	The state extent of tensor product states is multiplicative, i.e., 
	\begin{align}
		\xi\left(\bigotimes_{j=1}^m \ket{\psi_j}\right) = \prod_{j=1}^m \xi\left(\psi_j\right) \ , 
	\end{align}
	if each factor~$\ket{\psi_j}$ is a fermionic state of at most four qubits.
	\label{lem:extent_multiplicative}
\end{lemma}
Finally we note that a recent work has shown that observable estimation for a restricted class of Gaussian circuit acting on a structured type of fermionic magic state input can be efficiently simulated up to additive error \cite{oh2026classicalsimulationfreefermionicdynamics}. The specific restrictions the authors place on the circuit in order to preserve the structure of the input states preclude the type of magic state injection described in \cite{Hebenstreit_2019}, so this result is not in tension with the statement that the same magic states are a resource for universal fermionic computation when general Gaussian operations are allowed.

\section{Magic monotones \label{sec:magicmonotones}}

We define several magic monotones for states, gates and channels which are relevant for Gaussian-based simulation. The definitions of rank, pure state and unitary extent, and channel extent presented here are the Gaussian analogues of previously defined quantities respectively in~\cite{PhysRevX.6.021043,Bravyi_2019,seddonThesis} in the context of stabilizer-based simulations. Additionally, we introduce a new monotone called augmented channel extent which is a generalization of the channel extent. Lastly, we define the Gaussian fidelity which is the Gaussian analogue of the stabilizer fidelity introduced in~\cite{Bravyi_2019}. 

\begin{definition}[Rank, analogous to the definition of stabilizer rank in~\cite{PhysRevX.6.021043}]
	The rank~$\chi(\psi)$ of a~$n$-qubit state~$\psi \in \cal{H}_n$ is defined as
	\begin{alignat}{2}
		\chi(\ket{\psi}) 
		&= \min \left\{ \chi : \ket{\psi} = \sum_{j=1}^\chi c_j \ket{\phi_j} ,\ket{\phi_j}\in \cal{G}_n , c_j \in \bb{C}, \chi \in \bb{Z} \right\} \ .
	\end{alignat}
\end{definition}

See~\cite{Bravyi_2019,Heimendahl2021stabilizerextentis} respectively for the primal and the dual formulation of the stabilizer extent, which translate directly to the Gaussian analogue given in~\cref{def:extent} (see e.g.,~\cite[Chapter 5]{Boyd_Vandenberghe_2004} for the notion of dual program in convex optimization).

\begin{definition}[Pure state extent, analogous to the Definition 3 of stabilizer extent in~\cite{Bravyi_2019}]
\label{def:extent}
The (pure) state extent~$\xi(\psi)$ of a~$n$-qubit state~$\psi \in \cal{H}_n$ is defined as
\begin{alignat}{2}
	\xi(\ket{\psi}) 
	&= \min \left\{\|c\|_1^2: \ket{\psi} = \sum_{j} c_j \ket{\phi_j} ,\ket{\phi_j}\in \cal{G}_n , c_j \in \bb{C}\right\} 
	\quad&&\text{(primal form)}\ , \\
	\label{eq:extent-state-dual}
	&= \max \left\{ |\langle \omega, \psi \rangle|^2 : \ket{\omega} \in (\bb{C}^2)^{\otimes n} , |\langle{\omega}|{\phi}\rangle| \leq 1 \text{ for all } \ket{\phi} \in \cal{G}_n  \right\}
	\quad&&\text{(dual form)}\ .
\end{alignat}
\end{definition}

A vector~$\ket{\omega} \in (\bb{C}^2)^{\otimes n}$ which obeys~$|\langle{\omega} | {\phi} \rangle| \leq 1$ for all~$\ket{\phi} \in \cal{G}_n~$ is called a (valid) state extent witness.
A decomposition~$\ket{\psi} = \sum_{j} c_j \ket{\phi_j}$ where~$\ket{\phi_j}\in \cal{G}_n , c_j \in \bb{C}$ is called a feasible decomposition of the state~$\ket{\psi}$.
A feasible decomposition that achieves~$\xi(\ket{\psi}) = \|c\|_1^2$ is called an optimal decomposition of the state~$\ket{\psi}$ (with respect to the pure state extent).

\begin{definition}[Unitary extent and unitary rank, cf.~Definition 6 in~\cite{Bravyi_2019}]
Let~$U$ be an $n$-qubit unitary. 
Its unitary extent~$\xi(U)$ is
	\begin{align}
		\xi(U)=\min \left\{\| c \|_1^2: \chi \in \mathbb{N}, K_j \in \rm{G}_n, c_j \in \bb{C}\text{ such that } U=\sum_{j=1}^\chi c_j K_j \right\} \ .
	\end{align}
Its unitary rank~$\chi(U)$ is
	\begin{align}
		\chi(U)=\min \left\{ \chi : \chi \in \mathbb{N}, K_j \in \rm{G}_n, c_j \in \bb{C}\text{ such that } U=\sum_{j=1}^\chi c_j K_j \right\} \ .
	\end{align}
\end{definition}
A decomposition~$U=\sum_j c_j K_j$ where~$K_j \in \rm{G}_n, c_j \in \bb{C}$ is called a feasible decomposition of the unitary~$U$. A feasible decomposition that achieves~$\xi(U) = \|c\|_1^2$ is called an optimal decomposition of the unitary~$U$ with respect to the unitary extent.
The unitary extent gives the following upper bound for the pure state extent:
\begin{align}
	\label{eq:unitaryextentbound}
	\xi(U\ket{\phi}) \leq \xi(U) \qquad \text{for any Gaussian state~$\ket{\phi}\in \cal{G}_n$} \ .
\end{align}
This is, the pure state extent of a state obtained by applying a unitary~$U$ to any Gaussian state is upper bounded by the unitary extent of~$U$. This is due to the fact that an optimal decomposition~$U = \sum_{j} c_j V_j$ for the unitary~$U$ gives a feasible decomposition~$U \ket{\phi} = \sum_{j} c_j V_j \ket{\phi}$ for the state~$U\ket{\phi}$. As a consequence,~$\xi(U \ket{\phi} ) \leq \|c\|_1^2 = \xi(U)$.

Two states~$\ket{\psi_1}, \ket{\psi_2} \in\cal{H}_n$ are said to be Gaussian-unitarily-equivalent if there is a Gaussian unitary circuit~$V \in G_n$ such that~$\ket{\psi_2} = V \ket{\psi_1}$. Similarly, two unitaries~$U_1$ and~$U_2$ are said to be Gaussian-unitarily-equivalent if there is a Gaussian circuit~$V \in G_n$ such that~$U_2 = V U_1$. 
By the definition of the state extent, Gaussian-unitarily-equivalent states have the same state extent. Similarly, by the definition of unitary extent, Gaussian-unitarily-equivalent unitaries have the same unitary extent. Respectively,
\begin{align}
	\label{eq:extent_equiv}
	\xi(V \ket{\psi}) = \xi(\ket{\psi})
	\qquad\text{ and }\qquad
	\xi(V U) = \xi(U)
	\qquad\text{ for any } V \in G_n \ .
\end{align}

Before proceeding to define magic monotones for channels, it is convenient to define the notion of a convex-unitary channel.

\begin{definition}[Convex-unitary channel]
	A channel~$\cal{E}$ is a convex-unitary channel if it can be written as a convex combination of unitary channels, i.e., if~$\cal{E} = \sum_j p_j \mathcal{U}_j$ where~$\sum_j p_j = 1$ and~$\mathcal{U}_j = U_j^\dagger (\cdot) U_j$ with~$U_j$ a unitary operator for each~$j$.
\end{definition}

\begin{definition}[Convex-unitary channel extent, analogous to Definition 4.19 of stabilizer channel extent in~\cite{seddonThesis}]
	\label{def:channelextent}
The channel extent~$\Xi(\cal{E})$ of a convex-unitary channel~$\cal{E}$ is
	\begin{align}
	\Xi(\mathcal{E}) &= \min \left\{\sum_j p_j \xi\left(V_j\right): \cal{E}=\sum_j p_j \mathcal{V}_j, \mathcal{V}_j \, \text{a unitary channel}, p_j\geq 0, \sum_j p_j = 1  \right\} \ .
	\end{align}
\end{definition}
A decomposition~$\cal{E}=\sum_j p_j \mathcal{V}_j$ where each~$\mathcal{V}_j$ is a unitary channel and~$p_j \geq 0$ is called a feasible decomposition of the channel~$\mathcal{E}$. A decomposition that achieves~$\Xi(\mathcal{E})=\sum_j p_j \xi\left(V_j\right)$ is called an optimal decomposition of the channel~$\cal{E}$ with respect to the channel extent. A decomposition~$\cal{E}=\sum_j p_j \mathcal{V}_j$ of the channel~$\cal{E}$ is said to be equimagical if~$\xi(V_i) = \xi(V_j)$ for all~$i,j$. A feasible equimagical decomposition~$\cal{E}=\sum_j p_j \mathcal{V}_j$ of~$\cal{E}$ gives the upper bound~$\Xi(\cal{E}) \leq \xi(V_1)$ for the channel extent.

We introduce the following generalization of the channel extent, which allows a decomposition of a channel not only in terms of unitary channels but also in terms of general Gaussian channels.
\begin{definition}[Augmented channel extent]
\label{def:channelextentaugmented}
The augmented channel extent~$\Xiaug(\cal{E})$ of a convex-unitary channel~$\cal{E}$ is
\begin{align}
	\Xiaug(\mathcal{E}) = \min \Bigg\{&\sum_j p_j \xi\left(V_j\right) + \sum_j q_j : 
	\mathcal{E}=\sum_j p_j \mathcal{V}_j + \sum_k q_k \cal{G}^\mathrm{convex}_k ; \cal{V}_j \, \text{a unitary channel}; \\
	&\qquad \cal{G}^\mathrm{convex}_k \, \text{a convex-Gaussian channel}; p_j, q_k\geq 0; \sum_j p_j + \sum_k q_k = 1 \, \Bigg\} \ .
\end{align}	
\end{definition}
A decomposition~$\mathcal{E}=\sum_j p_j \mathcal{V}_j + \sum_k q_k \cal{G}^\mathrm{convex}_k$ where~$\cal{V}_j$ is a unitary channel,~$\cal{G}^\mathrm{convex}_k$ is a convex-Gaussian channel and~$p_j, q_k\geq 0$ is called a feasible decomposition of the channel~$\cal{E}$. If~$\Xiaug(\cal{E})=\sum_j p_j \xi\left(V_j\right) + \sum_j q_j$ then such a decomposition is called optimal (with respect to the augmented channel extent).
It is worth noting that the augmented channel extent is well-defined for a broader class of channels than convex-unitary channels. A detailed analysis of this extended set is left for future work.

We define the Gaussian fidelity which will be useful in proving a number of results in this work.
\begin{definition}[Gaussian fidelity, analogous to definition of stabilizer fidelity in Eq.~(125) of~\cite{Bravyi_2019}]
\label{def:fidelity}
The Gaussian fidelity~$F(\ket{\psi})$ of a~$n$-qubit state~$\ket{\psi}\in\cal{H}_n$ is the maximum squared overlap between the state~$\ket{\psi}$ and any~$n$-qubit Gaussian state, i.e., 
\begin{align}
F(\ket{\psi}) = \max_{\ket{\phi} \in \cal{G}_n} |\langle{\psi}|{\phi}\rangle|^2 \ . 
\end{align}
\end{definition}
Similarly to the stabilizer setting~\cite{Bravyi_2019}, the inverse Gaussian fidelity lower bounds the pure state extent:
\begin{align}
	\label{eq:xi-fidelity-bound}
	\xi(\ket{\psi}) \geq F(\ket{\psi})^{-1} \quad\text{ for any } \ket{\psi}\in\cal{H}_n \ .
\end{align}

We introduce an additional monotone called the dyadic negativity, previously defined in~\cite{Seddon_2021} in the context of stabilizer-based decompositions. This monotone will be useful in proving optimality results for decompositions of noisy rotation channels. 

\begin{definition}[Dyadic negativity, analogous to the Definition 5 of the (stabilizer) dyadic negativity in~\cite{Seddon_2021}]
\label{def:dyadic-negativity}
The dyadic negativity~$\Lambda(\rho)$ of a state~$\rho$ is 
	\begin{alignat}{2}
	\Lambda(\rho) &= \min \left\{\|q\|_1: \rho=\sum_j q_j \ket{L_j}\!\bra{R_j}; \ket{L_j}, \ket{R_j} \in  \cal{G}_n ; q_j \in \mathbb{C}\right\}  &&\quad\text{(primal form)}\\
	\label{eq:dyadicneg_dual}
	&=\max \left\{\tr[W \rho]: W \in \mathcal{W}_{\Lambda}\right\}  &&\quad\text{(dual form)} \ ,
	\end{alignat}
	where~$\mathcal{W}_{\Lambda} =\left\{W : W \text{ Hermitian}, |\langle L|W| R\rangle| \leq 1 \text{ for all } |L\rangle,|R\rangle \in  \cal{G}_n \right\}$.
\end{definition}

An operator~$W$ in the set~$\mathcal{W}_{\Lambda}$ is called a dyadic negativity witness and~$\mathcal{W}_{\Lambda}$ is called the set of dyadic negativity witnesses. The witness~$W$ which achieves~$\tr[W \rho] = \Lambda(\rho)$ is called optimal.

Next, we provide a lemma which is the Gaussian analogue of the ``sandwich theorem'' for stabilizers (see~\cite[Theorem 6.3]{seddonThesis}). This lemma establishes that the dyadic negativity lower bounds the channel extent.

\begin{lemma}[Fermionic sandwich lemma]
\label{lem:sandwich}
Consider a channel~$\cal{E}$ and with decomposition
\begin{align}
	\cal{E}(\cdot) = \sum_{j} p_j \cal{U}_j(\cdot) 
	\label{eq:Edecomp}
\end{align}
into unitary channels~$\cal{U}_j$.
We have
\begin{align}
	\label{eq:inequality-sandwich}
	\Lambda(\cal{E}(\ketbra{g})) \leq \Xi(\cal{E}) \leq  \sum_{j} p_j \xi(U_j)\ ,
\end{align}
for any Gaussian state~$\ket{g}\in\cal{G}_n$.
Moreover, if 
\begin{align}
\label{eq:dyadicneg-matches-channelextent}
\Lambda(\cal{E}(\ketbra{g})) = \sum_{j} p_j \xi(U_j)
\end{align}
then the decomposition~\cref{eq:Edecomp} is optimal with respect to the channel extent which is~$\sum_{j} p_j \xi(U_j)$.
\end{lemma}

\begin{proof}
	Consider a Gaussian state~$\ket{g}\in\cal{G}_n$. An optimal decomposition~$\cal{E}(\cdot) = \sum_j q_j \cal{V}_j (\cdot)$ of a channel~$\cal{E}$ with respect to the channel extent gives the following feasible decomposition of~$\cal{E}(\ketbra{g})$ with respect to the dyadic channel negativity 
\begin{align}
	\cal{E}(\ketbra{g}) 
	&= \sum_j q_j V_j^\dagger \ketbra{g} V_j \\
	\label{eq:sandwich-proof-aux1} &= \sum_{j,k,k'} q_j \overline{c_{j,k}} c_{j,k'} G_{j,k}^\dagger \ketbra{g} G_{j,k'} 
\end{align}
where~$V_j = \sum_k c_{j,k} G_{j,k}$ is an optimal unitary extent decomposition of~$V_j$ into Gaussian unitaries~$G_{j,k}\in G_n$ and~$\xi(V_j) = \|  c_j \|_1^2$. The decomposition in~\cref{eq:sandwich-proof-aux1} gives the following upper bound for the dyadic negativity
\begin{align}
	\Lambda( \cal{E}(\ketbra{g}) ) 
	\leq \sum_{j} q_j \left(\sum_k |c_{j,k}|\right)  \left(\sum_{k'} |c_{j,k'}| \right) 
	= \sum_{j} q_j \|c_{j}\|_1^2 
	= \sum_{j} q_j \xi(V_j) = \Xi(\cal{E}) \ ,
\end{align}
thus proving the first inequality in~\cref{eq:inequality-sandwich}. 
The second inequality in~\cref{eq:inequality-sandwich} results from considering the feasible decomposition given in~\cref{eq:Edecomp}.

Assume that~\cref{eq:dyadicneg-matches-channelextent} is satisfied. Then the set of inequalities  in~\cref{eq:inequality-sandwich} is tight, which implies that the channel extent is~$\sum_j p_j \xi(U_j)$ and the decomposition given in~\cref{eq:Edecomp} is optimal.

\end{proof}

\section{Rank and extent characterization of efficiently simulable non-fermionic states\label{sec:extentplus}}

In this section we give a resource-theoretic characterization of a class of states previously shown to be efficiently classical simulable by Brod \cite{Brod2016}. In that work it was shown that any $t$-qubit tensor product $\ket{\psi} = \bigotimes_{j=1}^t \ket{\psi_j}$, where $\ket{\psi_j}$ are arbitrary single-qubit states, can be prepared by initializing $t$ qubits in the state $\ket{0}$ along with just one additional qubit prepared in the non-fermionic state $\ket{+}$, followed by a unitary Gaussian circuit. We note that this immediately leads to the following lemma.
\begin{lemma}\label{lem:BrodConsequence}
For any arbitrary product of single-qubit states,  $\ket{\psi} = \bigotimes_{j=1}^t \ket{\psi_j}$, the tensor product of that state with the plus state has Gaussian rank and extent of exactly two,
\begin{equation}
\chi \left( \ket{\psi}\otimes \ket{+} \right) = \xi \left(\ket{\psi}\otimes \ket{+} \right) = \chi(\ket{+}) = \xi(\ket{+}) =  2 , \label{eq:singleQproductPlusAppended}
\end{equation}
and the rank and extent of the product state itself are upper bounded by
\begin{equation}
\xi\left(\ket{\psi}\right) \leq \chi\left(\ket{\psi}\right)   \leq 2 .
\end{equation}
\end{lemma}
\begin{proof}
It is manifestly true that the extent and rank of the state $\ket{+}$ are both $2$, since $\ket{0}$ and $\ket{1}$ are the only single-qubit Gaussian states (up to global phase), and therefore $\ket{+} = (\ket{0} + \ket{1})/\sqrt{2}$ is the only feasible decomposition that does not repeat terms. Prepending a state with a string of qubits prepared in the vacuum state cannot increase rank or extent, so $\ket{0}^{\otimes t} \otimes \ket{+}$ also has rank and extent equal to two. The first statement \cref{eq:singleQproductPlusAppended} then follows from Brod's result and the monotonicity of Gaussian rank and extent under unitary Gaussian circuits. 

The upper bound on the rank and extent of $\ket{\psi}$ follows from the strong monotonicity property \cite{Regula_2018}; these quantities cannot increase on average under free measurements. In the case of the product state $\ket{\psi}\otimes \ket{+}$, the non-Gaussian state on the first $t$ qubits $\ket{\psi}$ is the same whatever outcome is obtained on qubit $t+1$ if it is measured in the Fock basis. Appending a Fock state cannot change the Gaussian rank or extent, so $\xi(\ket{\psi})=\xi(\ket{\psi}\otimes \ket{0}) =  \xi(\ket{\psi}\otimes \ket{1})$. But since strong monotonicity implies the average extent cannot increase over all outcomes, and both outcomes have the same extent, it follows that  $\xi(\ket{\psi})\leq \xi(\ket{\psi}\otimes \ket{+})=2$. The same argument goes through for rank. The extent of $\ket{\psi}$ is upper bounded by the rank since if the state is Gaussian, both rank and extent are 1, whereas if it is non-Gaussian, the rank must be 2 but the extent can be smaller.
\end{proof}
We can then give the following corollary.
\begin{corollary}\label{cor:productStateBounds}
Consider a product state of the form  $\ket{\psi} = \bigotimes_{j=1}^t \ket{\psi_j}$, where $\ket{\psi_j}$ can be arbitrary single-qubit states. Let $\xi_{\mathrm{max}}$ be the maximum extent over all individual sites, $\xi_{\mathrm{max}} = \max_j \xi(\ket{\psi_j})$. Then the extent of the product state is bounded as
\begin{equation}
\xi_{\mathrm{max}} \leq \xi(\ket{\psi}) \leq 2.\label{eq:productStateBounds}
\end{equation}
Similarly, if any single-qubit state $\ket{\psi_j}$ has Gaussian rank 2, then $\chi(\ket{\psi})=2$; otherwise $\chi(\ket{\psi})=1$.
\end{corollary}

\begin{proof}
For any qubit $j$, by the strong monotonicity property of Gaussian extent, $\xi(\ket{\psi_j}) \leq \xi(\ket{\psi})$. Therefore the extent for any of the single-qubit states gives a lower bound on that of the $t$-qubit state, and the qubit saturating $\xi_{\mathrm{max}} $ gives the tightest of these. Along with \cref{lem:BrodConsequence}, this yields \cref{eq:productStateBounds}.

We can argue similarly for the Gaussian rank, except there are only two possible cases. Either $\ket{\psi}$ is a computational basis state, in which case its Gaussian rank is 1, or at least one qubit is in a superposition of $\ket{0}$ and $\ket{1}$, and therefore has Gaussian rank 2. In the latter case, the $t$-qubit state $\ket{\psi}$ is non-Gaussian, so must have rank at least 2, but \cref{lem:BrodConsequence} already shows it can be no larger than this. This gives the result.
\end{proof}
We remark that this behaviour is dramatically different from that of four-mode fermionic magic states, which can promote Gaussian circuits to universality \cite{Hebenstreit_2019}, and for which the extent is exactly multiplicative \cite{Dias_2024,cudby2024gaussiandecompositionmagicstates,reardonsmith2024fermioniclinearopticalextent}. 

As an example, we now give explicit optimal decompositions for a particular class of non-fermionic product states. We first define the equatorial single-qubit states~$\ket{+_\delta} = (\ket{0} + e^{i \delta} \ket{1})/\sqrt{2}$, and note that these are Gaussian-unitarily-equivalent to $\ket{+}$ by single-qubit $Z$-rotations (up to global phase), and therefore have $\xi(\ket{\psi_j})=\chi(\ket{\psi_j})=2$. Then define the set
\begin{equation}
E_t = \{ \ket{\psi}= \bigotimes_{j=1}^t \ket{\psi_j} | \ket{\psi_j} \in \{ \ket{0}, \ket{1}, \ket{+_{\delta_j}} , \, \delta_j \in [0,2\pi); \ket{\psi}\notin \mathcal{G}_t \}.
\end{equation}
The stipulation that $\ket{\psi}$ is not Gaussian excludes computational basis states, so that for any state $\ket{\psi}\in E_t$, there is at least one $j$ such that $\ket{\psi_j}= \ket*{+_{\delta_j}}$. Then by \cref{cor:productStateBounds}, we have $\xi(\ket{\psi})=\chi(\ket{\psi})=2$ for all $\ket{\psi}\in E_t$. We now illustrate how the optimal decomposition can be constructed. We will see that the decomposition must always split into an even and odd parity term. First we consider the case of a $k$-fold tensor product of the plus state, $\ket{+}^{\otimes k} \in E_k$.

Denote by~$S^+_kt$ ($S^-_k$) the set of~$k$-bit-strings with even (odd) parity, i.e., 
\begin{align}
	S_k^+ = \{ x \in \{0,1\}^k : |x| \text{ is even} \} 
	\quad\text{ and }\quad
	S_k^- = \{ x \in \{0,1\}^k : |x| \text{ is odd} \} \ ,
\end{align}
where~$|x| = \sum_{j=1}^k x_j$ is the Hamming weight of the string~$x\in\{0,1\}^k$. Consider the states 
\begin{align}
	\label{eq:enon} \ket{e_t} = \frac{1}{\sqrt{2^{k-1}}} \sum_{x\in S^+_k} \ket{x} 
	\quad\text{ and }\quad
	\ket{o_t} = \frac{1}{\sqrt{2^{k-1}}} \sum_{x\in S^-_k} \ket{x}
\end{align}
which are equal weight superpositions respectively of even and odd Hamming weight strings. Then we can write down the decomposition
\begin{equation}
\ket{+}^{\otimes k } = \frac{1}{\sqrt{2}} (\ket{e_k} + \ket{o_k}).
\end{equation}
The following lemma shows that this is a valid Gaussian decomposition. It is then easy to see that the decomposition is optimal with respect to both Gaussian rank and extent, the latter since $(2\cdot \abs*{1/\sqrt{2}})^2 = 2$.

\begin{lemma}
	\label{lem:psi-plus-minus-gaussian}
	The states~$\ket{e_k}$ and~$\ket{o_k}$ are Gaussian states.
\end{lemma}

\begin{proof}
The state~$\ket{e_k}$ can be obtained from the all zero state~$\ket{0}^{\otimes k}$ (which is a Gaussian state) as follows
\begin{align}
	\label{eq:psiplus-circuit-decomp}
	\ket{e_k} = R_{X_{k-1}Y_k}\left(\frac{\pi}{2}\right) R_{X_{k-2}Y_{k-1}}\left(\frac{\pi}{2}\right) \cdots R_{X_1 Y_2}\left(\frac{\pi}{2}\right) \ket{0}^{\otimes t} \ ,
\end{align}
where
\begin{align}
	R_{X_j Y_{j+1}} (\theta) = \exp(- i \frac{\theta}{2} X_j Y_{j+1}) 
\end{align}
is a matchgate. This implies that~$\ket{e_t}$ is a Gaussian state. 

We have~$\ket{o_t} = X_1 \ket{e_t}$. Since~$\ket{e_t}$ is a Gaussian state and~$X_1$ is a Gaussian unitary, it follows that~$\ket{o_t}$ is also a Gaussian state. This concludes the proof. 
\end{proof}

Having given the optimal decomposition for $\ket{+}^{\otimes k}$, we can show how to transform the even and odd terms with additional Gaussian gates to obtain the optimal decomposition for more general gates in $E_t$. Suppose $\ket{\psi}$ has $k$ qubits in state $\ket*{+_{\delta_j}}$ at arbitary positions $j$, and all other qubits are either $\ket{0}$ or $\ket{1}$. We start by appending a register of $t-k$ qubits in the zero states to obtain
\begin{equation}
\ket{0}^{\otimes{t-k}} \otimes \ket{+}^{\otimes{k}} = \frac{1}{\sqrt{2}} (\ket{0}^{\otimes{t-k}} \ket{e_k} + \ket{0}^{\otimes{t-k}}  \ket{o_k}) \label{eq:paddedDecomp}
\end{equation}
Then, as already noted, the plus states can be mapped to arbitrary $\ket{+_{\delta}}$ by single-qubit Gaussian rotations, 
\begin{align}
	\label{eq:+_+delta_equiv}
	\bigotimes_{j=1}^{k} \ket*{+_{\delta_j}} = e^{i \sum_j \delta_j / 2} \left(\prod_{j=1}^{k} R_{Z_j}({\delta_j})\right) \ket{+}^{\otimes k} \ ,
\end{align}
Next, as in Brod's argument \cite{Brod2016}, we can use the matchgate~$f\SWAP$ defined in~\cref{eq:fswap} to move $\ket{+_\delta}$ states past $\ket{0}$ states into the required locations. Then $X$ gates can be applied as necessary to flip $\ket{0}$ states to $\ket{1}$. Applying this same sequence of gates to each of the Gaussian terms $\ket{0}^{\otimes{t-k}} \ket{e_k}$ and  $\ket{0}^{\otimes{t-k}}  \ket{o_k}$ in~\cref{eq:paddedDecomp} gives a Gaussian decomposition of the target state. Since the amplitudes of these terms have not changed, it remains an optimal decomposition with $\xi(\ket{\psi}) =2$.

We leave as an open problem the optimal decomposition of more general product states, where $\xi(\ket{\psi_j}) < 2$ for all $j \in [t]$. We note that Brod's construction only implies an upper bound for the extent of $2$, but the extent for specific non-fermionic product states can be strictly smaller.

\section{Optimal decompositions of resourceful unitary gates \label{sec:Udecomp}}

We give optimal decompositions with respect to the unitary extent for the following gates (see~\cref{tab:optimaldecomp} for the definition of the gates):
\begin{enumerate}[1)]
	\item any two-qubit nearest-neighbour or diagonal fermionic gate (this includes the commonly used controlled-phase gate~$C(\theta)$ for an angle~$\theta\in[0,2\pi)$, the two-qubit rotation~$R_{ZZ}(\theta)$ around the~$ZZ$-axis by an angle~$\theta\in[0,2\pi)$, and the~$\SWAP$ gate)
	\item the Hadamard gate~$H$,
	\item the single-qubit rotation~$R_{Y}(\theta)$ around the~$Y$-axis by an angle~$\theta\in[0,2\pi)$.
\end{enumerate}
We summarize the results in~\cref{tab:optimaldecomp} and in~\cref{fig:unitary_extent}.
In the following subsections we provide the proofs for the results summarized in~\cref{tab:optimaldecomp}.

\begin{figure}[h]
    \centering
    \includegraphics[width=0.5\linewidth]{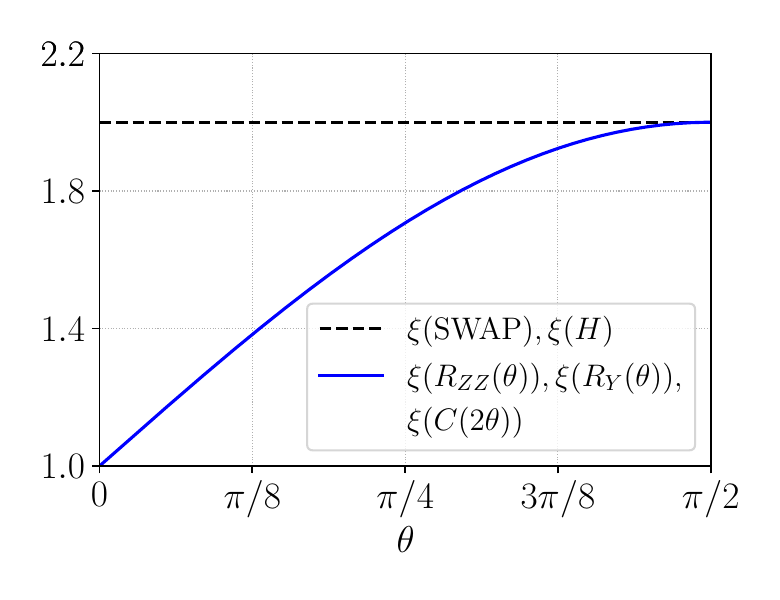}
    \caption{The gates~$\SWAP$ and~$H$ have unitary extent~$\xi(\SWAP) = \xi(H) = 2$. The gate~$C(\theta)$ has unitary extent~$C(\theta) = 1+|\sin(\theta/2)|$, the gates~$R_{ZZ}(\theta)$ and~$R_Y(\theta)$ have unitary extent~$\xi(R_{ZZ}(\theta)) = \xi(R_Y(\theta)) = 1+|\sin\theta|$.  }
    \label{fig:unitary_extent}
\end{figure}

In~\cref{sec:decomp_2qbit_gates_mom} we give optimal unitary extent decompositions for any two-qubit fermionic gate, and the respective value for the unitary extent. 
The proof relies on the so-called lifting lemma. This lemma allows `lifting' known optimal state extent decompositions of magic states to optimal unitary extent decompositions of the corresponding magic gates. This is, for a magic gate which can be implemented by the use of a Gaussian gadget together with the magic state. This construction was first considered in the context of stabilizer computations (see~\cite[Lemma 1]{Bravyi_2019}). In~\cref{sec:liftinglemma} we prove a similar result in the context of Gaussian computations. 
The gates~$ C(\theta), R_{ZZ}(\theta)$ and~$\SWAP$ can be implemented with a known Gaussian gadget that utilizes a fermionic magic state with known optimal state extent decomposition. This allows utilizing the lifting lemma to obtain optimal unitary extent decompositions, which we do in~\cref{sec:decomp_C_RZZ_SWAP}. In~\cref{sec:decomp_twoqbit_gates} we generalize this result to any fermionic two-qubit gate acting on nearest neighbour qubits. Because any two-qubit fermionic gate can be decomposed into a circuit with the~$ZZ$-rotation gate~$R_{ZZ}(\theta), \theta\in\bb{R}$ as the only non-Gaussian resource, an optimal decomposition of the latter gives an optimal decomposition of the former. 

In~\cref{sec:decomp_H_RY}, we present optimal decompositions for the gates~$H$ and~$R_Y(\theta)$. The proofs exploit the structure of these gates: although they are not fermionic (since they do not preserve parity), each can be expressed as the sum of a parity-preserving fermionic gate and a parity-flipping fermionic gate.

\subsection{Optimal unitary extent decompositions of two-qubit fermionic gates \label{sec:decomp_2qbit_gates_mom}}

Here, we give optimal decompositions for any two-qubit diagonal or Jordan-Wigner adjacent fermionic gate, and the respective unitary extent.

\subsubsection{Lifting lemma \label{sec:liftinglemma}}

Here we prove the so-called lifting lemma for Gaussian unitaries. Analogously to the stabilizer setting, this lemma allows mapping optimal decompositions of magic states to optimal decompositions of the corresponding magic gate. However, while the lemma in the stabilizer setting applies to diagonal gates and its proof is based on that (see~\cite[Lemma 1]{Bravyi_2019}), our proof relies instead on considering the gadget in~\cref{fig:gadget} which relates the magic state (whose optimal decomposition we wish to lift) to the magic gate (whose optimal decomposition we wish to obtain by lifting). Additionally, while for stabilizer computations the simplest magic gate is a one-qubit gate which can be applied using a gadget which takes as input a one-qubit magic state, for Gaussian computations the lowest-dimensional magic state is a four-qubit state, since all one-, two- and three-qubit fermionic states are Gaussian (see e.g.~\cite{Hebenstreit_2019} and~\cite[Proposition 1]{Melo_2013}). This leads to the peculiarity that gadgetizing the lowest-dimensional fermionic magic gates, which are two-qubit gates, requires a four-qubit magic state.

\begin{lemma}[Lifting lemma for Gaussian unitaries]
\label{lem:lifting}
	Consider a two-qubit fermionic unitary~$V$, the 4-qubit Gaussian state~$\ket{g} = \ket{\psi^+} \otimes \ket{\psi^+}$ with~$\ket{\psi^+} = \left( \ket{00} + \ket{11} \right) / \sqrt{2}$ and 
	\begin{align}
		\label{eq:ketVdecomposition}
		\ket{v} = \left(I \otimes V \otimes I \right) \ket{g} = \sum_{j=1}^\chi c_j \ket{g_j}
	\end{align}	
	where~$\ket{g_j} = (I \otimes K_j \otimes I) \ket{g}$ with~$K_j \in G_n$,~$j\in[\chi]$ Gaussian unitaries. 
	Then the following hold:
	\begin{enumerate}[1)]
          \item We have~$V^{\otimes t} =  \left(\sum_{j=1}^{\chi} c_j K_j\right)^{\otimes t}$ and~$\xi(V^{\otimes t}) \leq \norm{c}_1^{2t}$ for~$t\in\bb{N}$. 
	  \item If the decomposition of~$\ket{v}$ in~\cref{eq:ketVdecomposition} is optimal with respect to the state extent, i.e., if~$\xi(\ket{v}) = \| c \|_1^2$, then
		\begin{align}
			\label{eq:lem-lifting-xiV-equal}
			\xi(V) = \xi(\ket{v}) \ .
		\end{align}
		Furthermore, we have~$\xi(V^{\otimes t}) = \xi(V)^t$ for~$t\in\bb{Z}$.
	\end{enumerate}
\end{lemma}

\begin{proof}
Consider the gadget given in~\cref{fig:gadget} (see also~\cite[Fig. 1]{reardonsmith2024improved}). 
Consider the operator 
\begin{align}
	A = 4 \left( \bra{\psi^+} \otimes I^{\otimes 2} \otimes \bra{\psi^+} \right)
\end{align}
which implements this gadget.
We have
\begin{align}
	\label{eq:gadgetA1} A( I \otimes \ket{v} \otimes I  ) &=  V \ , \\
	\label{eq:gadgetA2} A( I \otimes \ket{K_j} \otimes I  ) &=  K_j  \quad\text{ with }\quad  \ket{K_j}:=\left(I  \otimes K_j \otimes I  \right) \ket{g}  \quad\text{ for all }\quad j\in[\chi] \ .
\end{align}
This together with~\cref{eq:ketVdecomposition} gives 
\begin{align}
	V 
	&= A( I  \otimes \ket{v} \otimes I  ) \\
	&= A\left( I \otimes \left( \sum_{j=1}^\chi c_j \ket{K_j} \right) \otimes I  \right) \ket{g}  \\
	&= \sum_{j=1}^\chi c_j  A( I  \otimes \ket{K_j} \otimes I  ) \\
	&= \sum_{j=1}^\chi c_j   K_j  \ .
\end{align}
This gives~$\xi(V) \leq \|c\|_1^2$ and more generally~$\xi(V^{\otimes t}) \leq \norm{c}_1^{2t}$. Considering this together with~\cref{eq:unitaryextentbound}, if the decomposition in~\cref{eq:ketVdecomposition} of~$\ket{v}$ is optimal, i.e.,~$\xi(\ket{v}) = \|c\|_1^2$, we obtain~\cref{eq:lem-lifting-xiV-equal}.

The~$t$-fold tensor product~$V^{\otimes t}$ can be implemented using~$t$ copies of the gadget in~\cref{fig:gadget} stacked vertically. We call the resulting gadget~$B$. Similar results to~\cref{eq:gadgetA1,eq:gadgetA2} hold. Namely,
\begin{align}
	B\left((I\otimes \ket{v}\otimes I)^{\otimes t}\right) = V^{\otimes t}
	\quad\text{ and }\quad
	B\left( \bigotimes\nolimits_{k=1}^t (I\otimes \ket{K_{j_k}} \otimes I)\right) = \bigotimes\nolimits_{k=1}^t K_{j_k} \ .
\end{align}
Hence, similarly to before we have~$\xi(V^{\otimes t}) = \xi(\ket{v}^{\otimes t})$. Moreover, due to~\cref{lem:extent_multiplicative} we have~$\xi(\ket{v}^{\otimes t}) = \xi(\ket{v})^t$, and the result~$\xi(V^{\otimes t}) = \xi(\ket{v})^t = \|c\|_1^{2t}$ follows.
\end{proof}

Next, we apply~\cref{lem:lifting} to find optimal decompositions for several magic gates of interest. We use the fact that the pure state extent is multiplicative for~$t$-fold tensor products of 4-qubit (magic) states -- see~\cref{lem:extent_multiplicative}. By~\cref{lem:lifting}, this gives multiplicativity of the unitary extent for~$t$-fold tensor products of two-qubit magic gates. This means that when considering the tensor product and applying several magic gates in the same circuit layer, there is no room for improvement over taking the tensor product of individual optimal decompositions. However, note that applying magic gates sequentially can admit improved decompositions. For example, applying the~$\SWAP$ gate twice gives the identity operator and consequently~$\xi(\SWAP^2) = 1$, a result which contrasts with~$\xi(\SWAP^{\otimes 2}) = 4$, which we will show next.

Notice that~\cref{fig:gadget} allows us to identify the state $\ket{v} = (I \otimes V \otimes I) \ket{\psi^+}^{\otimes 2}$ as the Choi-Jamiołkowski state~\cite{Jamiolkowski1972,Choi1975} of the two-qubit gate $V$. Note the unusual ordering of the qubits in the definition: instead of acting on the third and fourth qubits as per the usual definition $(I^{\otimes 2} \otimes V) \ket{\psi^+}^{\otimes 2}$ of the Choi state, $V$ acts on the second and third qubits.

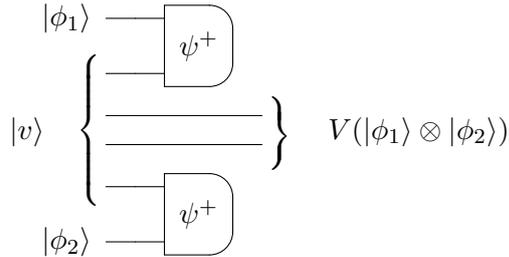
\begin{figure}[H]
	$$
	\Qcircuit @C=1em @R=1em {
	\lstick{\ket{\phi_1}} & \qw & \multimeasureD{1}{\psi^+} & \\
	\lstick{} & \qw & \ghost{\psi^+} & \\
	\lstick{} & \qw & \qw & \qw 
	\inputgroup{2}{5}{2em}{\hspace{9cm}V (\ket{\phi_1} \otimes \ket{\phi_2})} 
	\\
	\lstick{} & \qw & \qw & \qw 
	\inputgroupv{2}{5}{1.2em}{2em}{\ket{v}\hspace{2em}} 
	\gategroup{3}{4}{4}{4}{1.2em}{\}} \\
	\lstick{} & \qw & \multimeasureD{1}{\psi^+} & \\
	\lstick{\ket{\phi_2}} & \qw & \ghost{\psi^+} & \\
	} 
	$$
	\centering
        \caption{The gadget in the picture uses the magic state~$\ket{v} = (I \otimes V \otimes I) \ket{\psi^+}^{\otimes 2}$ to implement the associated magic gate~$V$ (see~\cite{Hebenstreit_2019}). The gadget consists of two Bell measurements which can be implemented using the matchgate~$G(H,H)$ and computational basis measurements~\cite{Hebenstreit_2019}. Postselecting the measurement outcome consisting of all zeros gives~$V (\ket{\phi_1} \otimes \ket{\psi_2})$ at the output of the circuit. Instead of postselecting, one can apply a circuit of Gaussian operations which uses ancillas initialized in computational basis states to correct when obtaining a measurement outcome different from all-zeros, see~\cite{Hebenstreit_2019}.}
\label{fig:gadget}
\end{figure}

\subsubsection{Controlled-phase~$C(\theta)$, the~$ZZ$-rotation~$R_{ZZ}(\theta)$, and~$\SWAP$ gates \label{sec:decomp_C_RZZ_SWAP}}

We apply~\cref{lem:lifting} to obtain optimal decompositions for the following gates. 

\hfill

\paragraph{Controlled-phase gate.}
The controlled-phase gate can be implemented using the gadget shown in~\cref{fig:gadget}~\cite{Hebenstreit_2019} (see also~\cite{reardonsmith2024improved}) using the magic state 
\begin{align}
	\label{eq:mtheta}
	\ket{m_\theta} = \frac{1}{2} \left( \ket{0000} + \ket{0011} + \ket{1100} + e^{i\theta} \ket{1111} \right)
\end{align}
and Gaussian operations. This gadget enables the implementation of a controlled-phase gate $C(\theta)_{j,k}$ on any two qubits $j$ and $k$, even when they are not nearest neighbours. The key observation is that the four-qubit resource state $\ket{m_\theta}$ has even parity on both its first and last pairs of qubits. As a result, each such pair of qubits can be moved to arbitrary positions in the circuit using $f\SWAP$ gates, allowing the controlled-phase gate to be applied to any two qubits.

In~\cite{reardonsmith2024improved} the authors proved that~$\ket{m_\theta}$ has state extent (see~\cite[Eq.~(30)]{reardonsmith2024improved})
\begin{align}
	\label{eq:extent-mtheta}
	\xi( \ket{m_\theta} ) = 1 + \left| \sin\left(\frac{\theta}{2}\right) \right| \ .
\end{align}
This value is achieved by the optimal decomposition
\begin{align}
	\label{eq:magic-state-m-optimal-decomposition}
	\ket{m_\theta} = e^{i \theta/4} \left( \cos\left(\frac{\theta}{4}\right) I + i \sin\left(\frac{\theta}{4}\right) Z_2 Z_3 \right) R_{Z_2}\left(\frac{\theta}{2}\right) R_{Z_3}\left(\frac{\theta}{2}\right) \ket{g} \ .
\end{align}

Notice that~$\ket{m_\theta} = (I  \otimes C(\theta) \otimes I  ) \ket{g}$.
This together with the optimal decomposition in~\cref{eq:magic-state-m-optimal-decomposition} and~\cref{lem:lifting} gives the unitary extent optimal decomposition
\begin{align}
	\label{eq:Ctheta-optimal-decomposition}
	C(\theta)  = e^{i \theta/4} \left( \cos\left(\frac{\theta}{4}\right) I + i \sin\left(\frac{\theta}{4}\right) Z \otimes Z \right)\left( R_{Z}\left(\frac{\theta}{2}\right) \otimes R_{Z}\left(\frac{\theta}{2}\right)\right) \ .
\end{align}
This decomposition achieves
\begin{align}
	\xi(  C(\theta) ) 
	&= \xi( \ket{m_\theta} )
	\label{eq:extent-C} = 1 + | \sin(\theta/2) |
\end{align}
where in the second equality we used~\cref{eq:extent-mtheta}. 

The optimal decomposition of~$\ket{m(\theta)}^{\otimes t}$ is the~$t$-fold tensor product of the optimal decomposition given in~\cref{eq:magic-state-m-optimal-decomposition}~\cite{cudby2024gaussiandecompositionmagicstates,reardonsmith2024fermioniclinearopticalextent}.  Hence, the result above generalizes to
\begin{align}
	\xi\left( C(\theta)^{\otimes t} \right) 
	&= \xi\left( \ket{m_\theta}^{\otimes t} \right)
	= \xi\left( \ket{m_\theta} \right)^t 
	=  \left(1 + | \sin(\theta/2) | \right)^t \ ,
\end{align}
and the~$t$-fold tensor product of the decomposition in~\cref{eq:Ctheta-optimal-decomposition} gives an optimal decomposition of~$C(\theta)^{\otimes t}$.

\hfill

\paragraph{Rotation about the~$ZZ$-axis.}
The~$ZZ$-rotation gate~$R_{ZZ}(\theta)$ acting on any two qubits (not necessarily nearest neighbour) is Gaussian-unitarily-equivalent to the gate~$C(\theta)$. Namely, we have 
\begin{align}
	R_{ZZ}(\theta) = e^{i \theta/2} \left(R_{Z}(\theta) \otimes R_{Z}(\theta)\right) C(-2\theta) \ ,
\end{align}
where~$R_{Z}(\theta)$ are matchgates. Then, by~\cref{eq:extent_equiv} we have 
\begin{align}
	\label{eq:extentRZZ}
	\xi( R_{ZZ}(\theta)  )= \xi( C(-2\theta) ) = 1 + | \sin\theta | \ .
\end{align}
An optimal decomposition is given by
\begin{align}
	R_{ZZ}(\theta) 
	\label{eq:RZZ-decomposition} &= \cos(\theta/2) I - i \sin(\theta/2) Z \otimes Z \ .
\end{align}
Similarly to the controlled-phase gate, \cref{eq:extentRZZ} generalizes to
\begin{align}
	\label{eq:xi_RZZ_tensor}
	\xi\left(R_{ZZ}(\theta)^{\otimes t} \right) = (1+|\sin(\theta)|)^t \ ,
\end{align}
and the~$t$-fold tensor product of  an optimal decomposition of~\cref{eq:RZZ-decomposition} is an optimal decomposition of~$R_{ZZ}(\theta)^{\otimes t}$.

\hfill

\paragraph{$\SWAP$ gate.}
The~$\SWAP$ gate acting on nearest neighbours is Gaussian-unitarily-equivalent to the controlled-phase gate~$C(\theta)$. Namely, we have
\begin{align}
	\label{eq:swap-ctheta} \SWAP 
	&= f\SWAP \cdot C(\pi) 
\end{align}
where~$f\SWAP$ given in~\cref{eq:fswap} is a matchgate (when acting on nearest neighbours). 
Then
\begin{align}
	\label{eq:xiSWAP}
	\xi(\SWAP) = \xi(C(\pi)) = 2 
\end{align}
and~\cref{eq:swap-ctheta} together with~\cref{eq:Ctheta-optimal-decomposition} gives the optimal decomposition
\begin{align}
	\SWAP 
	&= \frac{e^{i\pi/4}}{\sqrt{2}}f\SWAP\left( I + i Z \otimes Z \right) \left(R_Z\left(\frac{\pi}{2}\right) \otimes R_Z\left(\frac{\pi}{2}\right)\right) \\
	\label{eq:swap-optimal-decomposition} &= \frac{e^{i\pi/4}}{\sqrt{2}} G(-i I, X) \left( I + i Z \otimes Z \right) 
\end{align}
when acting on nearest-neigbours.
As for the controlled-phase gate, \cref{eq:xiSWAP} generalizes to
\begin{align}
	\xi\left( \SWAP^{\otimes t} \right) = 2^t \ ,
\end{align}
and the~$t$-fold tensor product of 
an optimal decomposition of~\cref{eq:swap-optimal-decomposition} is an optimal decomposition of~$ \SWAP^{\otimes t}$ if each $\SWAP$ acts on nearest neighbours. 

When acting on non-nearest neighbour qubits, the $f\SWAP$ gate is not a matchgate. Consequently, a non-nearest neighbour $\SWAP$ gate is not Gaussian-unitarily equivalent to the controlled-phase gate $C(\theta)$. Nevertheless, the decompositions presented here remain feasible. We leave an analysis of their optimality for future work.

\subsubsection{General two-qubit fermionic gates acting on nearest neighbours \label{sec:decomp_twoqbit_gates}}

The following lemma gives an optimal decomposition and the associated unitary extent for any two-qubit fermionic gate which 1) acts on nearest neighbours with respect to the Jordan-Wigner ordering or 2) is diagonal.

\begin{lemma}
\label{lem:extent2qubitgate}
Let~$U$ be a two-qubit fermionic gate.
The gate $U$ has unitary extent
\begin{align}
	\label{eq:decompU20} \xi(U) = 1 + |\sin c| 
\end{align}
and respective optimal decomposition
\begin{align}
	\label{eq:decompU2} U = (R_Z(t_1) \otimes R_Z(t_2)) \cdot R_{XX}(a) \cdot R_{YY}(b) \cdot \left( \cos\left(\frac{c}{2}\right) I - i \sin\left(\frac{c}{2}\right) Z \otimes Z \right) \cdot  (R_Z(t_3) \otimes R_Z(t_4)) \ ,
\end{align}
with~$t_1,t_2,t_3,a,b,c\in\bb{R}$ defined by~\cref{eq:gateUdecomp},
if
\begin{enumerate}[1)]
	\item \label{it:nn} $U$ acts on nearest neighbours or
	\item \label{it:diagonal}  $U$ is a diagonal gate.
\end{enumerate}
\end{lemma}

\begin{proof}
By~\cref{lem:fermionicgatedecomposition}, any fermionic two-qubit gate~$U$ admits a decomposition of the type
\begin{align}
	\label{eq:decompU3} U = (R_Z(t_1) \otimes R_Z(t_2)) \cdot R_{XX}(a) \cdot R_{YY}(b) \cdot R_{ZZ}(c) \cdot  (R_Z(t_3) \otimes R_Z(t_4)) 
\end{align}
where the parameters~$t_1,t_2,t_3,a,b,c$ can be computed by solving the corresponding system of equations. 

We first prove~\cref{it:nn}. When acting on nearest neighbours, the gates~$R_{XX}(t)$ and~$R_{YY}(t)$, $t\in\bb{R}$ are matchgates. The gates $R_Z(t) \otimes I$ and~$I \otimes R_Z(t), t\in\bb{R}$ are matchgates. Then, the gate~$R_{ZZ}(c)$ is the only non-Gaussian gate in~\cref{eq:decompU3}. Hence, by~\cref{eq:extent_equiv}, we have~$\xi(U) = \xi(R_{ZZ}(c))$. The claim in~\cref{eq:decompU20} is then a consequence of~\cref{eq:extentRZZ}. The decomposition in~\cref{eq:decompU2} is obtained by considering the optimal decomposition of~$R_{ZZ}(c)$ given by~\cref{eq:RZZ-decomposition}. As~\cref{eq:decompU2} achieves the unitary extent~$\xi(U) = \xi(R_{ZZ}(c)) = 1 + |\sin c|$, it is optimal. 

We now prove~\cref{it:diagonal}. The gate $U$ with decomposition given in~\cref{eq:decompU3} is diagonal 
if and only if $a = (k+\ell) \pi$ and $b = (k-\ell) \pi$ for $k,\ell\in\bb{Z}$. In this case, the gate $R_{XX}(a)$ is either proportional to the identity or to $XX$. Similarly, the gate $R_{YY}(b)$ is either proportional to the identity or to $YY$. Hence, $R_{XX}(a)$ and $R_{YY}(b)$ are Gaussian gates. The remain of the proof proceeds analogously to the proof of~\cref{it:nn}.
\end{proof}

Notice that for a two-qubit gate $U$ which 1) acts on nearest neighbours or 2) is diagonal, as $U^{\otimes t}$ and $R_{ZZ}(c)^{\otimes t}$ are Gaussian-unitarily-equivalent, by~\cref{eq:xi_RZZ_tensor,eq:extent_equiv} we have
\begin{align}
	\xi(U^{\otimes t}) = \xi( R_{ZZ}(c)^{\otimes t}) = (1 + |\sin c|)^t
\end{align}
and the tensor product of optimal decompositions \eqref{eq:decompU2} of $U$ is an optimal decomposition of $U^{\otimes t}$.

We do not give an optimal decomposition for the case when $U$ 1) acts on non-nearest neighbour qubits~$j$ and~$k$ and 2) is not diagonal. In this case, $R_{XX}(a)$ and $R_{YY}(b)$ are not matchgates, and finding an optimal decomposition of $U$ amounts to finding an optimal decomposition of $R_{XX}(a) R_{YY}(b) R_{ZZ}(c)$. Still, the decomposition \eqref{eq:decompU2} with each rotation gate~$R_{P}(\theta)$ decomposed as~$R_{P}(\theta) = \cos(\theta/2) I - i \sin(\theta/2) P$ with~$P \in \{X_jX_k,Y_jY_k,Z_jZ_k\}$ provides a feasible decomposition with extent~$\xi(U) = (1 + |\sin a|)(1 + |\sin b|)(1 + |\sin c|)$. We leave an analysis of optimality in this case for future work.

\subsection{Optimal unitary extent decompositions of single-qubit gates \label{sec:decomp_H_RY}}
We give optimal decompositions for the Hadamard~$H$ and~$Y$-rotation~$R_Y(\theta)$ gates, and the respective unitary extent.

\subsubsection{Hadamard gate \label{sec:hadamard}}

The Hadamard gate~$H$ admits the decomposition 
\begin{align}
	 \label{eq:Hdecomp}
	H = \frac{X + Z}{\sqrt{2}} \ .
\end{align}
Notice that~$X$ and~$Z$ are Gaussian unitaries. The feasible decomposition in~\cref{eq:Hdecomp} gives the upper bound
\begin{align}
	\label{eq:xiH-upperbound}
	\xi(H) \leq 2 \ .
\end{align}
By~\cref{eq:unitaryextentbound} and~\cref{lem:BrodConsequence} we have 
\begin{align}
	\label{eq:xiH-lowerbound}
	\xi(H) \geq \xi(H \ket{0}) = \xi(\ket{+}) = 2 \ .
\end{align}
Combining~\cref{eq:xiH-upperbound,eq:xiH-lowerbound} gives
\begin{align}
	\xi(H) = 2 \ .
\end{align}
Thus, \cref{eq:Hdecomp} is an optimal decomposition of the gate~$H$ with respect to the unitary extent.

\subsubsection{Rotations about the~$Y$- and~$X$-axes \label{sec:RY}}

The rotation gate about~$Y$-axis is
\begin{align}
	R_Y(\theta) 
	&= e^{-i \frac{\theta}{2} Y} \\
	\label{eq:RY-optimaldecomp} 
	&= \cos(\frac{\theta}{2}) I - i \sin(\frac{\theta}{2})  Y \ .
\end{align}
Each term in the decomposition in~\cref{eq:RY-optimaldecomp} is Gaussian. Hence, we have
\begin{align}
	\label{eq:xiRY-upperbound}
	\xi (R_Y(\theta) ) 
	&\leq \left( \left|\sin(\frac{\theta}{2})\right| + \left|\cos(\frac{\theta}{2})\right| \right)^2 \\
	&= 1 + \left|\sin\theta\right| \ .
\end{align}
Also, we have
\begin{align}
	\label{eq:xiRY-lowerbound}
	\xi (R_Y(\theta) ) 
	&\geq \xi (R_Y(\theta) \ket{0}) \\
	&\geq | \langle{\omega}|{R_Y(\theta) 0}\rangle |^2
\end{align}
where the first inequality follows from~\cref{eq:unitaryextentbound} and the second inequality follows the dual formulation of state extent given in~\cref{eq:extent-state-dual}, where~$\ket{\omega}$ is any state extent witness. Observe that 
\begin{align}
	\ket{\omega_{++}} = \ket{0} + \ket{1}
	\qquad \text{ and } \qquad
	\ket{\omega_{+-}} = \ket{0} - \ket{1}
\end{align}
are valid state witness since the only single-qubit Gaussian states are~$\ket{0}$ and~$\ket{1}$ with which~$\ket{\omega_{++}}$ and~$\ket{\omega_{+-}}$ have absolute overlap one. Observe that~$R_Y(\theta) \ket{0} = \cos(\theta/2) \ket{0} + \sin(\theta/2)\ket{1}$. For~$\theta \in [0, \pi]$ we have
\begin{align}
	\xi (R_Y(\theta) ) 
	&\geq | \langle{\omega_{++}}|{R_Y(\theta) 0} \rangle |^2 \\
	&= \left( \sin(\frac{\theta}{2}) + \cos(\frac{\theta}{2}) \right)^2 \\
	&= 1 + \sin\theta = 1 + |\sin\theta|
\end{align}
and for~$\theta \in [\pi, 2\pi]$ a similar result holds by considering the witness~$\ket{\omega_{+-}}$ instead.
Then, \cref{eq:xiRY-upperbound} together with~\cref{eq:xiRY-lowerbound} gives 
\begin{align}
	\label{eq:extent-RYtheta} \xi (R_Y(\theta) ) = 1 + |\sin\theta|
\end{align}
and~\cref{eq:RY-optimaldecomp} is an optimal decomposition of~$R_Y(\theta)$ with respect to the unitary extent.

The gate $R_X(\theta)$ is Gaussian-unitarily-equivalent to $R_Y(\theta)$. Namely,
\begin{align}
	R_X(\theta) = R_Z(-\pi/2) R_Y(\theta) R_Z(-\pi/2)^\dagger \, 
\end{align}
where the rotation $R_Z(-\pi/2)$ is a Gaussian gate. Then, by~\cref{eq:extent_equiv,eq:extent-RYtheta},
\begin{align}
	\xi(R_X(\theta)) = \xi(R_Y(\theta)) = 1 + |\sin\theta|
\end{align}
and 
\begin{align}
	R_X(\theta) = \cos(\frac{\theta}{2}) I - i \sin(\frac{\theta}{2}) X
\end{align}
is an optimal decomposition, since it achieves $\xi(R_X(\theta))$.

\section{Improved and optimal decompositions of noisy rotation channels \label{sec:noisydecomps}}

Here we give results about optimal (or improved) decompositions for rotation channels 
\begin{align}
	\cal{R_P}(\theta)(\cdot) = R_P(\theta) (\cdot) R_P(\theta)^\dagger
\end{align}
with~$R_P(\theta) = \exp(-i \frac{\theta}{2} P)$, subject to stochastic Pauli noise
\begin{align}
	\cal{E}_{\cal{P}'}(p)(\cdot) = (1-p) \cal{I} + p \cal{P}' = (1-p)(\cdot) + p P'(\cdot)P'{}^\dagger  \ ,
\end{align}
where~$P$ and~$P'$ are Pauli operators and~$p\in[0,1]$ is the error probability. This is, we are interested in decompositions of the channel~$\cal{E}_{\cal{P}'}(p) \circ \cal{R_P}(\theta)$.
Specifically, we consider the following cases:
\begin{enumerate}[1)]
	\item A rotation about the~$Y$-axis subject to~$Y$ noise
	\begin{align}
		\label{eq:NYdecomp1}
		\cal{N_Y} = \cal{E_Y}(p) \circ \cal{R_Y}(\theta) = (1-p)\cal{R_Y}(\theta) + p \cal{Y} \circ \cal{R_Y}(\theta) \ .
	\end{align} 
	\item A rotation about the~$ZZ$-axis subject to two-qubit~$ZZ$ noise
	\begin{align}
		\label{eq:NZZdecomp1}
		\cal{N_{ZZ}} = \cal{E_{ZZ}}(p) \circ \cal{R_{ZZ}}(\theta) = (1-p)\cal{R_{ZZ}}(\theta) + p \cal{ZZ} \circ \cal{R_{ZZ}}(\theta) \ .
	\end{align} 
	\item A rotation about the~$ZZ$-axis subject to single-qubit~$Z$ noise (without loss of generality we consider the noise in the first qubit)
	\begin{align}
		\label{eq:NZdecomp1}
		\cal{N'_{ZZ}} = \cal{N}_{\cal{ZZ}_{(1)}} = \cal{E}_{\cal{Z}_{1}}(p) \circ \cal{R_{ZZ}}(\theta) = (1-p)\cal{R_{ZZ}}(\theta) + p \cal{Z}_1 \circ \cal{R_{ZZ}}(\theta) \ .
	\end{align} 
	(We denote the channel with noise on the second qubit by~$\cal{N}_{\cal{ZZ}_{(2)}} = \cal{E}_{\cal{Z}_{2}}(p) \circ \cal{R_{ZZ}}(\theta)$.)
\end{enumerate}
In the cases considered, the Pauli channel~$\cal{P}'$ is a Gaussian unitary channel. As a result,~$\cal{E}_{\cal{P}'} \circ \cal{R_P}(\theta)$ can be expressed as equimagical decompositions, i.e., as a convex combination of unitary channels each of which has the same channel extent, in this case~$\xi(R_{P}(\theta))$.
Then, the decompositions in~\cref{eq:NYdecomp1,eq:NZZdecomp1,eq:NZdecomp1} allow upper bounding the channel extent of the noisy channel by the unitary extent of the respective noiseless unitary, i.e., 
\begin{align}
	\Xi( \cal{E}_{\cal{P}'} \circ \cal{R_P}(\theta) ) \leq \xi(R_{P}(\theta)) = 1 + |\sin\theta| \ .
\end{align}
In particular
\begin{align}
	\label{eq:xi-noisechannels-ub1}
	\Xi( \cal{N_Y} ), \Xi( \cal{N_{ZZ}} ) , \Xi( \cal{N_{ZZ}}' )  \leq 1 + |\sin\theta| \ .
\end{align}

This section is structured as follows. 
We give optimal decompositions for~$\cal{N_Y}$ and~$\cal{N_{ZZ}}$ in~\cref{sec:optimal_decomp_NY_NZZ}. We further give an improved decomposition of~$\cal{N_{ZZ}}'$ in~\cref{sec:improved_decomp_NZZ}.
By combining the decompositions for~$Z_1$,~$Z_2$ and~$Z_1 Z_2$ noise, in~\cref{sec:decompgeneraldephasing} we obtain an improved (though not necessarily optimal) decomposition for the case of general 2-qubit~$Z$-noise. 

Before proceeding, we give a feasible decomposition of any channel of the form: a rotation about the~$P$-axis subject to stochastic~$P$ noise,
\begin{align}
	\label{eq:NPdecomp1}
	\cal{N_P} = \cal{E_P} \circ \cal{R_P}(\theta) = (1-p)\cal{R_P}(\theta) + p \cal{P} \circ \cal{R_P}(\theta) \ ,
\end{align} 
where~$P$ is any Pauli operator.

\begin{lemma}
\label{lem:Protation_feasible_decomp}
	The channel~$\cal{N_P}$ with~$\cal{P}$ any Pauli operator has feasible decomposition
	\begin{align}
		\label{eq:ERP-decomp3}
		\cal{N_P} &= s \cal{R}_P(\varphi) + (1-s) \cal{R}_P(\pi-\varphi) 
	\end{align}
	where 
	\begin{align}
		\label{eq:varphi} \varphi &= \arcsin\left[ (1-2p) \sin\theta \right]
		\quad \text{ and }\quad
		s = \frac{1}{2} \left( 1 + (1-2p) \frac{\cos\theta}{\cos\varphi} \right) \ .
	\end{align}
\end{lemma}

\begin{proof}
	Notice that~$P = i R_P(\pi)$ which together with~\cref{eq:NPdecomp1} gives 
	\begin{align}
		\label{eq:NPdecomp2}
		\cal{N_P} 
		&= (1-p)\cal{R_P}(\theta) + p \cal{R_P}(\theta + \pi) \ .
	\end{align} 
	Consider the Ansatz
	\begin{align}
		\cal{N_P} &= s \cal{R}_P(\varphi) + (1-s) \cal{R}_P(\pi-\varphi) 
	\end{align}
	for a feasible decomposition of~$\cal{N_P}$ (see Fig. \cref{fig:circle-RP} for the intuition behind this Ansatz). We find the parameters~$s,\varphi$ such that the decomposition in~\cref{eq:ERP-decomp3} is feasible. This can be done by computing the transfer matrices associated respectively with~\cref{eq:NPdecomp2,eq:ERP-decomp3} which are
	\begin{align}
		\label{eq:T1} T_1 
		&= (1-p) R_P(\theta) \otimes \overline{R_P(\theta)} + p R_P(\theta+\pi) \otimes \overline{R_P(\theta+\pi)} \ , \\
		\label{eq:T2} T_2 
		&= s R_P(\varphi) \otimes \overline{R_P(\varphi)} + (1-s) R_P(\pi-\varphi) \otimes \overline{R_P(\pi-\varphi)} \ .
	\end{align}
	Solving~$T_1 = T_2$ gives the result. 
\end{proof}

\begin{figure}[h]
    \centering
    \begin{subfigure}[t]{0.485\textwidth}
        \centering
        \includegraphics[height=0.9\linewidth]{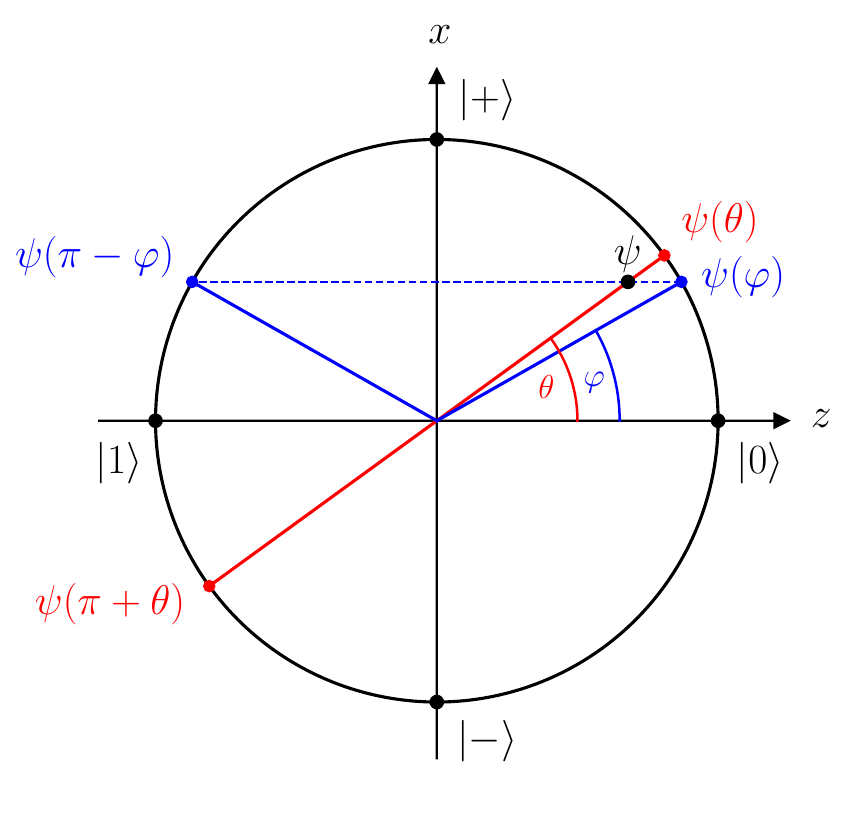}
		\caption{Decomposition of $\psi$ into Gaussian states. 
		The optimal decomposition is the convex combination $\psi = s \psi(\varphi)  + (1-s) \psi(\pi - \varphi)$ of the states $\psi(\varphi) = R_Y(\varphi) \ket{0}$ and $\psi(\pi - \varphi) = R_Y(\pi - \varphi) \ket{0}$ shown in blue, with $s$ and $\varphi$ given in~\cref{eq:varphi}.
		It is the equimagical decomposition: the states $\psi(\varphi)$ and $\psi(\pi - \varphi)$ are equidistant to the set of convex combinations of Gaussian states. The latter is the set of convex combinations of $\ket{0}$ and $\ket{1}$, represented by the line segment $[-1,+1]$ in the $Z$ axis.}
		\label{fig:circle-RP-a}
    \end{subfigure} \hfill
	\begin{subfigure}[t]{0.485\textwidth}
        \centering
        \includegraphics[height=0.9\linewidth]{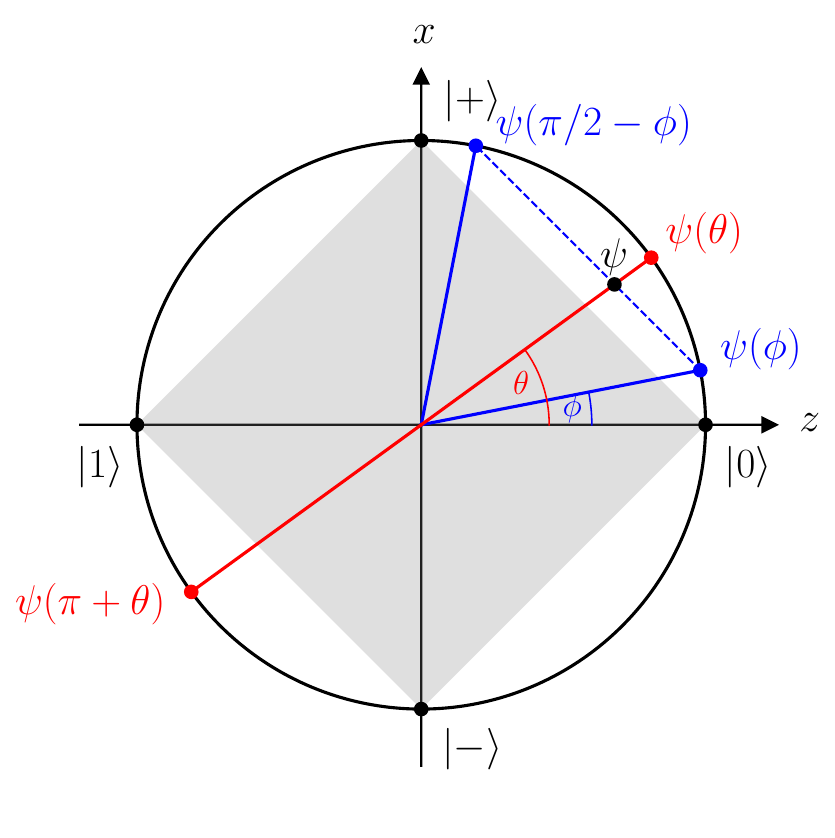}
        \caption{Decomposition of $\psi$ into stabilizer states. The optimal decomposition is a convex combination of the states $\psi(\phi) = R_Y(\phi) \ket{0}$ and $\psi(\pi/2 - \phi) = R_Y(\pi - \phi) \ket{0}$ shown in blue. It is the equimagical decomposition: the states $\psi(\phi)$ and $\psi(\pi/2 - \phi)$ are equidistant to the set of convex combinations of stabilizer states which is represented by the gray shaded region. See~\protect\cite{seddonThesis}.}
		\label{fig:circle-RP-b} 
    \end{subfigure} \flushbottom
    \caption{Bloch sphere illustrating the decomposition of the Choi-equivalent state $\psi = \cal{N}_\cal{Y}(\ketbra{0})$ of the noisy rotation channel $\cal{N}_\cal{Y} = (1-p) \cal{R_Y}(\theta) + p \cal{R_Y}(\theta + \pi)$ for $\theta = \pi/5$ and $p=0.08$. 
	A feasible (non-optimal) decomposition is the convex combination $\psi = (1-p) \psi(\theta)  + p \psi(\pi + \theta)$ of the states $\psi(\theta) = R_Y(\theta) \ket{0}$ and $\psi(\pi + \theta) = R_Y(\pi + \theta) \ket{0}$ shown in red.} 
    \label{fig:circle-RP}
\end{figure}

Recall the definition of an equimagical decomposition in~\cref{sec:magicmonotones}. 
The decomposition in~\cref{eq:ERP-decomp3} is equimagical because~$\xi(R_P(\varphi)) = \xi(R_P(\pi-\varphi))$ and it gives the upper bound for the channel extent
\begin{align}
	\label{eq:ERY-decomposition4-upperbound}
	\Xi(\cal{N_P}) \leq \xi(R_P(\varphi))  
	= 1 + |\sin\varphi| = 1 + (1-2p) |\sin\theta|
\end{align}
which is an improvement to the upper bound in~\cref{eq:xi-noisechannels-ub1} obtained with the initial decomposition~\cref{eq:NPdecomp1}. To build an intuition on how improved equimagical decompositions can be found, it is helpful to use the Choi–Jamiołkowski isomorphism~\cite{Jamiolkowski1972,Choi1975} to move to a geometric picture. We consider without loss of generality the case of single-qubit $Y$-rotations subject to Pauli $Y$ noise, $\cal{N_Y}$.  Recall that the Choi state corresponding to a single-qubit channel $\chan$ may be expressed $\rho_\chan = \id \otimes \cal{\chan} (\op{\psi_+})$, where $\ket{\psi^+} = \left( \ket{00} + \ket{11} \right) / \sqrt{2}$. Now notice that the Bell state $\ket{\psi^+}$ can be transformed into the $\ket{00}$ state by the Gaussian unitary $R_{XY}(\pi/2)$. Moreover this gate commutes with any convex combination of $Y$-rotations $\chan_Y =  \sum_i p_i \cal{R}_Y(\theta_i)$ on the second qubit. It follows that for any such $\chan_Y$, its Choi state $\id \otimes \chan_Y (\op{\psi_+})$ is Gaussian-unitarily equivalent to the product state $ \op{0} \otimes \chan_Y(\op{0})$.  Therefore every $\chan_Y$ may be uniquely identified with a single-qubit state $\chan_Y(\op{0})$, all of which lie in a single plane of the Bloch sphere (\cref{fig:circle-RP-a}). This allows us to reason about which probabilistic mixtures of $Y$-rotations yield the same channel by considering which ensembles of single-qubit pure states correspond to the same density matrix. Moreover, by considering the subset of states that correspond to the convex-Gaussian $Y$-channels (namely states of the form $\rho = (1-p) \op{0}+ p \op{1}$), we gain an intuition of which channels lie closer to the set of free operations. A similar argument holds for the stabilizer case~\cite{seddonThesis}, and we show the analogous situation in~\cref{fig:circle-RP-b}. One striking difference between the two pictures is that the set of free single-qubit states is simply a line segment of the $Z$ axis in the Gaussian case, whereas in the stabilizer case it is an octahedron embedded in the Bloch sphere (or a square when projected to the plane of interest). As a consequence, a much smaller component of $Y$ noise is required to bring a single-qubit channel inside the stabilizer polytope and thus make it free of magic, compared to the fermionic case, where only maximally noisy $Y$ channels become free.

In~\cref{fig:channel_extent} we show~\cref{eq:ERY-decomposition4-upperbound} for different values of noise. 

\begin{figure}[h]
    \centering
    \includegraphics[width=0.5\linewidth]{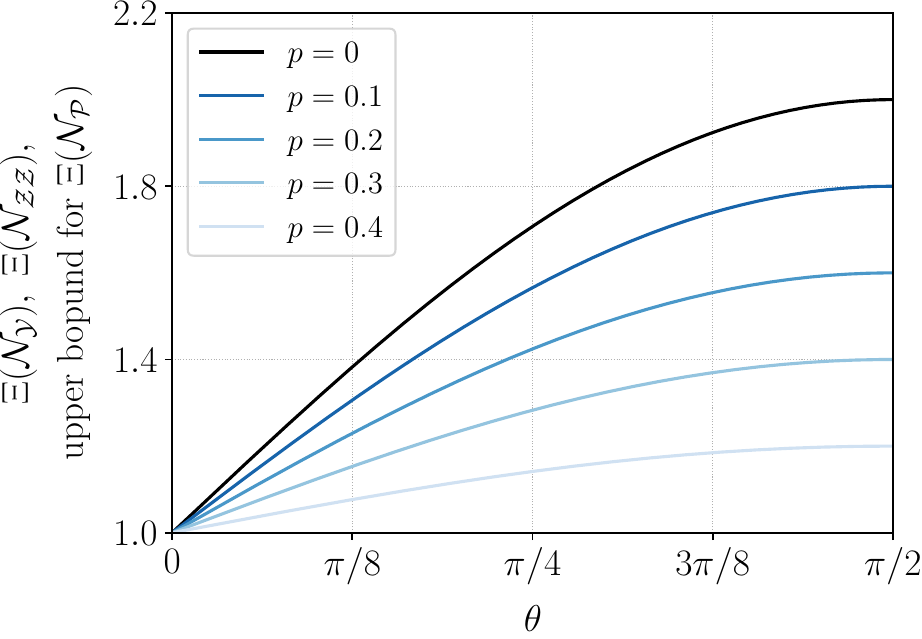}
    \caption{Upper bound for the channel extent~$\Xi(\cal{N_P})$ of the noisy rotation channel~$\cal{N_P} = \cal{E_P}(p) \circ \cal{R_P}(\theta)$ for any Pauli operator~$P$, for different values of the error probability~$p$. This upper bound is tight for the channels~$\cal{N}_Y$ and~$\cal{N}_{ZZ}$, i.e., the figure shows~$\Xi(\cal{N_Y})$ and~$\Xi(\cal{N_{ZZ}})$.}
    \label{fig:channel_extent}
\end{figure}

\subsection{Optimal decompositions for~$Y$ and~$ZZ$ rotations subject to Pauli noise \label{sec:optimal_decomp_NY_NZZ}}

Next we show that the feasible decomposition given in~\cref{lem:Protation_feasible_decomp} is optimal for the channels~$\cal{N_Y}$ and~$\cal{N_{ZZ}}$. This reduces to proving that the upper bound \eqref{eq:ERY-decomposition4-upperbound} is tight. The proof relies on finding an optimal dyadic negativity witness -- distinct in each case -- and using~\cref{lem:sandwich}. See~\cref{fig:channel_extent} for a visualization of~$\Xi(\cal{N_Y})$ and~$\Xi(\cal{N_{ZZ}})$ for different levels of noise. 

Observe that the decompositions of~$\cal{N_Y}$ and~$\cal{N_{ZZ}}$ given in~\cref{lem:NY_optimal,lem:NZZ_optimal}, respectively, are equimagical. This has implications for the runtime of the simulator proposed in~\cref{sec:algorithm} (see~\cref{thm:bitsamplingthm}). Namely, the actual runtime for simulating a circuit including $t$ copies of the channel $\cal{N_{ZZ}}$, i.e., $\cal{N_{ZZ}}^{\otimes t}$, scales linearly with $\Xi(\cal{N_{ZZ}})^t$ for every run of the simulator, and not merely on average. An analogous statement holds for $\cal{N_{Y}}$.

\begin{lemma}
\label{lem:NY_optimal}
The channel~$\cal{N_Y}$ has optimal decomposition 
\begin{align}
	\label{eq:ERY-decomp3}
	\cal{N_Y} &= s \cal{R}_Y(\varphi) + (1-s) \cal{R}_Y(\pi-\varphi) \ ,
\end{align}
with~$\varphi$ and~$s$ given in~\cref{eq:varphi}, and channel extent 
\begin{align}
	\Xi(\cal{N_Y}) = 1 + |\sin\varphi| = 1 + (1-2p)|\sin\theta| \ .
\end{align}
\end{lemma}

\begin{proof}
	By~\cref{lem:Protation_feasible_decomp} the decomposition in~\cref{eq:ERY-decomp3} is feasible. It remains to show optimality. Recall the definition of an (optimal) dyadic negativity witness in~\cref{sec:magicmonotones}.
	Consider the following Ansatz for the optimal witness of~$\Lambda(\cal{N_Y}(\ketbra{0}))$:~$\ketbra{w}$ with~$\ket{w} = \ket{0} + \ket{1}$ for~$\varphi\in[0,\pi]$ and~$\ket{w} = \ket{0} - \ket{1}$ for~$\varphi \in (\pi,2\pi)$. 
	Observe that~$\ketbra{w}$ is a valid witness since it is Hermitian and~$|\langle w | \phi\rangle| = 1$ for any single-qubit Gaussian state~$\ket{\phi}$ (recall that the only single-qubit Gaussian states are~$\ket{0}$ and~$\ket{1}$).
	We have 
	\begin{align}
		\Lambda(\cal{N_Y}(\ketbra{0})) 
		&\geq \tr[ \ketbra{w} \cal{N_Y}(\ketbra{0}) ] \\
		&= s |\langle w | R_Y(\varphi) | 0 \rangle |^2 + (1-s) | \langle w | R_Y(\pi - \varphi) | 0 \rangle |^2 \\
		&= 1 + |\sin\varphi| \\
		&= 1 + (1-2p)|\sin\theta| \ .
	\end{align}
	where in the third equality we used~$R_Y(\theta) \ket{0} = \cos(\theta/2) \ket{0} + \sin(\theta/2)\ket{1}$ and~$| \langle{w}|{R_Y(\theta) 0}\rangle|^2 = 1 + \sin\theta$ for any~$\theta\in[0,2\pi)$. By~\cref{lem:sandwich} this gives
	\begin{align}
		\Xi(\cal{N_Y}) \geq \Lambda(\cal{N_Y}(\ketbra{0}))  \geq 1 + (1-2p)|\sin\theta| \ .
	\end{align}
	This together with~\cref{eq:ERY-decomposition4-upperbound} gives the claim. 
\end{proof}

\begin{lemma}
\label{lem:NZZ_optimal}
	The channel~$\cal{N_{ZZ}}$ has optimal decomposition 
	\begin{align}
		\label{eq:ER{ZZ}-decomp3}
		\cal{N_{ZZ}} &= s \cal{R}_{ZZ}(\varphi) + (1-s) \cal{R}_{ZZ}(\pi-\varphi) \ ,
	\end{align}
	with~$\varphi$ and~$s$ given in~\cref{eq:varphi}, and channel extent 
	\begin{align}
		\label{eq:XiNZZ}
		\Xi(\cal{N_{ZZ}}) = 1 + \sin\varphi = 1 + (1-2p)|\sin\theta| \ .
	\end{align}
\end{lemma}

\begin{proof}
Consider~$\varphi \in [0,\pi]$. Assume that 
\begin{align}
	\ket{w} = \frac{1}{\sqrt{2}} \left(  \ket{0000} + \ket{1111} + i \ket{0011} + i \ket{1100}  \right) 
\end{align}
is a valid witness for the dyadic negativity, i.e.,~$W = \ketbra{w} \in \cal{W}_\Lambda$. For~$\varphi \in (\pi,2\pi)$ consider instead 
\begin{align}
	\ket{w'} = \frac{1}{\sqrt{2}} \left(  \ket{0000} + \ket{1111} - i \ket{0011} - i \ket{1100}  \right) \ .
\end{align}
We have
\begin{align}
	 \cal{N}_{\cal{ZZ}_{23}}(\ketbra{\psi^+}^{\otimes 2}) = s \ketbra{u(\varphi)} + (1-s) \ketbra{u(\pi - \varphi)} 
\end{align}
where~$\cal{N}_{\cal{ZZ}_{23}}$ is the 4-qubit channel which acts non-trivially with~$\cal{N_{ZZ}}$ on qubits 2 and 3 and
\begin{align}
	\ket{u(\vartheta)} = \frac{1}{2} \left( \ket{0000} + \ket{1111} + e^{i\vartheta} \left( \ket{0011} + \ket{1100} \right) \right)  \quad\text{for}\quad \vartheta\in[0,2\pi) \ .
\end{align}
We have
\begin{align}
	\tr[ W  \ketbra{u(\varphi)} ]  &= | \langle{w}|{u(\varphi)}\rangle |^2 = 1 + |\sin\varphi| \ , \\
	\tr[ W  \ketbra{u(\pi - \varphi)} ] &= | \langle{w}|{u(\pi - \varphi)}\rangle |^2 = 1 + |\sin\varphi| \ .
\end{align}
Then
\begin{align}
	\Lambda(\cal{N}_{\cal{ZZ}_{23}}(\ketbra{\psi^+}^{\otimes 2})) 
	&\geq \tr[ W  \cal{N}_{\cal{ZZ}_{23}}(\ketbra{\psi^+}^{\otimes 2})] \\
	&= 1 + |\sin\varphi| \\
	&= 1 + (1-2p) |\sin\theta| \ . 
\end{align}
Then we have
\begin{align}
	\Xi(\cal{N}_{\cal{ZZ}})
	&\ge \Xi(\cal{N}_{\cal{ZZ}_{23}}) \\
	&\ge \Lambda(\cal{N}_{\cal{ZZ}_{23}}(\ketbra{\psi^+}^{\otimes 2}))   
	\qquad \text{ by~\cref{lem:sandwich}} \\
	&\ge 1 + (1-2p) |\sin\theta| \ ,
\end{align}
where the first inequality is due to the fact that a feasible decomposition of the 4-qubit channel~$\cal{N}_{\cal{ZZ}_{23}}$ gives a feasible decomposition of the 2-qubit channel~$\cal{N}_{\cal{ZZ}}$. This together with~\cref{eq:ERY-decomposition4-upperbound} gives the claim. (The proof with~$W' = \ketbra{w'}$ proceeds similarly.)

It remains to show that~$W = \ketbra{w}$ is a valid witness, i.e., that it satisfies the conditions (see~\cref{sec:magicmonotones}):
\begin{enumerate}[1)]
	\item~$W = \ketbra{w}$ is Hermitian (this is immediately observed to be true),
	\item \label{it:witness-condition-2}~$|\langle{w}|{\phi}\rangle |^2 \leq 1$ for any pure Gaussian state~$\ket{\phi}\in\cal{G}_n$.
\end{enumerate}
We proceed to show that~\cref{it:witness-condition-2} is satisfied.
Consider the state~$\ket{v} = \ket{w}/\sqrt{2}$. The state~$\ket{v}$ is Gaussian-unitarily-equivalent to~$\ket{a_8} = (\ket{0000}+\ket{1111})/\sqrt{2}$. Specifically, we have
\begin{align}
	\ket{v} = G_{34} \ket{a_8} \qquad\text{with}\qquad
	G = 
	\begin{pmatrix}
	1/\sqrt{2} & 0 & 0 & i/\sqrt{2} \\
	0 & 1 & 0 & 0 \\
	0 & 0 & 1 & 0 \\
	i/\sqrt{2} & 0 & 0 & 1/\sqrt{2}
	\end{pmatrix} \ ,
\end{align}
where~$G$ is a matchgate, and the subscript~$23$ indicates that~$G$ acts on qubits~$2$ and~$3$. By definition~\ref{def:fidelity} of the Gaussian fidelity, Gaussian-unitarily-equivalent states have the same Gaussian fidelity. From~\cite{reardonsmith2024fermioniclinearopticalextent}, we know that~$F(\ket{a_8}) = 1/2$. Then, for any Gaussian state~$\ket{\phi}\in\cal{G}_n$ we have
\begin{align}
	|\langle{w}|{\phi}\rangle |^2 \leq \max_{\ket{\phi'}\in\cal{G}_n} |\langle{w}|{\phi'}\rangle |^2 
	= 2 F(\ket{v}) = 2 F(\ket{a_8}) = 1 \ .
\end{align}
Hence, \cref{it:witness-condition-2} is satisfied. The proof that~$\ketbra{w'}$ is a valid witness proceeds analogously. 
\end{proof}

\subsection{Rotation about the~$ZZ$-axis with single-qubit~$Z$ noise \label{sec:improved_decomp_NZZ}}

We have thus far considered noisy channels where the Pauli error~$P$ matches the axis of the Pauli rotation~$R_P(\theta)$. Now we focus on a channel where a two-qubit rotation is subject to commuting Pauli noise on just one of the qubits, namely
\begin{align}
	\label{eq:NZZ-Z-decomp0}
	\cal{N}'_\cal{ZZ} 
	&= (1-p) \cal{R}_{\cal{Z}\cal{Z}}(\theta) + p  \cal{Z}_1 \circ  \cal{R}_{\cal{Z} \cal{Z}}(\theta) \ .
\end{align}

While we cannot find a better decomposition as a convex combination of unitary channels, which is the search space for the channel extent optimization program, allowing the use of non-unitary Gaussian channels in the decomposition leads to a lower augmented channel extent. The latter is closely related to the channel extent and to the cost of classical simulating these channels. 

The decomposition in~\cref{eq:NZZ-Z-decomp0} is a feasible decomposition for the augmented channel extent. It gives the upper bound
\begin{align}
	\label{eq:NZZp-extent-bound0}
	\Xiaug(\cal{N}'_\cal{ZZ} ) \leq \Xi(\cal{N}'_\cal{ZZ} ) \leq 1 + |\sin\theta| \ .
\end{align}
We give an alternative feasible decomposition of~$\cal{N}'_\cal{ZZ}$ for the augmented channel extent.

\begin{lemma}
\label{lem:better_decomp_aug_channel_extent}
We show that~$\cal{N}'_\cal{ZZ}$ has feasible decomposition
\begin{align}
	\label{eq:decomposition-E-Znoise}
	\cal{N}'_\cal{ZZ} = (1-2p) \cal{R}_\cal{ZZ}(\theta) + 2 p \cal{E}
\end{align}
where~$\cal{E} (\cdot) = K_0 (\cdot) K_0^\dagger + K_1 (\cdot) K_1^\dagger~$ is a convex-Gaussian channel with
\begin{align}
	K_0 &= \Pi_0 \otimes R_{Z}(\theta) \ ,  \\
	K_1 &= \Pi_1 \otimes R_{Z}(-\theta) \ ,
\end{align}
where~$\Pi_0=\ketbra{0}$ and~$\Pi_1=\ketbra{1}$.
\end{lemma}

Applying the channel~$\cal{E}$ corresponds to performing a computational basis measurement on the first qubit and applying either~$R_{Z}(\theta)$ or~$R_{Z}(-\theta)$ on the second qubit conditioned on the measurement outcome. Such a decomposition signifies that a two-qubit~$ZZ$ rotation on a state where the first qubit is fully dephased can be simulated by a single-qubit rotation on the second qubit with the rotation orientation classically conditioned on the dephased qubit.

\begin{proof}[Proof of~\cref{lem:better_decomp_aug_channel_extent}]
The Choi state of the channel~$\cal{N}'_\cal{ZZ}$ is
\begin{align}
	\cal{N}'_{\cal{ZZ}_{23}} ( \ketbra{\psi^+}^{\otimes 2})
	&= \left( \cal{R}_{\cal{Z}_2\cal{Z}_3}(\theta) \circ \left( (1-p) \cal{I}+ p  \cal{Z}_2  \right) \right) \ketbra{\psi^+}^{\otimes 2}\\
	&=  \cal{R}_{\cal{Z}_2\cal{Z}_3}(\theta)  \left( \left( (1-p) \ketbra{\psi^+} + p \ketbra{\psi^-} \right) \otimes \ketbra{\psi^+}  \right)
\end{align}
where~$\ket{\psi^-} = (\ket{00} - \ket{11})/\sqrt{2}$ and the subscript~$23$ indicates that~$\cal{N}'_{\cal{ZZ}}$ acts on qubits~$2$ and~$3$.
Notice that
\begin{align}
	(1-p) \ketbra{\psi^+} + p \ketbra{\psi^-} = (1-2p) \ketbra{\psi^+} + 2p \left( \frac{1}{2} \ketbra{00} + \frac{1}{2} \ketbra{11}\right) \ .
\end{align}
where the expression in brackets is the Choi state for a single-qubit computational basis measurement, corresponding to Kraus operators~$M_0 = \ketbra{0}$ and~$M_1 = \ketbra{1}$. This indicates that the single-qubit dephasing channel with strength~$p$ is equivalent to measuring out a qubit with probability~$2p$.
Then
\begin{align}
	&\cal{N}'_{\cal{ZZ}_{23}} ( \ketbra{\psi^+}^{\otimes 2}) \\
	&=  (1-2p) \cal{R}_{\cal{Z}_2\cal{Z}_3}(\theta) (\ketbra{\psi^+}^{\otimes 2}) + 2p  \cal{R}_{\cal{Z}_2\cal{Z}_3}(\theta) \left( \frac{1}{2}  \left( \ketbra{00} +\ketbra{11} \right) \otimes  \ketbra{\psi^+} \right) \\
	\label{eq:choifinal} &=  (1-2p) \cal{R}_{\cal{Z}_2\cal{Z}_3}(\theta) (\ketbra{\psi^+}^{\otimes 2}) \\ \nonumber
                             &+ 2p \left( \frac{1}{2} \ketbra{00} \otimes \cal{R}_{\cal{Z}_1}(\theta)  \ketbra{\psi^+} + \frac{1}{2} \ketbra{11} \otimes \cal{R}_{\cal{Z}_1}(-\theta) \ketbra{\psi^+}\right)   \\
	\label{eq:choifinal2} &=  \left((1-2p) \cal{R}_{\cal{Z}_2\cal{Z}_3}(\theta) + 2p \left( \cal{M}_0 \otimes \cal{R}_{\cal{Z}_1}(\theta) + \cal{M}_1 \otimes \cal{R}_{\cal{Z}_1}(-\theta) \right) \right) \ketbra{\psi^+}^{\otimes 2}  \ ,
\end{align}
with~$M_0 = \ketbra{00}$ and~$M_1 = \ketbra{11}$
where we used
\begin{align}
	 R_{Z_2Z_3}(\theta) \left( \ket{aa} \otimes \ket{\psi^+} \right) &=  R_{Z_3}((-1)^a\theta)  \left( \ket{aa} \otimes \ket{\psi^+} \right) 
	 \quad \text{for } a\in\{0,1\} \ .
\end{align}
Notice that the expression in~\cref{eq:choifinal2} is the Choi state for the channel~$\cal{N}'_\cal{ZZ}$ in~\cref{eq:decomposition-E-Znoise}, and it straightforwardly gives the channel decomposition in~\cref{eq:decomposition-E-Znoise}. 
\end{proof}

The decomposition~\cref{eq:decomposition-E-Znoise} leads to the bound
\begin{align}
	\Xiaug(\cal{N}'_\cal{ZZ}) 
	&\leq (1-2p) \xi(\cal{R}_\cal{ZZ}(\theta)) + 2p \\
	&= (1-2p) ( 1 + |\sin\theta| ) + 2p \\ 
	&= 1 + (1-2p) |\sin\theta| \label{eq:xiaugNZ}
\end{align}
which is an improvement over~\cref{eq:NZZp-extent-bound0}.
Notice that the bound is tight for~$p=1/2$, in which case we have~$\Xiaug(\cal{N}'_\cal{ZZ}) = 1$. 

Similar arguments can be used to show that the channel~$\cal{N}_{\cal{ZZ}_{(2)}}$ where noise is applied instead on the second qubit, as decomposition
\begin{align}
	\label{eq:NZZ2}
	\cal{N}_{\cal{ZZ}_{(2)}} = (1-2p) \cal{R}_\cal{ZZ}(\theta) + 2 p \cal{E}'
\end{align}
where~$\cal{E} (\cdot) = K_0' (\cdot) K_0'{}^\dagger + K_1' (\cdot) K_1'{}^\dagger~$ is a convex-Gaussian channel with
\begin{align}
	K_0' &=  R_{Z}(\theta)  \otimes \Pi_0 \  ,  \\
	K_1' &= R_{Z}(-\theta) \otimes \Pi_1 \ .
\end{align}
Applying the same reasoning as in~\cref{eq:xiaugNZ}, the decomposition~\cref{eq:NZZ2} gives the upper bound 
\begin{align}
	\Xiaug(\cal{N}_{\cal{ZZ}_{(2)}}) 
	&\leq 1 + (1-2p) |\sin\theta| \ . \label{eq:xiaugNZ2}
\end{align}

We do not prove optimality of the decomposition \eqref{eq:decomposition-E-Znoise} for the channel~$\cal{N}'_\cal{ZZ}$.  We performed a numerical search over decompositions of the form \eqref{eq:decomposition-E-Znoise}, but did not identify any with a lower augmented channel extent. Still, this does not rule out the existence of better decompositions that may leverage more general convex-Gaussian channels than the computational basis measurement and feedforward strategy considered here. We note that a recent work \cite{Gonzalez-Garcia2025nonGaussianManyBody} has shown that $R_{ZZ}$ rotations subject to dephasing noise becomes convex-Gaussian when the dephasing rate on both qubits is above a certain non-trivial threshold. This is not in conflict with our decomposition here, as we deal with the case where the noise is on exactly one qubit; the argument in \cite{Gonzalez-Garcia2025nonGaussianManyBody} requires that both qubits are sufficiently nosiy. We leave the question of optimality of our decomposition for future work.

It is natural to ask whether it is really necessary to extend the library of free operations with adaptive channels in order to find an improved decomposition of the noisy gate~$\cal{N'_{ZZ}}$ described above. Here we offer evidence that the extension is necessary at least for some convex-unitary channels, by showing that there exists a channel that is not only convex-unitary but convex-Gaussian, and yet cannot be expressed as a convex combination of unitary Gaussian channels.
\begin{lemma}
Consider the channel~$\mathcal{N}'_\cal{ZZ}$ with~$p = 1/2$, so that one of the qubits is fully dephased, i.e.,
\begin{align}
\cal{N}'_\cal{ZZ} 
	&=\frac{1}{2}\left(\cal{R}_{\cal{Z}\cal{Z}}(\theta) + \cal{Z}_1 \circ  \cal{R}_{\cal{Z} \cal{Z}}(\theta) \right).
\end{align}
If~$\theta \in \bb{R}$ is not an integer multiple of~$\pi$, then there exists no decomposition as a convex combination into two-qubit unitary channels~$\cal{N}'_\cal{ZZ}  = \sum_j p_j \curlU_j$ such that the average extent~$\sum_j p_j \xi(U_j)~$ is equal to 1. Conversely,~$\cal{N}'_\cal{ZZ}$ with~$p = 1/2$ can be expressed as a channel composed of Gaussian operations with classical feed-forward, and therefore is convex-Gaussian and has augmented channel extent equal to 1.
\end{lemma}
\begin{proof}
First we note that the Choi state~\cite{Jamiolkowski1972,Choi1975} of a channel diagonal in the computational basis has a restricted form. In particular, the Choi states for two-qubit diagonal channels only have support on the basis~$\{\ket{aa}\otimes \ket{bb} \}$ in our convention, where~$a,b\in\{0,1\}$. Projected onto this 4-dimensional subspace, the Choi matrix for~$\cal{N}'_\cal{ZZ}$ with~$p = 1/2$ can be represented  in block diagonal form
\begin{equation}
\cal{N}'_{\cal{ZZ}_{23}}\left(\ketbra{\psi^+}^{\otimes 2}\right) = \frac{1}{4} \begin{pmatrix} 
1 & e^{-i \theta} &0& 0\\
e^{+i \theta} & 1& 0 &0\\
0 & 0 & 1 & e^{+i \theta} \\
0 & 0 & e^{-i\theta} & 1 
\end{pmatrix}\label{eq:dephasedchoi}
\end{equation}
In \red{\cref{sec:DIAGONAL}} we show that given a convex-unitary channel diagonal in the computational basis, any convex decomposition into unitaries must be in terms of unitary operations that are themselves diagonal. The only diagonal two-qubit unitary operations that are considered free in our simulation scheme are tensor products of single-qubit~$Z$ rotations (this includes the identity channel as the trivial case, and Pauli~$Z\otimes Z$ gates). The Choi matrix for a single-qubit~$Z$ rotation~$R_Z$ (acting on qubit~$2$), expressed in the reduced basis~$\{\ket{00},\ket{11}\}$, has the form
\begin{equation}
\mathcal{R}_{\mathcal{Z}_2} (\ketbra{\psi^+}) = 
\frac{1}{2} \begin{pmatrix}
1 & \alpha \\
\bar{\alpha} & 1 
\end{pmatrix}
\end{equation}
for some phase~$\alpha$. Then in our convention, the Choi state for the tensor product of this gate with \emph{any} other single-qubit unitary~$U$ can be expressed in block form as
\begin{equation}
\mathcal{R}_{\mathcal{Z}_2}(\ketbra{\psi^+})\otimes \rho_U  = \frac{1}{2} \begin{pmatrix} \rho_U & \alpha \rho_U \\
			\bar{\alpha}{\rho_U} & \rho_U
\end{pmatrix} \ .
\end{equation}
It should be clear that any convex combination of such unitary channels must have the form
\begin{equation}
\rho' = \begin{pmatrix}
A & B \\
C & A
\end{pmatrix} 
\end{equation}
for some~$2\times 2$ matrices~$A$,~$B$ and~$C$. In particular the diagonal blocks~$A$ are constrained to be equal. But inspecting~\cref{eq:dephasedchoi} we see that for the channel we wish to decompose, the blocks on the diagonal are equal only for~$\theta =k \pi$, where~$k$ is an integer. Therefore for a general~$\theta$ no such convex combination over Gaussian unitary channels is possible.

Conversely, in~\cref{eq:decomposition-E-Znoise} setting~$p=1/2$ explicitly gives a valid decomposition where the unitary operation applied after either outcome of the single-qubit measurement can be efficiently simulated.
\end{proof}

\subsection{Application: improved decomposition for generalized two-qubit~$Z$ noise \label{sec:decompgeneraldephasing}}

By combining the decompositions for dephasing~$Z_1$,~$Z_2$, and~$Z_1 Z_2$ noise, we can obtain improvements for more general dephasing channels. We give an example here for the~$ZZ$ rotation channel~$\cal{R_{ZZ}}(\theta)$ subject to the dephasing noise channel 
\begin{align}
	\cal{E}(p) = (1-p) \cal{I} + \frac{p}{3}( \cal{Z}_1 + \cal{Z}_2 + \cal{Z}_1\cal{Z}_2 ) \ ,
	\label{eq:channelNgendephasing}
\end{align}
this is, 
\begin{align}
	\cal{N}' &= \cal{E}(p) \circ \cal{R_{ZZ}}(\theta)  \label{eq:Nprime00} \\
	&= (1-p) \cal{R_{ZZ}}+ \frac{p}{3}( \cal{Z}_1 \circ \cal{R_{ZZ}} + \cal{Z}_2 \circ \cal{R_{ZZ}} + \cal{Z}_1\cal{Z}_2 \circ\cal{R_{ZZ}} ) \ . \label{eq:Nprime0}
\end{align}
Notice that the immediate decomposition in~\cref{eq:Nprime0} gives the upper bound for the channel extent
\begin{align}
	\Xi(\cal{N}') 
	&\leq (1-p) \xi(R_{ZZ}) + \frac{p}{3} \left( \xi(Z_1R_{ZZ}) +  \xi(Z_2R_{ZZ}) +  \xi(Z_1Z_2R_{ZZ})\right) \\
	&= \xi(R_{ZZ})
\end{align}
where we used the fact that~$Z_1$ and~$Z_2$ are Gaussian unitaries, and consequently~$\xi(Z_1R_{ZZ}) = \xi(Z_2R_{ZZ}) = \xi(Z_1Z_2R_{ZZ}) = \xi(R_{ZZ})$.

The following decomposition 
\begin{align}
	\cal{N}'
	&= \frac{1}{3}\left(\cal{N}_{\cal{ZZ}_{(1)}} + \cal{N}_{\cal{ZZ}_{(2)}} + \cal{N}_{\cal{ZZ}} \right) \ ,
\end{align} 
leads to an improved lower bound for the channel extent. We utilize the optimal decomposition of~$\cal{N}_{\cal{ZZ}}$ obtained in~\cref{lem:NZZ_optimal}. The resulting decomposition has channel extent
\begin{align}
	\Xi(\cal{N}') 
	&\leq \frac{1}{3}( \Xi(\cal{N}_{\cal{ZZ}_{(1)}}) + \Xi(\cal{N}_{\cal{ZZ}_{(2)}})  + \Xi(\cal{N}_{\cal{ZZ}}) )  \\
	&\leq \frac{2}{3}\left( 1 + |\sin\theta| \right) + \frac{1}{3}\left(1 + (1-2p)|\sin\theta|\right) \\
	&= 1 + (1-2/3p)|\sin\theta| \ . \label{eq:Xigeneraldephasing}
\end{align} 

If we augment the set of allowed Gaussian channels, including Gaussian measurements and feedforward, we may utilize the improved decompositions for~$\cal{N}_{\cal{ZZ}_{(1)}}$ and~$\cal{N}_{\cal{ZZ}_{(2)}}$ given respectively in~\cref{lem:better_decomp_aug_channel_extent,eq:NZZ2} and obtain an improved  decomposition for~$\cal{N}'$ with augmented channel extent upper bounded by 
\begin{align}
	\Xi^*(\cal{N}') 
	&\leq \frac{1}{3}( \Xi^*(\cal{N}_{\cal{ZZ}_{(1)}}) + \Xi^*(\cal{N}_{\cal{ZZ}_{(2)}})  + \Xi^*(\cal{N}_{\cal{ZZ}})  )  \\
	&\leq \frac{1}{3}( \Xi^*(\cal{N}_{\cal{ZZ}_{(1)}}) + \Xi^*(\cal{N}_{\cal{ZZ}_{(2)}})  + \Xi(\cal{N}_{\cal{ZZ}})  )  \\
	&\leq 3 \cdot \frac{1}{3}(1 + (1-2p)|\sin\theta|) \\
	&= 1 + (1-2p)|\sin\theta| \ ,
\end{align} 
where we used~\cref{eq:xiaugNZ} and~\cref{eq:XiNZZ} from~\cref{lem:NZZ_optimal}.

Finally we note that the channel \cref{eq:Nprime0} can alternatively be written in product form
\begin{equation}
\cal{N}' = ((1-q)\cal{I} + q \cal{Z}_1\cal{Z}_2)  \circ ((1-q)\cal{I} + q \cal{Z}_1)\circ ((1-q)\cal{I} + q \cal{Z}_2)\circ \cal{R}_{\cal{ZZ}}
\end{equation}
for some $q$. We compare again with the result of \cite{Gonzalez-Garcia2025nonGaussianManyBody}, which implies that if dephasing probability $q$ is sufficiently large, the channel becomes convex-Gaussian. This suggests that in our augmented extent scheme, admitting a suitable class of Gaussian channel should yield improved decompositions. We leave such improvements and their optimality for future work.

\section{Improved and optimal decompositions for the fermionic nonlinearity \label{sec:FNL}}

The fermionic nonlinearity, mentioned in~\cref{sec:priorwork}, is a magic monotone for channels. It was introduced in~\cite{Hakkaku_2022} and it is formally defined as follows. 

\begin{definition}[Fermionic nonlinearity, see Eq.~3 in~\cite{Hakkaku_2022}]
	\label{def:FNL}
	The fermionic nonlinearity of a fermionic channel~$\cal{E}$ is
	\begin{align}
		\FNL(\cal{E}) = \min \left\{ \|q\|_1 : \cal{E}(\cdot) = \sum_j q_j L_j (\cdot) R_j ; L_j, R_j \in G_n , q_j \in \bb{R} \right\} \ .
	\end{align}
\end{definition}

In~\cref{sec:Udecomp} we gave optimal decompositions for any two-qubit fermionic gate with respect to the unitary extent. In~\cref{sec:optimalFNLdecomp}, we show these provide optimal decompositions for the respective unitary channels with respect to the fermionic non-linearity (see~\cref{lem:FNLoptimal}). 
In~\cref{sec:noisydecomps} we obtained optimal (and improved) decompositions for rotation channels subject to dephasing noise. In~\cref{sec:improvedFNLdecomp_noisychannel}, we show these give a decomposition of the same channel with reduced $L^1$-norm $\|q\|_1$. For certain parameter regimes, our decompositions are an improvement over the decompositions obtained in~\cite{Hakkaku_2022}, thus providing a tighter bound on the fermionic nonlinearity.
\subsection{Optimal decompositions of two-qubit unitary channels \label{sec:optimalFNLdecomp}}

The fermionic nonlinearity is closely related to the dyadic negativity given in~\cref{def:dyadic-negativity}. It can be thought of as the dyadic negativity for channels, instead of density operators, and is therefore also closely related to the dyadic channel negativity introduced in the stabilizer setting by~\cite{seddonThesis}. The fermionic nonlinearity of a channel~$\cal{E}$ upper bounds the dyadic negativity of the state~$\cal{E}(\ketbra{g})$ for any Gaussian state~$\ket{g}\in\cal{G}_n$, i.e.,
\begin{align}
	\label{eq:FNL_bound_Lambda}
	\FNL( \cal{E}) \geq \Lambda(\cal{E}(\ketbra{g})) \quad\text{ for any }\quad \ket{g}\in\cal{G}_n \ .
\end{align}

In the following lemma, we argue that optimal unitary extent decompositions for a gate~$U$ give optimal fermionic nonlinearity decompositions for the corresponding unitary channel~$\cal{U}(\cdot)=U (\cdot) U^\dagger$. 
As a consequence, we have an optimal fermionic nonlinearity decomposition of any fermionic two-qubit unitary channel~$\cal{U}$, which is built using the unitary extent optimal decompositions of the respective gate~$U$ obtained in~\cref{lem:extent2qubitgate}. Specifically, this applies to the channels~$\cal{C}(\theta),\cal{R_{ZZ}}(\theta)$ and~$\cal{SWAP}$. Their optimal fermionic nonlinearity decompositions are built using the unitary-extent-optimal decompositions presented in~\cref{tab:optimaldecomp}.

\begin{lemma}
\label{lem:FNLoptimal}
Let~$U$ be a two-qubit unitary with optimal unitary extent decomposition
\begin{align}
	U = \sum_{j=2}^\chi c_j V_j \ ,
\end{align} 
i.e.,~$\xi(U)=\|c\|_1^2$. Let~$\ket{u} = (I \otimes U \otimes I)\ket{\psi^+}^{\otimes 2}$
be the Choi state~\cite{Jamiolkowski1972,Choi1975} of~$U$. Assume that 
\begin{align}
	\label{eq:FNLassump} \xi(\ket{u}) = \xi(U) \ .	
\end{align}
Then the channel~$\cal{U}(\cdot)=U (\cdot) U^\dagger$ has optimal fermionic nonlinearity decomposition
\begin{align}
	\label{eq:FNLoptimaldecomp}
	\cal{U}(\cdot) = \sum_{j,k=1}^\chi c_j \overline{c_k} V_j (\cdot) V_k^\dagger \ .
\end{align}
Moreover,~$\FNL(\cal{U}) = \xi(U)$. 
\end{lemma}

\begin{proof}
	The decomposition in~\cref{eq:FNLoptimaldecomp} is a feasible decomposition of the channel~$\cal{U}$. It gives the upper bound 
	\begin{align}
		\label{eq:FNLaux1}
		\FNL(\cal{U}) \leq \sum_{j,k=1}^\chi |c_j\overline{c_k}| = \|c\|_1^2 = \xi(U) \ .
	\end{align}
	We now prove the converse bound. Consider~$\ket{u} = (I \otimes U \otimes I) \ket{g}$ with~$\ket{g}\in G_4$ a Gaussian state. By~\cite[Lemma 1]{Seddon_2021} (see also~\cite{Regula_2018}) we have
	\begin{align}
		\label{eq:FNLaux2}
		\Lambda(\ketbra{u}) = \xi( \ket{u} ) \ .
	\end{align}
	This result was proven for superpositions of stabilizer states, but it extends more generally, including to superpositions of Gaussian states.
	Combining~\cref{eq:FNLaux2} with~\cref{eq:FNL_bound_Lambda} gives
	\begin{align}
		\FNL(\cal{U}) \geq \FNL(\cal{I} \otimes \cal{U} \otimes \cal{I}) \geq \Lambda((I\otimes U \otimes I)\ketbra{g} (I\otimes U^\dagger \otimes I)) = \Lambda(\ketbra{u}) = \xi(\ket{u}) \ ,
	\end{align}
	where the first inequality is due to the fact that a feasible decomposition of~$\cal{U}$ gives a feasible decomposition of~$\cal{I} \otimes \cal{U} \otimes \cal{I}$. 
	We chose~$\ket{g} = \ket{\psi^+}^{\otimes 2}$. Then, from the assumption in~\cref{eq:FNLassump}, we have
	\begin{align}
		\FNL(\cal{U}) \geq \xi(\ket{u}) = \xi(U)
	\end{align}
	which together with~\cref{eq:FNLaux1} gives the claim. 
\end{proof}

In~\cite{Hakkaku_2022} the authors give upper bounds for the fermionic nonlinearity of the rotation channel~$\cal{R_{ZZ}}(2 \theta)$. These bounds are obtained by numerically searching for feasible decompositions using a discrete subset of fermionic Gaussian gates, rather than the full continuous set of all such gates. They find (see~\cite{Mitarai_2021} and~\cite[Section 6]{reardonsmith2024improved})
\begin{align}
	\label{eq:HakFNLbound}
	\FNL(\cal{R_{ZZ}}(2 \theta)) \leq 1 + 2|\sin(2\theta)| \ .
\end{align}
Here, by~\cref{lem:FNLoptimal} and the optimal decomposition of~$\cal{R_{ZZ}}$ given in~\cref{tab:optimaldecomp} we have established that 
\begin{align}
	\FNL(\cal{R_{ZZ}}(2 \theta)) = \xi(R_{ZZ}(2 \theta)) = 1 + |\sin(2\theta)| \ ,
\end{align}
which improves upon the bound in~\cref{eq:HakFNLbound}. The resulting performance gain in the quasiproba\-bi\-li\-ty-based simulation of this channel was previously observed in~\cite{reardonsmith2024improved} (see Section~6 therein). Furthermore, by providing analytical expressions for optimal fermionic nonlinearity decompositions, we eliminate the need to solve the linear program in~\cref{def:FNL}, which in practice relies on a discretization of the feasible space, as done in~\cref{eq:HakFNLbound}, a choice that can significantly influence the optimization outcome.

\begin{figure}[h]
    \centering
    \includegraphics[width=0.5\linewidth]{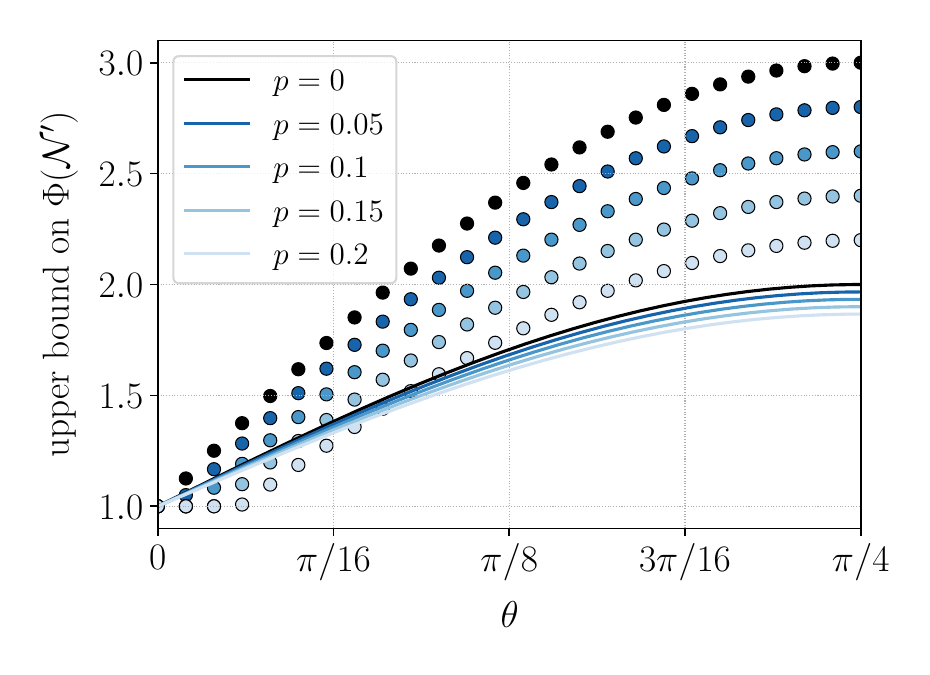}
    \caption{The full lines correspond to~$1 + (1-2/3p)|\sin (2 \theta)|$, the upper bound obtained in~\cref{eq:FNLNprime} for the fermionic nonlinearity~$\FNL(\cal{N}')$ of the channel~$\cal{N}' = \cal{E}(p) \circ \cal{R_{ZZ}}(\theta)$ defined in~\cref{eq:Nprime00}, for different values of the error probability~$p$. The circles correspond to the upper bound for~$\FNL(\cal{N}')$ obtained in~\protect\cite{Hakkaku_2022} (see Fig.~2 therein) -- the data was kindly provided by the authors. }
    \label{fig:channel_extent}
\end{figure}

\subsection{Improved decomposition of the dephased~$ZZ$ rotation channel \label{sec:improvedFNLdecomp_noisychannel}}

In~\cite{Hakkaku_2022} the authors analyse the effect of the dephasing channel~$\cal{N}$ defined in~\cref{eq:channelNgendephasing} in reducing the fermionic nonlinearity and thus the computational cost in quasiprobability-based simulation. Specifically, they study the fermionic nonlinearity~$\Phi(\cal{N}')$ of the~$ZZ$ rotation channel~$\cal{N}' = \cal{N} \circ \cal{R_{ZZ}}(\theta)$ subject to dephasing noise defined in~\cref{eq:Nprime00}. The authors numerically obtain feasible decompositions over a discretized set of allowed operations to upper bound~$\Phi(\mathcal{N}')$. The resulting values are shown in~\cite[Figure 2]{Hakkaku_2022}.

We obtain improved bounds on~$\Phi(\mathcal{N}')$ for some parameter regimes. A feasible decomposition with respect to the channel extent gives a feasible decomposition with respect to the fermionic nonlinearity. Then, the decomposition of~$\mathcal{N}'$ presented in~\cref{sec:decompgeneraldephasing}, with associated channel extent~$\Xi(\mathcal{N}')$ defined in~\cref{eq:Xigeneraldephasing}, constitutes a feasible fermionic nonlinearity decomposition. This yields the upper bound
\begin{align}
	\FNL(\cal{N}') \leq \Xi(\cal{N}') 
	&= 1 + (1-2/3p)|\sin (2 \theta)| \ . \label{eq:FNLNprime}
\end{align}
In~\cref{fig:channel_extent} we compare these values to those presented in~\cite[Figure 2]{Hakkaku_2022}. Notice that for~$p > 0$ and small angles~$\theta$ the decompositions obtained in~\cite{Hakkaku_2022} achieve a tighter bound for the fermionic nonlinearity, and for certain points the channel seems to be convex-Gaussian (i.e., it can be written as a convex combination of unitary Gaussian channels) -- this again may be compared with the local dephasing noise study in \cite{Gonzalez-Garcia2025nonGaussianManyBody}, and suggests our decomposition in this case may not be optimal with respect to extent. We leave further improvements of our decompositions for future work.

\section{Classically simulated sampling from non-Gaussian circuits\label{sec:algorithm}}

Here we briefly review algorithmic primitives for simulating circuits with non-free states and operations common to the stabilizer and Gaussian setting. 
We describe how we can simulate sampling for non-Gaussian circuits subject to convex-unitary noise channels. This is an adaptation of analogous algorithms developed in the context of stabilizer-based simulation. Reference~\cite{Seddon_2021} gave an algorithm for sampling from Clifford circuits with noisy magic state inputs, while~\cite{seddonThesis} sketched how this could be extended to circuits composed of non-stabilizer convex-unitary channels. The structure of the algorithm described below is essentially identical to that in the stabilizer setting. The difference in the Gaussian setting is only in which channels are non-free, how they are decomposed, and polynomial factors in the simulation cost.

The simulation task is the following.
\begin{definition}[Bit-string sampling task (approximate)]\label{def:bitsampling_approx}
Given a positive integer~$w\leq n$ and a description of a~$n$-qubit channel 
\begin{equation}
	\label{eq:channelseq}
\chan = \chan_T \circ \chan_{T-1} \circ \ldots \circ \chan_1 \ ,
\end{equation}
where each~$\chan_t$ is a convex-unitary channel with known decomposition, output a bit-string~$x$ of length~$w$ drawn with probability~$p_\text{sim}(x)$ such that the distribution~$p_\text{sim}$ is~$\epsilon$-close in~$L^1$-distance to the distribution
\begin{equation}
p(x) = \Tr(\Pi_x \chan(\op{0^n}))
\end{equation}
where~$\Pi_x = \op{x}{x} \otimes I_{n-w}$, with~$\ket{x}$ is a computational basis state on the first~$w$ qubits, and~$I_{n-w}$ is the identity on the unmeasured subsystem.
\end{definition}
Note that this simulation task subsumes several related tasks as specific cases:
\begin{enumerate}[1)]
\item we can draw from the distribution~$p(x) = \Tr(\Pi_x \chan(\op{g}))$ for any Gaussian state~$\ket{g}$ by inserting a suitable matchgate circuit before the first channel;
\item we can draw from the distribution~$p(x) = \Tr(\Pi_x \chan(\rho))$ for a non-Gaussian state~$\rho$ with some additional runtime overhead provided we have a known decomposition of a convex-unitary channel that prepares the state~$\rho$ from~$\op{0^n}$;
\item noiseless circuits are included since unitary channels are the extreme points of the set of convex-unitary channels.
\end{enumerate}
We will show that, if the optimal decomposition is known, the expected simulation runtime scales with the product of the convex-unitary channel extent for each channel in the sequence \eqref{eq:channelseq}. However, since in general we do not always know the optimal decomposition of a convex-unitary channel, we need a notion of cost with respect to a given decomposition.

\begin{definition}[Channel decomposition cost]\label{def:oraclecost}
Let~$A =\{\chan_t\}_{t=1}^T$ be a set of convex-unitary channels. 
Let~$\mathcal{D}$ be an oracle that, given an index~$t$ (or other description of the channel~$\chan_t$), returns with probability~$p_j$ a description
\begin{align}
	\label{eq:DUt}
	\cal{D}(U_j) = \bb{U}_j = \{(c_{j,k},R_{j,k})\}_k
\end{align}
of a decomposition of a unitary operator
\begin{equation} 
U_j = \sum_k c_{j,k} G_{j,k}
\end{equation}
where every~$G_{j,k}$ is Gaussian, such that~$\chan_t = \sum_j p_j \mathcal{U}_j$. The matrix~$R_{j,k} \in O(2n)$ is the matrix that describes the Gaussian unitary~$G_{j,k}$ according to~\cref{eq:UR_R_relation}.
Then the cost~$E$ associated with the channel~$\chan_t$ with respect to the oracle~$\mathcal{D}$ is
\begin{equation}
E(\mathcal{D},\mathcal{\chan}_t) = \sum_j p_j \left(\sum_k \abs{c_{j,k}} \right)^2 \ .
\end{equation}
An oracle is called equimagical over the channel set~$A$ if, for any~$\chan_t \in A$, the~$L^1$-norm of the decomposition returned by the oracle does not depend on the unitary operator~$U_j$. That is, $\sum_k \abs{c_{j,k}} = \sum_{k} \abs{c_{j',k}}$ for any pair~$(j,j')$, so that
\begin{equation}
E(\mathcal{D},\mathcal{\chan}_t) = \left(\sum_k \abs{c_{j,k}} \right)^2 \quad \text{ for any } j \ .
\end{equation}
An oracle is called optimal over~$A$ if~$E(\mathcal{D},\mathcal{\chan}_t) = \Xi(\chan_t)$ for all~$\chan_t \in A$.
\end{definition}

The oracle in~\cref{def:oraclecost} gives a way to formalize the notion of a recipe to construct a good decomposition of a given class of channels. This allows us to make the following statement about the  classical simulation cost.

\begin{theorem}\label{thm:bitsamplingthm}
Given an oracle~$\mathcal{D}$ for the channel set~$A = \{ \chan_t \}_{t=1}^T$, there exists a classical algorithm such that the runtime~$\tau$ to perform the task specified in~\cref{def:bitsampling_approx} has expectation
\begin{equation}
\mathbb{E}(\tau) = O\left(\poly\left(n,T,w,\frac{1}{\epsilon}\right)\log(p^{-1}_\text{sim})\prod_{t=1}^T E(\mathcal{D},\chan_t)\right),
\end{equation}
where the average is over independent runs of the algorithm.
If the oracle~$\mathcal{D}$ is equimagical over the set~$A$, then the expression gives the actual runtime.
If the oracle is optimal over~$A$, then 
\begin{align}
	\mathbb{E}(\tau) = O\left(\poly\left(n,T,w,\frac{1}{\epsilon}\right)\log(p^{-1}_\text{sim})\prod_{t=1}^T \Xi(\chan_t)\right) \ .
\end{align}
\end{theorem}

This result motivates the study of channel decompositions in~\cref{sec:Udecomp,sec:noisydecomps}. In our study of noisy channels, we are guided by the intuition that in many cases we expect the simulation cost factor per channel to smoothly approach unity as the noise rate increases. Using the cost model above, even a small reduction in cost at the level of individual channels can lead to significant savings in runtime for intermediate-sized systems.  
Since~\cref{thm:bitsamplingthm} follows by adapting algorithms already known from the stabilizer setting~\cite{bravyiComplexityQuantumImpurity2017a,Bravyi_2019,Seddon_2021,seddonThesis} to the fermionic setting, we do not give a full proof here, instead sketching the argument. For a full technical proof see~\cref{sec:convex_unitary_channel_simulation}. A key component in ensuring that the cost is linear rather than quadratic in the cost function $E(\mathcal{D},\chan_t)$ is the following result from~\cite{bravyiComplexityQuantumImpurity2017a}.
\begin{lemma}[Fermionic fast norm estimation~\cite{bravyiComplexityQuantumImpurity2017a}]
\label{lem:fastnorm}
Given a vector~$\ket{\Psi} = \sum_{a=1}^\chi c_a \ket{g_a}$ where~$g_a$ are Gaussian states, there is a procedure that takes as input the vector~$\Psi$, precision~$\epsilon > 0$, and failure probability~$p_f > 0$ and outputs a value~$y$ such that, with probability at least~$1-p_f$, the following inequality holds
\begin{equation}
(1-\epsilon) \norm{\Psi}^2 \leq y \leq (1+\epsilon) \norm{\Psi}^2 \ .
\end{equation}
The procedure has runtime~$O(n^{7/2} \epsilon^{-2} p_f^{-1} \chi)$.
\end{lemma}

The procedure in~\cref{lem:fastnorm} can be slightly modified by using a median of means argument as done in~\cite[Section 4.1]{Bravyi_2019} (see also e.g.~\cite[Lemma 10]{reardonsmith2024improved}) for a similar fast norm algorithm for linear combinations of stabilizer states. This gives an improvement in the runtime from~$O(n^{7/2} \epsilon^{-2} p_f^{-1} \chi)$ to~$O(n^{7/2}  \epsilon^{-2} \log(1/p_f) \chi)$. We refer to this procedure as \textsc{FastNorm}.

Including $\textsc{FastNorm}$, the bit-string sampling algorithm relies on four main subroutines, as follows.
\begin{itemize}
\item \textsc{SparsifyCircuit}: Given an integer $k$ and a description of a unitary circuit as a list of $T$ unitary operators with known decomposition, returns a list of $k$ Gaussian circuits of length $T$, equipped with phase and normalization factor, in time $O(kT)$. (Check~\cref{algo:sparsifyCircuit} in~\cref{sec:sampling_simulator_details} for the detailed algorithm.)
\item \textsc{GaussianEvolve}: The subroutine called ``evolve'' in~\cite{Dias_2024}, this takes as input a description of a Gaussian state $\ket{\phi}$ and Gaussian unitary operation $V$, and returns a phase-sensitive description of the updated state $V\ket{\phi}$ in time $O(n^3)$.
\item \textsc{EvolveCircuit}: Given a list of $k$ Gaussian circuits $V_j$, each of gate count $T$, coefficients $c_j$ and initial Gaussian state $\ket{g}$, return a description of the final state $\sum_{j=1}^k c_j V_j \ket{g}$ as a superposition of Gaussian states, in time $O(k T n^3)$. Technically this is achieved by repeated calls to \textsc{GaussianEvolve}. (Check~\cref{algo:evolveCircuit} in~\cref{sec:sampling_simulator_details} for the detailed algorithm.)
\item \textsc{FastNorm}: Given a vector $\ket{\Psi}$ described as a linear combination of $k$ Gaussian states, returns an  estimate of $\norm{\Psi}^2$ up to multiplicative precision $\epsilon$ with probability $1-p_f$, in time $O(n^{7/2}  \epsilon^{-2} \log(1/p_f) k)$.
\end{itemize}

The main algorithm then proceeds as follows.
\begin{enumerate}
\item For each convex-unitary channel $\chan_t $ in the composed circuit $\chan = \chan_T \circ \chan_{T-1} \circ \ldots \circ \chan_1$, query the oracle $D$ to obtain at random a unitary evolution $U_t$, thereby selecting a unitary trajectory through the noisy circuit, $U = \prod_{t=1}^T U_t$. Suppose for simplicity that the optimal decomposition of each $U_t$ is known. Then the cost for each unitary gate is $\xi(U_t)$, and the total cost for this run of the algorithm will be $E = \prod_{t=1}^T \xi(U_t)$.
\item Set $k = \left\lceil 4 E / \delta \right\rceil$, then call \textsc{SparsifyCircuit} for the unitary trajectory $U$ to obtain a sparsified circuit $U_{\text{sparse}} = \sum_{j=1}^k c_k V_k$, where each $V_k$ is a Gaussian unitary circuit, where $\delta$ is a parameter we call the sparsification error. \label{step:sparsifycircstep} 
\item  Call \textsc{EvolveCircuit} to obtain a sparsified approximation of the final state for the selected unitary trajectory, $\ket{\Omega} = U_{\text{sparse}} \ket{0^n} \approx U \ket{0^n}$.
\item For $i=1$ to $w$: \label{step:FastNormEst}
\begin{enumerate}
\item Call \textsc{FastNorm} to estimate $\Pr(x_i=0) = \norm{\Pi_{i}\ket{\Omega}}^2/\norm{\ket{\Omega}}^2$, the probability of obtaining the $+1$ outcome for a $Z$ measurement on the $i$-th qubit. \
\item Randomly select $x_i=0$ with probability $\Pr(x_i=0) $, otherwise $x_i=1$.
\item Make the update
\begin{equation}
\ket{\Omega} \leftarrow \frac{(I + (-1)^{x_i} Z_i)\ket{\Omega}}{\norm{(I + (-1)^{x_i} Z_i)\ket{\Omega}}} \ . 
\end{equation}
\end{enumerate}
\item Output the string $x_1, \ldots, x_w$.
\end{enumerate}
It should be clear that averaging over all possible unitary trajectories, and ignoring the sparsification step, the ensemble of final pure states generated by the algorithm is precisely equivalent to the target density operator $\rho = \chan(\op{0^n})$. Then if the calls to \textsc{FastNorm} were replaced by exact computation of $\Tr[\Pi_x U \op{0^n} U^\dagger] $, the distribution of samples returned by the algorithm would exactly replicate the output distribution of the noisy quantum device. 

For the approximate algorithm, by adapting the arguments given in~\cite{Seddon_2021,seddonThesis}, it can be shown that by choosing $k$ as described above (subject to some technical caveats, see~\cref{sec:SPARSE_delta}), the ensemble of sparsified pure states $\ket{\Omega}$ generated by the algorithm, $\rho_S = \sum_{\Omega} \Pr(\Omega) \op{\Omega}/\bra{\Omega}\ket{\Omega}$, is $\delta$-close to the target noisy final state,
\begin{equation}
\norm{\rho_1 - \chan(\op{0^n})}_1 \leq \delta + O(\delta^2) \ .
\end{equation}
Then, by~\cref{lem:fastnorm}, which guarantees the correctness of the \textsc{FastNorm} procedure~\cite{bravyiComplexityQuantumImpurity2017a}, it follows that the distribution of bit strings output by the algorithm is close in $L^1$-norm to the exact distribution $p(x) = \Tr(\Pi_x \chan(\op{0^n}))$. 

For circuits with significant magic, the runtime of the algorithm will be dominated by the $O(k)$ factor in each of the steps \ref{step:sparsifycircstep}-\ref{step:FastNormEst}, since this grows exponentially with the number of magic gates. Since we set $k \propto E = \prod_{t=1}^T \xi(U_t)$, we have the expected linear dependence on the (product of) unitary extent. Then by~\cref{def:oraclecost}, the average runtime of the algorithm is linear in the oracle cost $\prod_{t=1}^T E(\mathcal{D},\chan_t)$. In the case where the oracle is optimal over the channels $\chan_t$ involved in the circuit, the cost will be linear in $\prod_{t=1}^T \Xi(\chan_t)$.  Factors polynomial in the number of qubits $n$, the string length $w$, the circuit length $T$, and the precision $\epsilon^{-1}$ arise from the calls to \textsc{GaussianEvolve} and \textsc{FastNorm}. See~\cref{sec:sampling_simulator_details} for the full technical proof for this setting.

In sections~\cref{sec:Udecomp} and~\cref{sec:noisydecomps} we showed that an optimal and efficient oracle exists for many important two-qubit gates and channels by giving simple analytic decompositions. In particular, in~\cref{sec:noisydecomps}  we saw that noise often reduces the overhead of simulating a circuit in this context. For certain noisy gates where the optimal decomposition is not known, we gave decompositions with cost reduced relative to the noiseless case.

For specific noisy channels we were not able to find an improved decomposition in terms of convex-unitary channels alone. Instead we obtained improved average cost by augmenting the set of convex-unitary channels with adaptive Gaussian channels. For this reason, in~\cref{sec:SPARSE_adapt} we describe a novel sparsification scheme for circuits where the sequence of gates is conditioned on mid-circuit measurements, and prove that the sparsified maps output from this procedure satisfy a generalized version of the ensemble sampling lemma introduced the non-adaptive case~\cite{Seddon_2021,seddonThesis}. Then in~\cref{sec:modified_simulator_sparsifications} we sketch how the bit-string sampling algorithm can be extended to such circuits.  In these cases we must contend with the fact that the mid-circuit measurement probabilities are state-dependent and not known in advance, and must be estimated on the fly by the simulator.

\section{Approximating non-Gaussian channels with sparsified circuits\label{sec:SPARSE}}

In this section we make the link between the Gaussian extent of a state or operation and its simulation cost, by reviewing the sparsification techniques of~\cite{Bravyi_2019,Seddon_2021} in~\cref{sec:SPARSE_delta} and discussing their application to fermionic states before showing how sparsification can be extended to more general fermionic circuits in~\cref{sec:SPARSE_adapt}. This connection motivates the programme of work reported in prior sections, where we gave optimal and improved decompositions of non-Gaussian states, unitary gates and more general channels, with respect to the relevant resource measure.
\subsection{Sparsification of Gaussian states \label{sec:SPARSE_delta}} 
In Ref.~\cite{Bravyi_2019} the authors introduce a routine that, given a superposition~$\ket{\psi}$, randomly generates a sparsified vector~$\ket{\Omega}$ close to~$\ket{\psi}$ in expectation. The degree of approximation is established by the so-called sparsification lemma, see~\cite[Lemma 6]{Bravyi_2019} and also~\cite[Corollary 1]{Seddon_2021}, which establishes that a sparsification~$\ket{\Omega}$ is close to~$\ket{\psi}$ on average. Although the sparsification algorithm and lemma were first established for superpositions of stabilizer states, they hold also for superpositions of Gaussian states. For completeness we give here the sparsify routine. 

\begin{algorithm}
\caption{Sparsify}
\begin{algorithmic}[1]
\Require (description of) state~$\ket{\psi} = \sum_j c_j \ket{\psi_j}$, integer~$k\in\bb{N}$
\Ensure vector~$\ket{\Omega}$ with~$k$ terms satisfying~\cite[Lemma 6]{Bravyi_2019}
\For{$m = 1$ to~$k$}
    \State choose index~$j$ with probability~$|c_j| / \|c\|_1$
    \State Set~$\ket{\omega_m} \gets c_j / |c_j| \ket{\psi_j}$
\EndFor
\State \textbf{return}~$\displaystyle \ket{\Omega} = (\|c\|_1 / k ) \sum_{m=1}^k \ket{\omega_m}$
\end{algorithmic}
\end{algorithm}

In Ref.~\cite{Seddon_2021} the authors establish a new sparsification lemma (see~\cite[Lemma 16]{Seddon_2021}) which guarantees that the ensemble~$\rho_1$ of normalized sparsifications
\begin{align}
	\label{eq:rho1}
	\rho_1 = \sum_\Omega \Pr(\Omega) \frac{\ketbra{\Omega}}{\langle{\Omega}|{\Omega}\rangle}
\end{align}
is close to the state~$\ketbra{\psi}$. 
The proof of~\cite[Lemma 16]{Seddon_2021} relies on two lemmas: the ensemble sampling lemma~\cite[Lemma 17]{Seddon_2021} and the variance bound~\cite[Lemma 18]{Seddon_2021}. Since these do not make use of the structure of stabilizer states, only on the notion of superpositions over free states, they hold true more generally. In particular, in the fermionic case, the sparsification lemma can be stated as follows.
\begin{lemma}[Sparsification lemma, see Lemma 16 in~\cite{Seddon_2021} for the stabilizer analogue]\label{sec:stateSparseLemma}
Let~$\rho_1$ be the mixed state in Eq.~\eqref{eq:rho1}. Let~$\ket{\psi}$ be an input state with known decomposition~$\ket{\psi} = \sum_{j} c_j \ket{\phi_j}$, where~$\ket{\phi_j}$ are Gaussian states. Let~$c$ be the vector whose elements are the coefficients~$c_j$, and 
\begin{align}
	\cte(\psi) = \|c\|_1 \sum_j |c_j| \cdot |\langle{\psi}|{\phi_j}\rangle|^2 \ .
	\label{eq:cteC}
\end{align}
Then there is a critical precision~$\delta_c = 8(\cte(\psi) - 1)/\|c\|_1^2$ such that for every target precision~$\delta_S$ satisfying~$\delta_S > \delta_c$, one can sample pure states from an ensemble~$\rho_1$, where every pure state drawn from~$\rho_1$ has Gaussian rank at most~$ k = \lceil 4 \|c\|_1^2 / \delta_S \rceil$ and 
$$\|\rho_1 - \ketbra{\psi}\|_1 \leq \delta_S + O(\delta_S^2) \ .~$$
\end{lemma}
Recall that the runtime of each of the subroutines \textsc{EvolveCircuit} and \textsc{FastNorm} is linear in this effective rank $k\propto \norm{c}_1^2$, so it is this that ultimately leads to classical simulation methods whose runtime scales linearly with extent. It is further noted in~\cite{Seddon_2021} that, if required, precision $\delta_S < \delta_c$ is still possible to achieve by increasing the effective Gaussian rank to
\begin{equation}
k = 4 \norm{c}_1^2 \left( \frac{\cte-1}{\delta_S \norm{c}_1^2} + \frac{1}{\delta_S} \right) + O(1).
\end{equation}
 It can be argued that the ratio $\cte/ \norm{c}^2_1$ is typically exponentially small for many-qubit magic states~\cite{Seddon_2021}, which can be understood as follows. Since the states $\ket{\phi_j}$ are Gaussian, $\cte$ can be upper bounded as follows
\begin{align}
\cte(\psi) = \|c\|_1 \sum_j |c_j| \cdot |\langle{\psi}|{\phi_j}\rangle|^2 \leq  \|c\|^2_1 F(\psi),
\end{align} 
where $F(\cdot)$ is the Gaussian fidelity. Because we expect the Gaussian fidelity to be exponentially small for non-trivial many-qubit magic states, we should find that the ratio $\cte(\psi) / \|c\|^2_1$ decays exponentially with the number of qubits.	Indeed, the quantity can be computed exactly and efficiently for tensor product magic states in cases where the extent is multiplicative. Moreover, in the stabilizer setting,~\cite{Seddon_2021} showed that for the class of magic states known as Clifford magic states, the critical precision is exactly zero, so for these states arbitrary precision is always achievable with rank scaling as $O( \norm{c}_1^2/\delta_S)$. 

Here we give the value of $\cte$ for several non-Gaussian states in the fermionic setting, finding that indeed several important states satisfies $\cte(\psi) = 1$ for their optimal decompositions, implying $\delta_c=0$. In general, given a state with known decomposition $\ket{\psi} = \sum_j c_j \ket{\phi_j}$, if the Gaussian states $\{\phi_j\}$ are orthogonal (which is known to be the case for the optimal decomposition of key magic states), we have that
\begin{align}
	\cte(\psi) = \| c \|_1 \sum_j |c_j|^3 \quad \Longrightarrow \quad \frac{\cte}{\| c \|^2_1} = \frac{ \sum_j |c_j|^3}{  \sum_k |c_k|} \ . \label{eq:cte}
\end{align}
For a normalized magic state, we have $\sum_j |c_j|^2 = 1$, so  $\sum_j |c_j|^3 < 1$, whereas $\sum_k |c_k| \geq \sqrt{\xi(\psi)} > 1$, so we find that the ratio is smaller than 1, and exponentially so for many-qubit states. 

Consider in particular  the fermionic magic state~\cite{Hebenstreit_2019, reardonsmith2024improved}
\begin{align}
	\label{eq:mtheta_sparse}
	\ket{m(\theta)} = \frac{1}{2} \left( \ket{0000} + \ket{0011} + \ket{1100} + e^{i\theta} \ket{1111} \right)\ ,
\end{align}
which is the resource state for the two-qubit controlled-phase gate $\mathrm{diag}(1,1,1,e^{i\theta})$.  For $\theta=\pi$ this is the resource state for $CZ$, and Gaussian-unitary equivalent to the magic state for SWAP.  This family of states was shown in~\cite{reardonsmith2024improved} to have optimal decomposition of the form 
\begin{equation}
\ket{m(\theta)} =\cos(\frac{\theta}{4}) \ket{\phi_\theta} + \sin (\frac{\theta}{4}) \ket{\phi'_{\theta}}
\end{equation}
where $\ket{\phi_\theta}$ and $\ket{\phi'_\theta}$ are Gaussian and $\bra{\phi_\theta}\ket{\phi'_\theta}=0$. This yields extent
\begin{equation}
\xi(\ket{m(\theta)}) = \left(\abs{\cos(\frac{\theta}{4}) } +  \abs{\sin(\frac{\theta}{4}) }\right)^2 = 1 + \abs{\sin(\frac{\theta}{2})}\ .
\end{equation}
It follows that $\cte$ takes the value
\begin{align}
	\cte(\ket{m(\theta)}) &= \sqrt{\xi(\psi)} \cdot \left( \left|\sin(\frac{\theta}{4})\right|^3 + \left|\cos(\frac{\theta}{4})\right|^3 \right) \label{eq:C_m_theta}
\end{align}
It is straightforward to show that $\sum_j |c_j| \cdot |\langle{\psi}|{\phi_j}\rangle|^2$ is multiplicative for tensor product states $\ket{\psi_t} = \ket{\psi}^{\otimes t}$ when the extent is multiplicative, as is the case for four-mode fermionic magic states~\cite{cudby2024gaussiandecompositionmagicstates,reardonsmith2024fermioniclinearopticalextent}. Therefore we have 
\begin{equation}
\cte(\ket{m(\theta)}^{\otimes t}) = \cte(\ket{m(\theta)})^t\ .
\end{equation}
For the special case $\theta = \pi$, evaluating~\cref{eq:C_m_theta} gives
\begin{equation}
\cte(\ket{m(\pi)}^{\otimes t}) = 1 \quad \forall t\ .
\end{equation}
We saw in~\cref{sec:extentplus} that states $\ket{\psi}$ in the set $E_t$, including arbitrary $t$-qubit tensor products of single-qubit states $\ket{+_\delta} =( \ket{0} + e^{i\delta} \ket{1})/\sqrt{2}$, can always be decomposed as an equal superposition of two orthogonal Gaussian states. We can easily check that when $c_0 = c_1 = 1/\sqrt{2}$, from \cref{eq:cte} we have $\cte(\psi) = 1$. This includes the $t$-fold product of the plus state,
\begin{equation}
\cte\left(\ket{+}^{\otimes t}\right) = 1 \quad \forall t\ .
\end{equation}
We saw in~\cref{sec:Udecomp}  that there is a tight connection between Gaussian decompositions of unitary gates, and those of the magic resource states that can be used to implement them. It was shown in~\cite{seddonThesis} that the ensemble sampling sparsification technique of~\cite{Seddon_2021} can be readily adapted to non-Gaussian unitary circuits (and, by linearity, to convex-unitary non-Gaussian circuits), such that terms from gates composed in sequence are sampled from in a similar way to terms from subsystem states in a product state, so we omit the full description here. Instead, in the next subsection we report a new result for sparsification of circuits that include intermediate non-unitary operations.

\subsection{Ensemble sampling lemma for adaptive circuits\label{sec:SPARSE_adapt}}

In this section we develop an analogue  for sparsifying adaptive circuits of the form 
\begin{equation}
\chan = \curlU_T \circ \curlA_T \circ \ldots \curlU_1 \circ \curlA_1 \circ \curlU_0  
\end{equation}
where $\curlU_t$ denotes the unitary channel $U_t(\cdot)U_t^\dagger$ and $\curlA_t$ is an adaptive channel with Kraus operators $K_{t,0}$ and $K_{t,1}$.  We will assume that the unitaries $U_t$ are in general non-free gates with known decomposition $U_t = \sum_j \alpha_{t,j} V_{t,j}$ for free gates $V_{t,j}$, whereas $K_{t,m}$ are free operators (for example, in the fermionic Gaussian case, projectors onto computational basis states composed with a Gaussian unitary). Replacing $ \curlA_t$ with Gaussian unitary operations recovers the unitary-case variant of~\cite{seddonThesis}. Then we can write the channel as a sum over trajectories,
\begin{equation}
\chan = \sum_{\arr{m}} K_{\arr{m}} (\cdot) K_{\arr{m}}^\dagger, \quad \text{where} \quad K_{\arr{m}}=U_T K_{T,m_T} \ldots U_1 K_{1,m_1} U_0.
\label{eq:kraus_sparified}
\end{equation}
Here the summation is over all bit strings $\arr{m}$ of length $T$. Note that for any initial free state $\ket{\phi_0}$, we can expand 
\begin{align}
K_{\arr{m}} \ket{\phi_0} &= \sum_{v_T} \alpha_{T,v_T} V_{T,v_T}  K_{T,m_T} \ldots \sum_{v_1} \alpha_{1,v_1} V_{1,v_1} K_{1,m_1} \sum_{v_0} \alpha_{0,v_0} V_{0,v_0} \ket{\phi_0} \\
			& = \sum_{v_0,\ldots,v_T} c_{v_0,\ldots,v_T} |\phi_{\arr{m},v_0,\ldots,v_T}\rangle = \sum_{\arr{v}} c_{\arr{v}} |\phi_{\arr{m},\arr{v}}\rangle  \label{eq:origdecomp}
\end{align}
where $|\phi_{\arr{m},v_0,\ldots,v_T}\rangle$ are subnormalized free states. The form of each vector
\begin{align}
	|\phi_{\arr{m},\arr{v}} \rangle  = |\phi_{\arr{m},v_0,\ldots,v_T} \rangle = V_{T,v_T} K_{T,m_T} \cdots V_{1,v_1} K_{1,m_1} V_{0,v_0} | \phi_0 \rangle
	\label{eq:phimv}
\end{align}	
depends on the sequence of measurement outcomes $\arr{m}$, but the coefficients 
\begin{align}
	c_{\arr{v}} = c_{v_0,\ldots,v_T} =  \prod_{t=0}^T \alpha_{t,v_t}
	\label{eq:defc}
\end{align}
do not, depending only on the unitary operations in the sequence. We will use the shorthand $\norm{c}_1 = \sum_{v_0,v_1,\ldots,v_T} \abs{c_{v_0,\ldots,v_T}}$. This is equal to the product of the $L^1$-norms of the individual gate decompositions. In general $\ket{\psim} = K_{\arr{m}} \ket{\phi_0}$ will be subnormalized, but since each adaptive channel is CPTP we have 
\begin{equation} 
\sum_{\arr{m}}\Tr[K_{\arr{m}} \op{\phi_0} K_{\arr{m}}^\dagger] = 1 \ .
\end{equation}
We define the sparsification of the final state $\rho_f = \chan(\op{\phi_0})$ as follows. First we specify that we require~$k$ terms in the specification. Then:
\begin{enumerate}
\item \label{it:sparsification1}For~$i = 1$ to~$k$:
\begin{enumerate}
\item For~$t = 1$ to~$T$:\label{step:samplegate}
\begin{enumerate}
\item Sample an index~$v$ from the decomposition of~$U_t$ with probability~$\abs{\alpha_{t,v}}/\norm{\alpha_t}_1$, where~$\norm{\alpha_t}_1 = \sum_v \abs{\alpha_{t,v}}$, and set~$W_{i,t} = (\alpha_{t,v}/\abs{\alpha_{t,v}})V_{t,v}$. 
\end{enumerate}
\item Let~$S_{i,\arr{m}} = (\prod_{t=T}^1 W_{i,t} K_{t,m_t}) W_{i,0}$. \label{step:composegates}
\end{enumerate}
\item \label{it:sparsification2} Let~$S_{\arr{m}}$ be defined by
\begin{equation}
S_{\arr{m}} =  \frac{\norm{c}_1}{k} \sum_{i=1}^k S_{i,\arr{m}} \ . \label{eq:sparseKraus}
\end{equation} 
\item \label{it:sparsification3} Let the sparsification of the final state of the circuit be defined by
\begin{equation}
\Omega = \sum_{\arr{m}} S_{\arr{m}}\op{\phi_0} S^\dagger_{\arr{m}} = \sum_{\arr{m}} \op{\Omega_{\arr{m}}} \ .  \label{eq:sparseFinalStateAdapt}
\end{equation}
\end{enumerate}
We note that this description of the sparsified final state does not need to be stored explicitly in the context of a simulation algorithm. One need only store the $T$ sampled Gaussian unitaries and the total phase for each of the $k$ terms in the sparsified Kraus operator in~\cref{eq:sparseKraus}. This pattern of Gaussian unitary gates does not depend on $\arr{m}$, which indexes each trajectory through the adaptive circuit, so we use the same sparsification pattern for all $\arr{m}$ and the space required is $O(kT)$. It can be shown that the following results hold.
\begingroup
\renewcommand{\thelemma}{\ref{lem:ensembleAdapt}}
\begin{lemma}[Restated]
The procedure defined in~\cref{it:sparsification1,it:sparsification2,it:sparsification3}, which has runtime $O(kT)$, returns a description of a sparsified state $\Omega$ with probability $Pr(\Omega)$, such that the expected normalized state defined by
\begin{equation}
\rho_1 = \mathbb{E}\left( \frac{\Omega}{\Tr[\Omega]} \right) = \sum_\Omega \Pr(\Omega) \frac{\Omega}{\Tr[\Omega]}
\end{equation}
satisfies
\begin{equation}
\norm{\rho_1 - \chan(\op{\phi_0})}_1 \leq  2  \frac{\norm{c}_1^2}{k} + \sqrt{\Var(\Tr[\Omega])}\ . \label{eq:ensembleresult}
\end{equation}
\end{lemma}
\endgroup

\begingroup
\renewcommand{\thelemma}{\ref{lem:varianceAdapt}}
\begin{lemma}[Restated]
The variance in the normalization of the randomly drawn operator $\Omega$ from~\cref{lem:ensembleAdapt} can be bounded as
\begin{equation}
\Var(\Tr \Omega ) \leq \frac{4(\cte - 1)}{k} + 2 \left( \frac{\norm{c}_1^2}{k} \right)^2 + \frac{10 - 12\cte}{k^2} +  O\left( \frac{\cte}{k^3}\right).
\end{equation}
\end{lemma}
\endgroup
We leave the full proofs of~\cref{lem:ensembleAdapt,lem:varianceAdapt} to~\cref{sec:sparseTechnical}, since they follow those of~\cite{Seddon_2021} for the case of non-stabilizer mixed states, with some additional steps to account for as yet undetermined trajectories through adaptive circuits. In~\cref{sec:modified_simulator_sparsifications} we sketch how state descriptions of the form \eqref{eq:sparseFinalStateAdapt} can be used in the context of a classical simulation algorithm, leading to a runtime that scales with the cost of each channel as described in~\cref{thm:bitsamplingthm}, except that the decomposition of each channel now admits adaptive Gaussian operations. The optimal cost is given by the augmented channel extent (see~\cref{def:channelextentaugmented}).

\section{Conclusions and discussion \label{sec:conclusion}}

The main motivation and context for our work has been the recently developed framework for classically simulating near-Gaussian fermionic quantum circuits~\cite{Dias_2024,reardonsmith2024improved}, and its relevance to quantum simulation on near-term hardware~\cite{phasecraft2025quantinuum,phasecraft2025google}. In particular~\cref{sec:algorithm} described an algorithm for simulating sampling from the output distribution of a (possibly noisy) quantum circuit, closely related to simulation techniques previously known in the near-Clifford setting~\cite{Bravyi_2016_fastnorm, Bravyi_2019, Seddon_2021, seddonThesis}. The notion of sparsification~\cite{Bravyi_2019, Seddon_2021} considered in~\cref{sec:SPARSE} connects the runtime of this algorithm for a given circuit to the (average) squared $L^1$-norm of decompositions of resourceful channels into Gaussian gates. Minimizing the $L^1$-norm for some target gate or state (or its average in the case of non-unitary channels or mixed states) corresponds to evaluating the relevant magic monotone (see~\cref{sec:magicmonotones}), for example the unitary extent in the case of unitary operations. Our main contribution is to provide analytic expressions for the provably optimal decomposition of important resource gates and channels in near-Gaussian circuits, and improved feasible decompositions in several other cases. We believe the various improved decompositions we have found can be used to streamline and accelerate classical simulation of near-Gaussian circuits, and we illustrate in~\cref{fig:practicalCosts} how modest improvements in the $L^1$ norm of the available decomposition at the gate level can lead to dramatic differences in runtime for circuits containing many non-Gaussian circuit elements.
We proved optimal decompositions for any two-qubit fermionic gate acting on nearest neighbours -- including the two-qubit controlled-phase gate, the $\SWAP$ gate, controlled-phase gates including $CZ$, and the $ZZ$-rotation. Moreover, we showed that the unitary extent is multiplicative when these gates are applied in parallel, so our two-qubit decompositions immediately extend to optimal decompositions for the gates applied as a layer across many qubits. When the two-qubit gates are applied to qubits that are not nearest neighbour, the decompositions remain optimal if the gates are diagonal. If not, the decompositions are still feasible, but not necessarily optimal. We leave a proof of optimality for future work. We also gave optimal decompositions for the non-parity-preserving single-qubit gates Hadamard and $Y$-rotation. 

The relevant monotone for the decomposition of non-unitary channels into probabilistic mixtures of Gaussian unitary gates is the convex-unitary channel extent. 
Here the quantity we seek to minimize is the expected unitary extent when decomposing the channel into a probabilistic mixture of unitary channels.
We provided an improved decomposition of the $P$-rotation subject to dephasing $P$ noise, where $P$ is any Pauli operator. We showed this decomposition is optimal with respect to convex-unitary extent when $P$ is $Y$ or nearest neighbour $ZZ$. These optimal decompositions have the advantage of being equimagical, meaning that the runtime of the simulator becomes deterministic, so that the $t$-fold product of the convex-unitary extent for $t$ magic gates quantifies the actual runtime for drawing a bit string from the output distribution, rather than the expected runtime.

We also extended the notion of channel extent by introducing the augmented channel extent, supplementing the set of unitary channels with adaptive Gaussian channels, thereby extending the type of decomposition we are able to minimize over, and enabling potential reductions in the effective average extent. In~\cref{sec:SPARSE} we provided a new protocol for sparsification of this new type of decomposition, extending previous results~\cite{Bravyi_2019,Seddon_2021,seddonThesis} guaranteeing closeness of the sparsified approximation in the trace-norm, and thereby maintaining the connection between monotone value and classical simulation cost. We showed that this provides a practical benefit by finding an improved decomposition of the two-qubit $ZZ$-rotation subject to single-qubit $Z$ noise. We proved that there exists no convex decomposition into unitary channels that reduces the average extent relative to the noiseless $ZZ$ rotation; it is only by admitting Gaussian measurement with classical feed-forward that we obtain our improved decomposition. This is somewhat remarkable since the target channel is itself convex-unitary.

While our main focus in this work has been on decompositions of gates and channels, we saw how the result of \cite{Brod2016} leads to the conclusion that $n$-qubit products of single-qubit states have constant cost in extent-based simulation. A simulation taking $\ket{+}^{\otimes n}$ is as cheap in our setting as a simulation with a single copy $\ket{+}$. This apparent gain is offset by the limited usage of the state $\ket{+}$: although it can in principle be used to implement the Hadamard and $Y$-rotation gates, because it is not fermionic it cannot be injected at arbitrary positions in the circuit \cite{Hebenstreit_2019}. Still, the optimal decomposition of states such as $\ket{+_\delta}^{\otimes n}$ can find applications in specific circuit instances where the input state on some subsystem is initialised in a superposition of non-fixed parity for algorithmic reasons \cite{Cleve1998HadTest,Villa2020quench,Sun2025quench}.

We further established that optimal extent decompositions provide optimal fermionic nonlinearity~\cite{Hakkaku_2022} decompositions in the case of unitary gates. We found that our decompositions for the dephased $ZZ$-rotation improve on those obtained numerically by Hakkaku et al.~\cite{Hakkaku_2022} for certain parameter regimes (large angles or small angles with weak noise). At the same time, we observe that the decompositions from~\cite{Hakkaku_2022} retain the advantage in the regime of small angles subject to strong noise. Our results focus on the gate and its convex-roof extension, as well as the extent of the channel. Extent-based simulators are usually preferred to quasiprobability methods because their runtime scales linearly (instead of quadratically) with the relevant magic monotone. This speedup is due to a fast-norm approximation algorithm used to estimate outcome probabilities which can be expressed as norms of linear combinations of Gaussian states (as is the case for projective-valued measurements such as occupation number measurements). It is worth pointing out that this speedup is not available when the task is to compute expectation values of observables that do not have a low-rank decomposition in terms of Gaussian projectors. In this case, it may be beneficial to use quasiprobability simulators if polynomial factors in system size and precision turn out to be smaller in practice. In this situation the fermionic nonlinearity optimal decompositions we propose may be helpful.

Our work has many parallels with previous research in stabilizer-based simulation, but we observe several features that are specific to the fermionic Gaussian setting. Firstly, the set of Gaussian gates is continuous, in contrast to the discrete set of Clifford gates. This makes the optimization problem of finding optimal decompositions of gates and channels harder. Numerically, it requires discretizing the search space, i.e. choosing a discrete set of Gaussian gates or channels into which to decompose. We avoid these obstacles by proving optimal decompositions analytically. Additionally, qubit order matters in the Gaussian setting, as certain gates are only matchgates if acting on nearest neighbours. This constraint is why our decompositions for some two-qubit fermionic gates are optimal only when the addressed qubits are nearest neighbours in the Jordan-Wigner chain -- we leave a generalization to non-nearest neighbours as an open question. Thirdly, we have observed that fermionic magic is highly robust to noise, in contrast to stabilizer magic. Specifically, there are rotation channels subject to dephased noise which only become convex-Gaussian when the noise is maximally dephasing. This dramatically contrasts with the stabilizer case, where modest levels of noise render the channel a mixture of stabilizer channels. This is illustrated by the fact that, in the case of one qubit, the set of convex-Gaussian states is a line in the Bloch sphere, while convex combinations of stabilizer states form an octahedron -- see~\cref{fig:circle-RP}.

We leave the following questions open for future research. We have obtained improved decompositions for $ZZ$-rotations subject to generalized $Z$ noise by combining known decompositions for single-qubit $Z$ and two-qubit $ZZ$ dephasing noise, but we expect these composite decompositions to be non-optimal with respect to unitary or augmented channel extent. Constructing improved decompositions and proving their optimality remains open.
Another direction would be to search for improved and optimal decompositions for other types of noisy gate, for example, non-Gaussian Pauli rotations subject to non-commuting stochastic Pauli noise, or amplitude-damping noise. A final question relating to optimality is concerned with gate sequences. In this work we have mainly considered the decomposition of individual Pauli rotation gates (and tensor products thereof). This raises the question of whether there is an advantage to be had in composing sequences of rotation gates before searching for an optimal linear decomposition. Magic monotones must in general be submultiplicative under composition, as trivially a unitary operation can always be followed by its inverse. Conversely, we have seen that the unitary extent is exactly multiplicative for certain gates applied in parallel. It is less clear under what conditions the extent of composed non-cancelling rotations with overlapping qubit support is strictly submultiplicative, and whether this can be used to obtain speedups in classical simulation.

Another open research direction is to what extent these newly found decompositions can reduce the cost of the Majorana Propagation algorithm~\cite{miller2025simulationfermioniccircuitsusing}, which, instead of Schrödinger-evolving a (sparsified, approximate) description of the state, Heisenberg-evolves a (sparsified, approximate) description of the observable. Like the simulation techniques used in this paper, the Majorana Propagation algorithm is at its most performant for purely Gaussian circuits and its cost for exact simulation scales exponentially with the number of non-Gaussian gates. However, to the best of our knowledge no tight relation between its simulation cost and any magic monotone is known.

Finally, we note that in the final stages of preparing this manuscript we became aware of recent parallel work studying the classical simulation of Clifford or Gaussian circuits with noisy magic initial states~\cite{heo2026noisymagicinputs}, which also appears to be closely related to the density-operator stabilizer rank simulator proposed in~\cite{Seddon_2021}. Our results are complementary to those of~\cite{heo2026noisymagicinputs}, since we primarily study decompositions of channels rather than states, and we resolve some questions around optimality not addressed in~\cite{heo2026noisymagicinputs}.
\hfill

\begin{figure}[h]
    \centering
    \begin{subfigure}[t]{0.485\textwidth}
        \centering
        \includegraphics[height=0.7\linewidth]{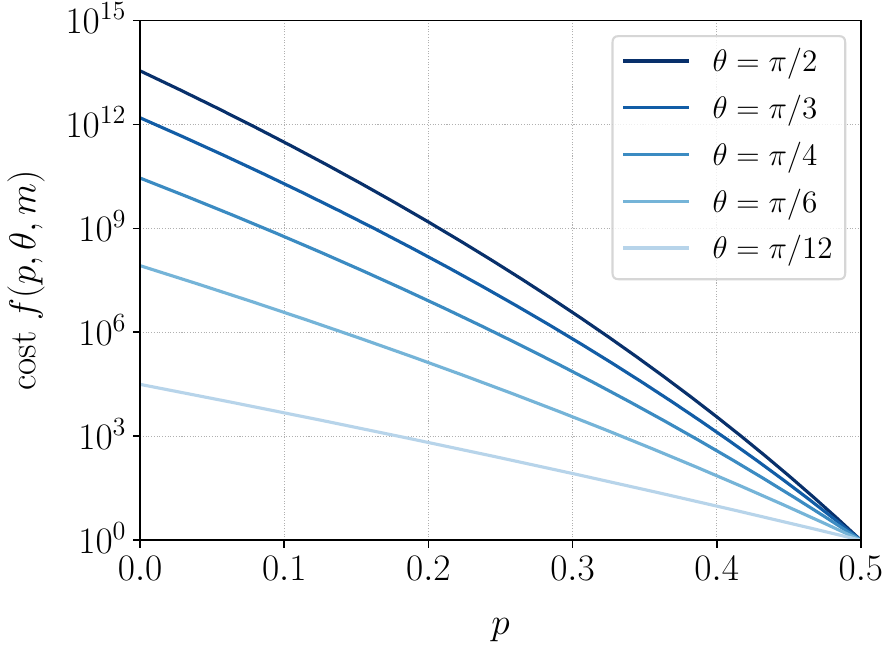}
		\caption{$m = 45$}
		\label{fig:}
    \end{subfigure} \hfill
	\begin{subfigure}[t]{0.485\textwidth}
        \centering
        \includegraphics[height=0.7\linewidth]{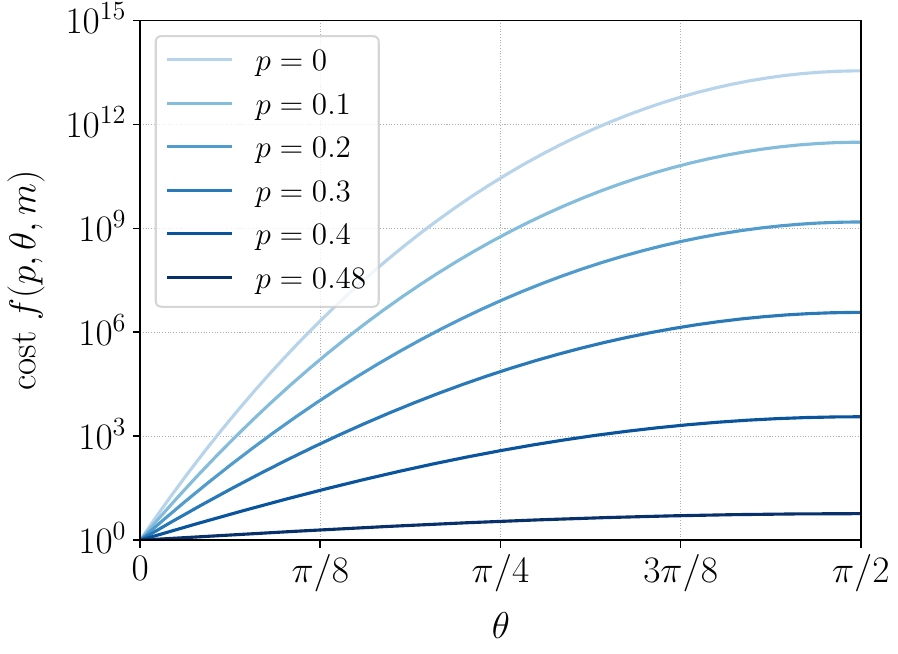}
        \caption{$m = 45$}
		\label{fig:} 
    \end{subfigure} \flushbottom \\ \vspace{4mm}
	\begin{subfigure}[t]{0.485\textwidth}
        \centering
        \includegraphics[height=0.7\linewidth]{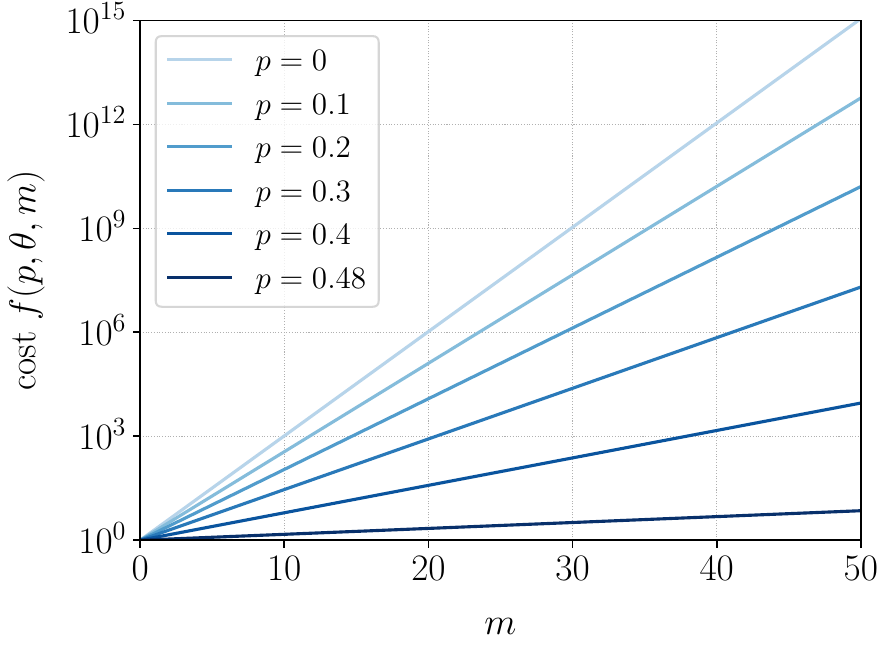}
		\caption{$\theta = \pi/2$}
		\label{fig:}
    \end{subfigure} \hfill
	\begin{subfigure}[t]{0.485\textwidth}
        \centering
        \includegraphics[height=0.7\linewidth]{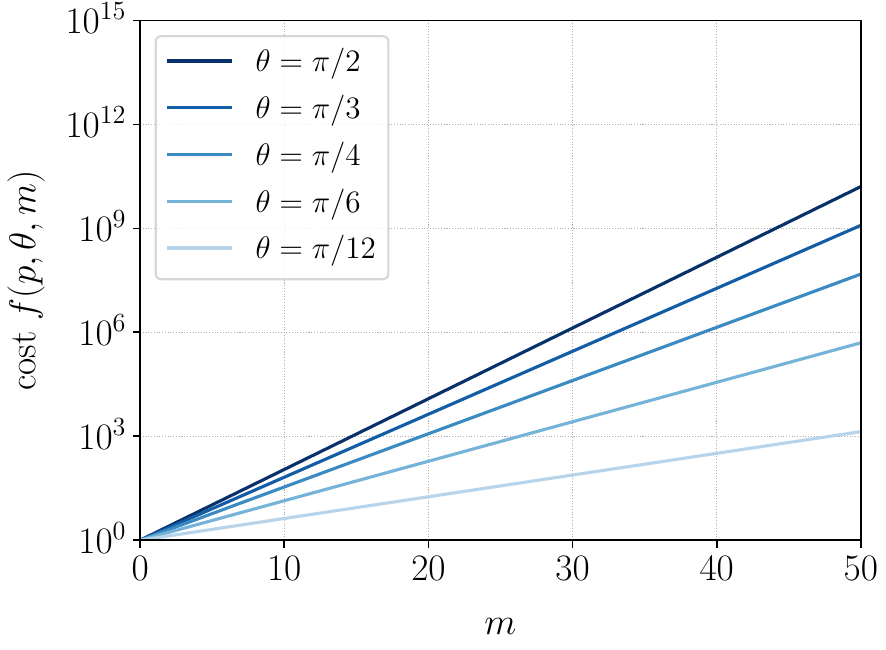}
        \caption{$p = 0.2$}
		\label{fig:} 
    \end{subfigure} \flushbottom 
    \caption{The function cost $f(p,\theta,m) = (1-(1-2p)\sin\theta)^m$ for different values of $m, p$ and $\theta$. 
	The function $f(p,\theta,m)$ is equal to $\Xi(\cal{N_Y})^m = \Xi(\cal{N_{ZZ}})^m $ and it is an upper bound for $\Xi^*(\cal{N'_{ZZ}})^m$ and $\Xi^*(\cal{N'})^m$, as shown in~\cref{sec:improved_decomp_NZZ,sec:decompgeneraldephasing}.
This upperbounds the cost of simulating $m$ copies of certain channels in parallel: 	
The cost of simulating $\cal{N_Y}^{\otimes m}$ grows linearly with $\Xi(\cal{N_Y}^{\otimes m}) \leq \Xi(\cal{N_Y})^m$. Similarly, the cost of simulating $\cal{N_{ZZ}}^{\otimes m}$ or $\cal{N'_{ZZ}}^{\otimes m}$ grows linearly with $\Xi(\cal{N_Y})^m$, as $\Xi^*(\cal{N'_{ZZ}}),\Xi^*(\cal{N'}) \leq \Xi(\cal{N_Y})$.
	}
    \label{fig:practicalCosts}
\end{figure}

\hfill
\newpage
\noindent\textbf{Acknowledgements.} 
BD acknowledges support by the European Research Council under grant
agreement no. 101001976 (project EQUIPTNT). BD thanks Marc Langer for insightful discussions. All authors would also like to thank Shigeo Hakkaku et al. for kindly providing the numerical data for the data points in Figures 1 and 2 from their article~\cite{Hakkaku_2022}. The authors would also like to thank Faisal Alam, Adrian Chapman, Steve Flammia, Joel Klassen, Raul Santos, Sabrina Wang and Guillermo Gonz\'alez-Garc\'{\i}a for helpful discussions on our work and comments on the manuscript.

\phantomsection
\addcontentsline{toc}{section}{References}
\printbibliography
\newpage
\phantomsection
\addcontentsline{toc}{section}{Appendix}

\appendix
\crefalias{section}{appendix}
\crefalias{subsection}{appendix}
\noindent{\textbf{\LARGE{Appendix}}}
\numberwithin{equation}{section}
\section{Proofs for sparsification of adaptive circuits\label{sec:sparseTechnical}}

In this appendix we prove~\cref{lem:ensembleAdapt,lem:varianceAdapt}.

Recall that in~\cref{sec:SPARSE_adapt} we outlined a new procedure for sparsifying adaptive circuits that generalizes the results of~\cite{Seddon_2021} and~\cite{seddonThesis}, which respectively dealt with (convex combinations of) states and unitary channels. In our setting, we assume we are given a channel
\begin{equation}
\chan = \curlU_T \circ \curlA_T \circ \ldots \curlU_1 \circ \curlA_1 \circ \curlU_0  \label{eq:appendixChannelSequence}
\end{equation}
where $\mathcal{U}_t$ are non-Gaussian unitary operations, and $\mathcal{A}_j$ are adaptive Gaussian channels, for a Gaussian initial state $\ket{\phi_0}$. 
The final state can be expressed in terms of Kraus operators as $\chan(\op{\phi_0}) = \sum_{\arr{m}} K_{\arr{m}} \op{\phi_0} K_{\arr{m}}^\dagger$, so that for each Kraus operator $K_{\arr{m}}$, labelled by a trajectory $\arr{m}$, the final state $K_{\arr{m}} \ket{\phi_0}= \sum_{\arr{v}} c_{\arr{v}} \ket*{\phi_{\arr{m},\arr{v}}}$ is a linear combination of Gaussian states $ \ket*{\phi_{\arr{m},\arr{v}}}$ defined in~\cref{eq:phimv} with coefficients $c_{\arr{v}}$ defined in~\cref{eq:defc}. Notice that the states $\ket*{\phi_{\arr{m},\arr{v}}}$ depend on the trajectories $\arr{m}$ but the coefficients $c_{\arr{v}}$ do not, since they depend only on the decompositions of the unitary gates. This is crucial for making the sparsification scheme practical; since there can be exponentially many trajectories $\arr{m}$ we propose that the description of the final state is in terms of the Gaussian unitaries sampled at each step, rather than an explicit description of the density operator. 

The main result is that there is a procedure that randomly outputs a description of an operator $\Omega = \sum_{\arr{m}} \op{\Omega_{\arr{m}}}$ in time $O(kT)$, where each $\ket{\Omega_{\arr{m}}}$ is a non-normalized pure state with $k$ Gaussian terms, such that $\Omega/\Tr[\Omega]$ can be made close in expectation to the final state $\chan(\op{\phi_0})$ of the circuit by choosing large enough $k$. 

Before proving~\cref{lem:ensembleAdapt}, we first generalize an argument from the proof of the sparsification tail bound (Lemma 7) in~\cite{Bravyi_2019}, concerning the normalization $\Tr[\Omega]$. Recall that we assume each unitary operator $U_t$ involved in~\cref{eq:appendixChannelSequence} has a known Gaussian decomposition  $U_t = \sum_j \alpha_{t,j} V_{t,j}$, so that the $L^1$-norm associated with the total circuit is $\norm{c}_1 = \prod_{t} \left(\sum_j \abs{\alpha_{t,j}} \right)$. Then the statement of our intermediate result is as follows.
\begin{lemma}[Average normalization of a sparsified operator]\label{lem:avgnormalisation}
Let $\Omega$ be an operator randomly drawn by the sparsification procedure defined in~\cref{it:sparsification1,it:sparsification2,it:sparsification3} in~\cref{sec:SPARSE_adapt}. Then its expected normalization is given by
\begin{equation}
 \mathbb{E}(\Tr[\Omega]) = \frac{\norm{c}_1^2}{k} + \frac{k(k-1)}{k^2} = 1 + \frac{\norm{c}_1^2-1}{k}.
\end{equation}
\end{lemma}
\begin{proof}
Recall that we have defined
\begin{equation}
\Omega = \sum_{\arr{m}} S_{\arr{m}}\op{\phi_0} S^\dagger_{\arr{m}} = \sum_{\arr{m}} \op{\Omega_{\arr{m}}}
\end{equation}
where
\begin{equation}
S_{\arr{m}} =  \frac{\norm{c}_1}{k} \sum_{i=1}^k S_{i,\arr{m}}
\end{equation}
and the $S_{i,\arr{m}}$ are generated randomly during steps \ref{step:samplegate} and \ref{step:composegates} of the procedure, defined by
\begin{equation}
S_{i,\arr{m}} = \left(\prod_{t=T}^1 W_{i,t} K_{t,m_t}\right) W_{i,0}
\end{equation}
where $K_{t,m_t}$ are the (Gaussian) Kraus operators associated with each Gaussian adaptive channel $\mathcal{A}_t$, and $W_{i,t}$ are selected at random from the decomposition of each gate $U_t$, setting $W_{i,t} = (\alpha_{t,v}/\abs{\alpha_{t,v}})V_{t,v}$ with probability $\abs{\alpha_{t,v}}/\norm{\alpha_t}_1$. 
Note that we have defined the procedure in such a way that for fixed $i$, the same Gaussian gates $W_{i,t}$ are selected for all $\arr{m}$, and the efficiently simulable Kraus operators $K_{t,m_t}$ are fixed for all sparsifications.

By linearity and cyclicity of the trace, we have 
\begin{align}
\Tr[\Omega] = \sum_{\arr{m}} \bra{\Omega_{\arr{m}}}\ket{\Omega_{\arr{m}}} = \frac{\norm{c}_1^2}{k^2} \sum_{p,q} \sum_{\arr{m}} \bra{\phi_0} S^\dagger_{p,\arr{m}} S_{q,\arr{m}} \ket{\phi_0} =  \frac{\norm{c}_1^2}{k^2} \sum_{p,q} \sum_{\arr{m}} \bra*{\omega_{p,\arr{m}}} \ket*{\omega_{q,\arr{m}}}
\end{align}
where we have defined $|\omega_{i,\arr{m}}\rangle = S_{i,\arr{m}} \ket{\phi_0}$. Note that each index $i$ labels a particular randomly chosen pattern of Gaussian gates, sampled independently for each $i$, so we need to differentiate between terms where $p=q$ and those where $p\neq q$, i.e.,
\begin{equation}
\Tr[\Omega]  =  \frac{\norm{c}_1^2}{k^2} \sum_p \left( \sum_{\arr{m}} \bra*{\omega_{p,\arr{m}}}\ket*{\omega_{p,\arr{m}}} +  \sum_{q\neq p} \sum_{\arr{m}} \bra*{\omega_{p,\arr{m}}} \ket*{\omega_{q,\arr{m}}} \right).
\end{equation}
Referring back to the original decomposition~\eqref{eq:origdecomp} and steps \ref{step:samplegate} and \ref{step:composegates} of the procedure, let $\arr{v} = (v_0,v_1,\ldots,v_T)$ define a pattern of Gaussian gates from the known decomposition. Recalling that this pattern does not depend on $\arr{m}$, we see that when generating $S_{i,\arr{m}}$, the pattern $\arr{v}$ is sampled with probability
\begin{equation}
\mathrm{Pr}({\arr{v} })= \prod_{t=1}^T \frac{\abs{\alpha_{t,v_t}}}{\norm{\alpha_t}_1} = \frac{\abs*{c_{\arr{v}}}}{\norm{c}_1} \label{eq:patternprob}
\end{equation}
where $c_{\arr{v}} = \prod_{t=1}^T \alpha_{t,v_t}$, and the last equality follows since
\begin{equation}
\prod_{t=1}^T \norm{\alpha_t}_1 = \prod_{t=1}^T \sum_{v_t} \abs{\alpha_{t,v_t}} = \sum_{\arr{v}} \abs{\prod_{t=1}^T \alpha_{t,v_t}} = \sum_{\arr{v}} \abs*{c_{\arr{v}}} = \norm{c}_1.
\end{equation}
So $|\omega_{i,\arr{m}}\rangle = (c_{\arr{v}}/\abs*{c_{\arr{v}}}) |\phi_{\arr{m},\arr{v}}\rangle$ with probability $\mathrm{Pr}({\arr{v} })$ given by~\cref{eq:patternprob}. Notice that with each pattern of Gaussian gates labelled $\arr{v}$ we can associate a complete set of Kraus operators 
\begin{equation}
K_{\arr{m},\arr{v}} = \prod_{t=T}^1 (V_{t,v_t} K_{t,m_t}) V_{0,v_0}
\end{equation}
such that $K_{\arr{m},\arr{v}} \ket{\phi_0} = |\phi_{\arr{m},\arr{v}}\rangle$ and $\sumvec{m} K^\dagger_{\arr{m},\arr{v}} K_{\arr{m},\arr{v}}  = \id$. This follows from the fact that each $V_{t,v_t}$ is unitary, and each pair of operators $\{K_{t,0}$, $K_{t,1}\}$ is itself a complete set of Kraus operators associated with a CPTP map. Taking the expected value of the trace, 
\begin{align}
\mathbb{E}(\Tr[\Omega]) = \frac{\norm{c}_1^2}{k^2} \left(  \sum_p \mathbb{E}\left[ \sum_{\arr{m}}  \bra*{\omega_{p,\arr{m}}}\ket*{\omega_{p,\arr{m}}} \right] + \sum_p \sum_{q\neq p} \mathbb{E}\left[   \sum_{\arr{m}} \bra*{\omega_{p,\arr{m}}} \ket*{\omega_{q,\arr{m}}} \right] \right). \label{eq:expTrOmega}
\end{align}
For the first term we have
\begin{align}
\mathbb{E}\left[ \sum_{\arr{m}} \langle\omega_{p,\arr{m}} | \omega_{p,\arr{m}}\rangle \right] & = \sum_{\arr{v}} \frac{\abs*{c_{\arr{v}}}}{\norm{c}_1} \sum_{\arr{m}} \bra*{\phi_{\arr{m},\arr{v}}} \frac{c^*_{\arr{v}}}{\abs*{c_{\arr{v}}}} \cdot \frac{c_{\arr{v}}}{\abs*{c_{\arr{v}}}} \ket*{\phi_{\arr{m},\arr{v}}}  \\
 & = \sumvec{v}  \frac{\abs{c_{\arr{v}}}}{\norm{c}_1}   \bra*{\phi_0} \sumvec{m} K^\dagger_{\arr{m},\arr{v}} K_{\arr{m},\arr{v}}  \ket*{\phi_0}= 1,
\end{align}
where the last equality comes from the completeness of the Kraus operators $K_{\arr{m},\arr{v}}$ over $\arr{m}$, and the fact that $\sum_{\arr{v}}\abs{c_{\arr{v}}} = \norm{c}_1$. Then for the second term, where $p\neq q$, we have 
\begin{align}
\mathbb{E}\left[   \sum_{\arr{m}} \bra*{\omega_{p,\arr{m}}} \ket*{\omega_{q,\arr{m}}} \right] &=\sumvec{v} \sumvec{w} \frac{\abs*{c_{\arr{v}}}}{\norm{c}_1}  \frac{\abs*{c_{\arr{w}}}}{\norm{c}_1} \sumvec{m} \bra*{\phi_{\arr{m},\arr{v}}} \frac{c^*_{\arr{v}}}{\abs*{c_{\arr{v}}}} \cdot \frac{c_{\arr{w}}}{\abs*{c_{\arr{w}}}} \ket*{\phi_{\arr{m},\arr{w}}}   \\
& = \frac{1}{\norm{c}_1^2} \sumvec{m} \left( \sumvec{v} \bra*{\phi_{\arr{m},\arr{v}}} c^*_{\arr{v}} \right)  \left(\sumvec{w} c_{\arr{w}}\ket*{\phi_{\arr{m},\arr{w}}}  \right) \\
& = \frac{1}{\norm{c}_1^2} \sumvec{m} \bra{\phi_0} K^\dagger_{\arr{m}} K_{\arr{m}} \ket{\phi_0} =  \frac{1}{\norm{c}_1^2}.
\end{align}
Substituting back into~\cref{eq:expTrOmega} and recalling that there are $k$ independently sampled terms in total, we obtain
\begin{align}
\mathbb{E}(\Tr[\Omega])& = \frac{\norm{c}_1^2}{k^2} \left(  \sum_p 1 + \sum_p \sum_{q\neq p} \frac{1}{\norm{c}_1^2} \right) \\
  & =  \frac{\norm{c}_1^2}{k^2}  \left( k + \frac{k(k-1)}{\norm{c}_1^2} \right) \\
  & = \frac{\norm{c}_1^2}{k} + \frac{k(k-1)}{k^2},
\end{align}
which is the required expression.
\end{proof}
Next we restate the main lemma for convenience, and prove the result below.
\begingroup
\renewcommand{\thelemma}{\ref{lem:ensembleAdapt}}
\begin{lemma}[Restated]
The procedure defined in~\cref{it:sparsification1,it:sparsification2,it:sparsification3} in~\cref{sec:SPARSE_adapt}, which has runtime $O(kT)$, returns a description of a sparsified state $\Omega$ with probability $Pr(\Omega)$, such that the expected normalized state defined by
\begin{equation}
\rho_1 = \mathbb{E}\left( \frac{\Omega}{\Tr[\Omega]} \right) = \sum_\Omega \Pr(\Omega) \frac{\Omega}{\Tr[\Omega]}
\end{equation}
satisfies
\begin{equation}
\norm{\rho_1 - \chan(\op{\phi_0})}_1 \leq  2  \frac{\norm{c}_1^2}{k} + \sqrt{\Var(\Tr[\Omega])}\ . \label{eq:ensembleresult}
\end{equation}
\end{lemma}
\addtocounter{lemma}{-1} 
\endgroup
\begin{proof}
The proof closely follows the arguments of~\cite{Seddon_2021}, but requires taking care of some additional technical details to deal with adaptive channels. We seek to bound the distance
\begin{equation}
\norm{\rho_1 - \chan(\op{\phi_0})}_1 = \Delta .
\end{equation}
Using the triangle inequality, we have
\begin{equation}
\Delta \leq \norm{\rho_1 - \rho_2}_1 + \norm{\rho_2 - \chan(\op{\phi_0})}_1, \label{eq:triangled}
\end{equation}
where $\rho_2 = \mathbb{E}(\Omega)/\mu$, and $\mu=\mathbb{E}(\Tr[\Omega])$. Then
\begin{align}
\norm{\rho_1 - \rho_2}_1 &= \left\| \mathbb{E}\left( \frac{\Omega}{\Tr[\Omega]}\right) - \frac{\mathbb{E}(\Omega)}{\mu} \right\|_1 = \left\|\mathbb{E}\left[\Omega\left( \frac{1}{\Tr[\Omega]}- \frac{1}{\mu} \right)\right] \right\|_1 \ .
\end{align}
 Then by Jensen's inequality,
\begin{align} 
\norm{\rho_1 - \rho_2}_1  &\leq \mathbb{E} \left[ \left\|\Omega\left( \frac{1}{\Tr[\Omega]}- \frac{1}{\mu}\right) \right\|_1 \right] \\
 &=  \mathbb{E} \left( \abs{ \frac{1}{\Tr[\Omega]}- \frac{1}{\mu}} \norm{\Omega}_1 \right)  \\
  &=  \mathbb{E} \left( \abs{ \frac{1}{\Tr[\Omega]}- \frac{1}{\mu}}\Tr[\Omega] \right)  \\
  &  = \frac{1}{\mu} \mathbb{E} (\abs{\mu - \Tr[\Omega]}) \leq \mathbb{E} (\abs{\mu - \Tr[\Omega]}) 
\end{align}
where the last step used $\mu \geq 1$, which is guaranteed by~\cref{lem:avgnormalisation}. The final expression is the expected deviation of $\Tr[\Omega]$ from the mean, so as in the original ensemble sampling lemma we can again use Jensen's inequality to obtain
\begin{equation}
\norm{\rho_1 - \rho_2}_1 \leq \sqrt{(\mathbb{E}\abs{\mu - \Tr[\Omega]})^2 }\leq  \sqrt{\mathbb{E}(\abs{\mu - \Tr[\Omega]}^2) } = \sqrt{\Var(\Tr[\Omega])}. 
\label{eq:boundrho1rho2}
\end{equation}
Now, to bound the second term in~\cref{eq:triangled} we observe that 
\begin{equation}
\norm{\rho_2 - \chan(\op{\phi_0})}_1 = \left\|\frac{\mathbb{E}(\Omega)}{\mu} - \sumvec{m} \op{\psi_{\arr{m}}} \right\|_1 =  \frac{1}{\mu} \left\|\mathbb{E}(\Omega) - \mu \sumvec{m} \op{\psi_{\arr{m}}} \right\|_1 \label{eq:normpart2}
\end{equation}
where $\ket{\psim} = K_{\arr{m}}\ket{\phi_0}$. Expanding the expected $\Omega$, using the same notation as in the proof of~\cref{lem:avgnormalisation},
\begin{align}
\mathbb{E}(\Omega) &=  \frac{\norm{c}_1^2}{k^2} \mathbb{E}\left(\sumvec{m} \sum_{p,q} \op*{\omega_{p,\arr{m}}}{\omega_{q,\arr{m}}} \right)\\
				 &= \frac{\norm{c}_1^2}{k^2}\left(\sum_{p}  \mathbb{E}\left(\sumvec{m}  \op*{\omega_{p,\arr{m}}}{\omega_{p,\arr{m}}}\right) + \sum_p \sum_{q\neq p}  \mathbb{E}\left(\sumvec{m}  \op*{\omega_{p,\arr{m}}}{\omega_{q,\arr{m}}}\right)\right) \label{eq:expOmega}.
\end{align}
For the first term, we define $\sigma = \mathbb{E}\left(\sumvec{m}  \op*{\omega_{p,\arr{m}}}{\omega_{p,\arr{m}}}\right)$, then
\newcommand{\Kmv}{K_{\arr{m},\arr{v}}}
\begin{align}
  \sigma & = \sumvec{v} \frac{\abs{c_{\arr{v}}}}{\norm{c}_1} \sumvec{m} \frac{c^*_{\arr{v}}}{\abs*{c_{\arr{v}}}} \cdot \frac{c_{\arr{v}}}{\abs*{c_{\arr{v}}}} \op*{\phi_{\arr{m},\arr{v}}} \\
  			& =  \sumvec{v} \mathrm{Pr}(\arr{v})\sumvec{m} \Kmv \op*{\phi_0} \Kmv^\dagger
\end{align}
where $K_{\arr{m},\arr{v}}$ is as defined in the proof of~\cref{lem:avgnormalisation}. For any $\arr{v}$, the sum $\sumvec{m}  \op*{\phi_{\arr{m},\arr{v}}} $ is a normalised density operator, since $K_{\arr{m},\arr{v}}$ over $\arr{m}$ are the Kraus operators for a CPTP map. Moreover $  \mathrm{Pr}(\arr{v})$ is a probability distribution over $\arr{v}$, so $\sigma$ is itself a normalised density operator, $\norm{\sigma}_1 = \Tr[\sigma] = 1$. Now, considering the second term in~\cref{eq:expOmega},
\begin{align}
\mathbb{E}\left(\sumvec{m}  \op*{\omega_{p,\arr{m}}}{\omega_{q,\arr{m}}}\right) &= \sumvec{v} \sumvec{w} \frac{\abs*{c_{\arr{v}}}}{\norm{c}_1}  \frac{\abs*{c_{\arr{w}}}}{\norm{c}_1} \sumvec{m}  \frac{c_{\arr{v}}}{\abs*{c_{\arr{v}}}} \op*{\phi_{\arr{m},\arr{v}}}{\phi_{\arr{m},\arr{w}}} \frac{c^*_{\arr{w}}}{\abs*{c_{\arr{w}}}} \\
			&= \frac{1}{\norm{c}_1^2}\sumvec{m} \left(\sumvec{v} c_{\arr{v}} \ket*{\phi_{\arr{m},\arr{v}}} \right) \left(\sumvec{w} c^*_{\arr{w}} \ket*{\phi_{\arr{m},\arr{w}}} \right) \\
			&=  \frac{1}{\norm{c}_1^2}\sumvec{m} \op{\psim} = \frac{1}{\norm{c}_1^2} \chan(\op{\phi_0}) \ .
\end{align}
Substituting back into~\cref{eq:expOmega}
\begin{align}
\mathbb{E}(\Omega) = \frac{\norm{c}_1^2}{k^2}\left( k \sigma + \frac{k(k-1)}{\norm{c}_1^2} \chan(\op{\phi_0}) \right) =  \frac{\norm{c}_1^2}{k} \sigma + \left(1-\frac{1}{k}\right) \chan(\op{\phi_0})\ .
\end{align}
Then from~\cref{eq:normpart2}, and using~\cref{lem:avgnormalisation} for the value of $\mu$,
\begin{align}
\norm{\rho_2 - \chan(\op{\phi_0})}_1 &= \frac{1}{\mu} \left\|  \frac{\norm{c}_1^2}{k} \sigma +\left[ \left(1-\frac{1}{k}\right) - \mu  \right]\chan(\op{\phi_0}) \right\|_1 \\
							&=    \frac{\norm{c}_1^2}{k\mu}\norm{   \sigma - \chan(\op{\phi_0}) }_1  \\
							& \leq \frac{\norm{c}_1^2}{k} \left(\norm{\sigma}_1 + \norm{ \chan(\op{\phi_0})}_1\right) \leq 2  \frac{\norm{c}_1^2}{k} \ .
\end{align}
Putting this together with our bound in~\cref{eq:boundrho1rho2} on $\norm{\rho_1 - \rho_2}_1$ we have
\begin{equation}
\norm{\rho_1 -  \chan(\op{\phi_0})}_1 \leq  2  \frac{\norm{c}_1^2}{k} + \sqrt{\Var(\Tr[\Omega])} \ , \label{eq:ensembleresult}
\end{equation}
which is the required expression.
\end{proof}
Finally, we again follow a similar argument to~\cite{Seddon_2021} to prove~\cref{lem:varianceAdapt}, which states that the variance of $\Tr(\Omega)$ can be bounded as 
\begin{equation}
\Var(\Tr[\Omega]) \leq \frac{4(\cte - 1)}{k} + 2 \left( \frac{\norm{c}_1^2}{k} \right)^2 + \frac{10 - 12\cte}{k^2} +  O\left( \frac{\cte}{k^3}\right)\ .
\end{equation}
\begin{proof}
By definition,
\begin{equation}
\Var(\Tr[\Omega]) = \mathbb{E}(\Tr[\Omega]^2) - (\mathbb{E}(\Tr[\Omega]))^2\ .
\end{equation}
The second term is $\mu^2$, which we have already bounded in~\cref{lem:avgnormalisation}. To bound the first term, following the proof steps of~\cite{Seddon_2021} we first expand 
\begin{align}
\Tr[\Omega]^2 &= \left( \sum_{\arr{m}} \braket{\Omega_{\arr{m}}} \right)^2 \\
			& = \frac{\norm{c}_1^4}{k^4}  \left( \sum_{p} \sumvec{m}\bra*{\omega_{p,\arr{m}}}\ket*{\omega_{p,\arr{m}}} + \sumvec{m}  \sum_{p} \sum_{q\neq p } \bra*{\omega_{q,\arr{m}}}\ket*{\omega_{p,\arr{m}}} \right)^2 \\
			& =  \frac{\norm{c}_1^4}{k^4}  \left( \sum_{p} 1 + \sumvec{m}  \sum_{p} \sum_{q\neq p } \bra*{\omega_{q,\arr{m}}}\ket*{\omega_{p,\arr{m}}} \right)^2 \\
			& =  \frac{\norm{c}_1^4}{k^4}  \left( k^2 + 2kB +B^2 \right) \label{eq:var1stterm}
\end{align}
where $\sumvec{m}\bra*{\omega_{p,\arr{m}}}\ket*{\omega_{p,\arr{m}}}=1$ follows from the arguments of~\cref{lem:avgnormalisation}, and we define 
\begin{equation}
B = \sumvec{m}  \sum_{p} \sum_{q\neq p } \bra*{\omega_{q,\arr{m}}}\ket*{\omega_{p,\arr{m}}}\ . \label{eq:defB}
\end{equation}
From the proof of~\cref{lem:avgnormalisation} we have
\begin{align}
 (\mathbb{E}(\Tr[\Omega]))^2 = \mu^2 &= \left(\frac{\norm{c}_1^2}{k^2} \left[ k + \mathbb{E}\left(\sumvec{m}  \sum_{p} \sum_{q\neq p } \bra*{\omega_{q,\arr{m}}}\ket*{\omega_{p,\arr{m}}} \right)\right]\right)^2 \\
	&= \frac{\norm{c}_1^4}{k^4} \left(k^2 + 2k \mathbb{E}(B) + \mathbb{E}(B)^2\right) \label{eq:var2ndterm}\ .
\end{align}
Comparing~\cref{eq:var1stterm} and~\cref{eq:var2ndterm}, we have
\begin{align}
\Var(\Tr[\Omega]) & = \mathbb{E}(\Tr[\Omega]^2) - (\mathbb{E}(\Tr[\Omega]))^2 \\
			 & =  \frac{\norm{c}_1^4}{k^4} \left(\mathbb{E}(B^2) - (\mathbb{E}(B))^2 \right)=  \frac{\norm{c}_1^4}{k^4} \Var(B)\ . \label{eq:VarintermsofB}
\end{align}
Then by~\cref{lem:avgnormalisation},
\begin{equation}
(\mathbb{E}(B))^2 = \left(\frac{k(k-1)}{\norm{c}_1^2} \right)^2 = \frac{k^2(k-1)^2}{\norm{c}_1^4}\ . \label{eq:squaredexpectedB}
\end{equation}
It remains to obtain an expression for $\mathbb{E}(B^2)$. We note that the key technical difference with~\cite{Seddon_2021} is that whereas in the prior work the equivalent quantity to $B$ is a sum over inner products between normalized pure states, in our definition the inner products are between  unnormalized $\ket*{\omega_{p,\arr{m}}}$ associated with Kraus operators $K_{\arr{m},\arr{v}}$. Expanding $B^2$, 
\begin{align}
B^2 &=\left(\sumvec{m}  \sum_{p} \sum_{q\neq p } \bra*{\omega_{q,\arr{m}}}\ket*{\omega_{p,\arr{m}}}\right) \left(\sumvecp{m}  \sum_{r} \sum_{s\neq r } \bra*{\omega_{s,\arr{m}'}}\ket*{\omega_{r,\arr{m}'}}\right) \\
 & = \sum_{(p,q,r,s)\in \mathcal{A}} \sumvec{m}  \bra*{\omega_{q,\arr{m}}}\ket*{\omega_{p,\arr{m}}} \sumvecp{m}   \bra*{\omega_{s,\arr{m}'}}\ket*{\omega_{r,\arr{m}'}} + B'
\end{align}
where $\mathcal{A}$ is the set of all tuples $(p,q,r,s)$ such that $p,q,r,s$ are integers in the range $[1,k]$ and are all distinct, and $B'$ is the sum over all remaining terms where at least one pair of indices is shared between the two inner products. The first term is the simplest to deal with:
\begin{align}
\mathbb{E}\left(  \sumvec{m}  \bra*{\omega_{q,\arr{m}}}\ket*{\omega_{p,\arr{m}}} \sumvecp{m}   \bra*{\omega_{s,\arr{m}'}}\ket*{\omega_{r,\arr{m}'}} \right)
& =   \sum_{\arr{m},\arr{m}'} \mathbb{E}(\!\bra*{\omega_{q,\arr{m}}})\mathbb{E}(\!\ket*{\omega_{p,\arr{m}}})\mathbb{E}(\!\bra*{\omega_{s,\arr{m}'}})\mathbb{E}(\!\ket*{\omega_{r,\arr{m}'}}) \ .
\end{align}
We can make this move because $\ket*{\omega_{p,\arr{m}}}$ are $\ket*{\omega_{q,\arr{m}'}}$ are sampled independently when $p\neq q$. Note that they remain independent when $\arr{m}'=\arr{m}$, since $\arr{m}$ just labels a fixed term in the summation, and does not correspond to correlations in the sense of a random variable. We have
\begin{equation}
\mathbb{E}(\!\ket*{\omega_{p,\arr{m}}}) = \sumvec{v} \frac{\abs{c_{\arr{v}}}}{\norm{c}_1} \frac{c_{\arr{v}}}{\abs{c_{\arr{v}}}} \ket*{\phi_{\arr{m},\arr{v}}} = \frac{\ket{\psi_{\arr{m}}}}{\norm{c}_1} =  \frac{K_{\arr{m}} \ket{\phi_0}}{\norm{c}_1} \ ,
\end{equation}
where we used the decomposition of the exact unnormalized state $\sum_{\arr{v}} c_{\arr{v}}  \ket*{\phi_{\arr{m},\arr{v}}}= \ket{\psi_{\arr{m}}} = K_{\arr{m}} \ket{\phi_0}$. Then we have
\begin{equation}
\sum_{\arr{m}} \mathbb{E}(\bra*{\omega_{q,\arr{m}}}\ket*{\omega_{p,\arr{m}}}) = \sum_{\arr{m}} \frac{\bra{\phi_0} K^\dagger_{\arr{m}} K_{\arr{m}}  \ket{\phi_0}}{\norm{c}^2_1} = \frac{1}{\norm{c}^2_1}
\end{equation}
so
\begin{equation}
\sum_{\arr{m},\arr{m}'} \mathbb{E}\left(    \bra*{\omega_{q,\arr{m}}}\ket*{\omega_{p,\arr{m}}}   \bra*{\omega_{s,\arr{m}'}}\ket*{\omega_{r,\arr{m}'}} \right)  = \frac{1}{\norm{c}^4_1}\ .
\end{equation}
Then we have 
\begin{equation}
B^2 = \sum_{(p,q,r,s)\in \mathcal{A}} \norm{c}_1^{-4} + B' 
	= \frac{k(k-1)(k-2)(k-3)}{ \norm{c}_1^4} + B' \label{eq:Bsquared}.
\end{equation}
We break $B'$ into terms depending on which indices are shared
\begin{align}
B' = B_{p=r,q\neq s} +  B_{p=s,q \neq r} + B_{p\neq r,q = s} + B_{p\neq s, q=r} + B_{p=r,q=s}+B_{p=s,q=r}\ .
\end{align}
We deal with each in turn, using $\mathcal{A}_n$ to mean the set of length $n$ tuples where all entries are distinct. We can check that 
\begin{equation}
B_{p=r,q\neq s} = \sum_{\arr{m},\arr{m}'}  \sum_{(p,q,s) \in \mathcal{A}_3}   \bra*{\omega_{q,\arr{m}}}\ket*{\omega_{p,\arr{m}}}   \bra*{\omega_{s,\arr{m}'}}\ket*{\omega_{p,\arr{m}'}} = \overline{B}_{p\neq r, q=s}
\end{equation}
and similarly we can check that $B_{p=s,q\neq r} = \overline{B}_{p\neq s,q= r}$. Then we must bound
\begin{equation}
\mathbb{E}(B' ) = 2 \Re[\mathbb{E}(B_{p=r,q\neq s} + B_{p=s,q\neq r})] + \mathbb{E}(B_{p=r,q=s}+B_{p=s,q=r}) \label{eq:expBprime} \ . 
\end{equation}
We have
\begin{align}
\frac{\mathbb{E}(B_{p=s,q\neq r})}{k(k-1)(k-2)} &= \sum_{\arr{v},\arr{w},\arr{x}} \frac{\abs{c_{\arr{v}}c_{\arr{w}}c_{\arr{x}} }}{\norm{c}_1^3}\sum_{\arr{m},\arr{m}'}    \frac{c^*_{\arr{v}}  c^*_{\arr{w}} c_{\arr{w}}  c_{\arr{x}}}
{\abs{c^*_{\arr{v}}  c^*_{\arr{w}} c_{\arr{w}}  c_{\arr{x}}} }
\bra*{\phimv{m}{v}} \ket*{\phimv{m}{w}} \bra*{\phimpv{m}{w}} \ket*{\phimpv{m}{x}} \\
		&=  \sum_{\arr{w}} \frac{\abs{c_{\arr{w}}}}{\norm{c}_1^3} \sum_{\arr{m},\arr{m}'}   \left( \sum_{\arr{v}} c^*_{\arr{v}} \bra*{\phimv{m}{v}}\right) \op*{\phimv{m}{w}}{\phimpv{m}{w}}\left( \sum_{\arr{x}} c_{\arr{x}} \ket*{\phimpv{m}{x}}\right)  \\
				& = \sum_{\arr{w}} \frac{\abs{c_{\arr{w}}}}{\norm{c}_1^3}  \cdot  \abs{ \sum_{\arr{m}} \bra*{\psim}\ket*{\phimv{m}{w}} }^2 \ .
\end{align}
For $B_{p=r,q\neq s}$ we obtain the same bound on the absolute value by the triangle inequality,
\begin{align}
\abs{\frac{\mathbb{E}(B_{p=r,q\neq s})}{k(k-1)(k-2)}}
& \leq \sum_{\arr{w}} \frac{\abs{c_{\arr{w}}}}{\norm{c}_1^3}   \cdot \abs{\sumvec{m}   \bra*{\psim}\ket*{\phimv{m}{w}}} \cdot \abs{\sum_{\arr{m}'} \bra*{\psi_{\arr{m}'}}\ket*{\phimpv{m}{w}} } \\
& =  \sum_{\arr{w}} \frac{\abs{c_{\arr{w}}}}{\norm{c}_1^3}  \cdot \abs{\sumvec{m}   \bra*{\psim}\ket*{\phimv{m}{w}}}^2\ .
\end{align}
For $B_{p=s,q=r}$, we have
\begin{align}
\mathbb{E}(B_{p=s,q= r}) &= \mathbb{E}\left( \sum_{\arr{m},\arr{m}'}  \sum_{(p,q) \in \mathcal{A}_2}   \bra*{\omega_{q,\arr{m}}}\ket*{\omega_{p,\arr{m}}}   \bra*{\omega_{p,\arr{m}'}}\ket*{\omega_{q,\arr{m}'}} \right) \\
		&= k(k-1) \sum_{\arr{v},\arr{w}}  \sum_{\arr{m},\arr{m}'} \frac{\abs{c_{\arr{v}} c_{\arr{w}}}}{\norm{c}_1^2} \frac{c_{\arr{v}} c_{\arr{w}} c^*_{\arr{v}} c^*_{\arr{w}}}{\abs{c_{\arr{v}} c_{\arr{w}}}^2} \bra*{\phimv{m}{v}} \ket*{\phimv{m}{w}}     \bra*{\phimpv{m}{w}}\ket*{\phimpv{m}{v}} \\
		&=  k(k-1)  \sum_{\arr{v},\arr{w}}   \frac{\abs{c_{\arr{v}} c_{\arr{w}}}}{\norm{c}_1^2} \sumvec{m} \bra*{\phimv{m}{v}} \ket*{\phimv{m}{w}}    \sum_{\arr{m}'}  \bra*{\phimpv{m}{w}}\ket*{\phimpv{m}{v}}  \\
		& =  k(k-1)  \sum_{\arr{v},\arr{w}}   \frac{\abs{c_{\arr{v}} c_{\arr{w}}}}{\norm{c}_1^2} \cdot \abs{ \sumvec{m} \bra*{\phimv{m}{v}} \ket*{\phimv{m}{w}}    }^2\ .
\end{align}
Define an ordering of the trajectories $\arr{m}$, so that each is labelled by some index $i$, $\arr{m}_i$, such that $1\leq i \leq 2^T$ (since each adaptive channel has Kraus rank 2, and there are $T$ such channels). Then letting $d$ be the dimension of the Hilbert space, define matrices $M_{\arr{v}}$ of dimension $d \times 2^T$, where the $i$-th column is the vector $\ket*{\phimv{m_i}{v}}$,
\begin{equation}
M_{\arr{v}} = \left(\ket*{\phimv{m_1}{v}}  \, \ket*{\phimv{m_2}{v}} \cdots  \ket*{\phimv{m_{2^T}}{v}} \right)\ .
\end{equation}
Then we can express the sum $ \sumvec{m} \bra*{\phimv{m}{v}} \ket*{\phimv{m}{w}} $ in terms of the Frobenius inner product,
\begin{equation}
\langle M_{\arr{v}}, M_{\arr{w}} \rangle_F = \Tr(M_{\arr{v}}^\dagger M_{\arr{w}}) =  \sumvec{m} \bra*{\phimv{m}{v}} \ket*{\phimv{m}{w}}\ .
\end{equation}
Then by the Cauchy-Schwarz inequality
\begin{align}
\abs{\langle M_{\arr{v}}, M_{\arr{w}} \rangle_F}^2 & \leq \langle M_{\arr{v}}, M_{\arr{v}} \rangle_F \langle M_{\arr{w}}, M_{\arr{w}} \rangle_F \\
								& = \left(\sumvec{m} \bra*{\phimv{m}{v}} \ket*{\phimv{m}{v}}\right) \left( \sum_{\arr{m}'}  \bra*{\phimpv{m}{w}}\ket*{\phimpv{m}{w}} \right) \\
								& = 1 \ , 
\end{align}
where the last equality follows since the operators $K_{\arr{m},\arr{v}}$ that define the vectors $\ket*{\phimv{m}{v}}= K_{\arr{m},\arr{v}}\ket{\phi_0}  $ satisfy  $\sumvec{m}  K^\dagger_{\arr{m},\arr{v}} K_{\arr{m},\arr{v}} = \id$. We therefore have
\begin{equation}
\mathbb{E}(B_{p=s,q= r})  \leq  k(k-1)  \sum_{\arr{v},\arr{w}}   \frac{\abs{c_{\arr{v}} c_{\arr{w}}}}{\norm{c}_1^2} = k(k-1) 
\end{equation}
since $\frac{\abs{c_{\arr{v}} c_{\arr{w}}}}{\norm{c}_1^2}$ form a probability distribution. Similarly for $B_{p=r,q=s}$ we obtain
\begin{align}
\abs{\mathbb{E}(B_{p=s,q= r})} 
	& = \abs{k(k-1)  \sum_{\arr{v},\arr{w}}   \frac{\abs{c_{\arr{v}} c_{\arr{w}}}}{\norm{c}_1^2} \sumvec{m} \bra*{\phimv{m}{v}} \ket*{\phimv{m}{w}}    \sum_{\arr{m}'}  \bra*{\phimpv{m}{v}}\ket*{\phimpv{m}{w}}  }\\
		&\leq k(k-1)  \sum_{\arr{v},\arr{w}}   \frac{\abs{c_{\arr{v}} c_{\arr{w}}}}{\norm{c}_1^2} \cdot \abs{\sumvec{m} \bra*{\phimv{m}{v}} \ket*{\phimv{m}{w}}   }^2 \leq k(k-1) \ , 
\end{align}
where the first inequality comes from the triangle inequality.

Taking the absolute value of~\cref{eq:expBprime}, using the triangle inequality and substituting all the obtained expressions back in we have
\begin{align}
\abs{\mathbb{E}(B')} & \leq 4k(k-1)(k-2) \sum_{\arr{w}} \frac{\abs{c_{\arr{w}}}}{\norm{c}_1^3} \cdot  \abs{\sumvec{m}   \bra*{\psim}\ket*{\phimv{m}{w}}}^2 + 2  k(k-1) \\
			& = 4k(k-1)(k-2) \frac{\cte}{\norm{c}_1^4} +2k(k-1)  \label{eq:expectedBprime}
\end{align} 
where we define $\cte =  \norm{c}_1\sum_{\arr{w}} \abs{c_{\arr{w}}} \cdot  \abs*{\sumvec{m}   \bra*{\psim}\ket*{\phimv{m}{w}}}^2$.
Substituing the expressions~\eqref{eq:Bsquared},~\eqref{eq:expectedBprime} and~\eqref{eq:squaredexpectedB} into the expression for the variance of $\Tr[\Omega]$
\begin{align}
\Var(\Tr[\Omega])  &=   \frac{\norm{c}_1^4}{k^4} \left(\mathbb{E}(B^2) - (\mathbb{E}(B))^2 \right) \\
			&\leq  \frac{\norm{c}_1^4}{k^4} \Bigg[ \frac{k(k-1)(k-2)(k-3)}{ \norm{c}_1^4} +  4k(k-1)(k-2) \frac{\cte}{\norm{c}_1^4} + 2  k(k-1) \\
			& \quad \quad  \quad \quad -  \frac{k^2(k-1)^2}{\norm{c}_1^4} \Bigg] \\
			& = 4\left( \frac{1}{k} - \frac{3}{k^2} + \frac{2}{k^3}\right)\cte + 2\norm{c}^4_1\left(\frac{1}{k^2} - \frac{1}{k^3} \right) - \left(\frac{4}{k} - \frac{10}{k^2} + \frac{6}{k^3}\right) \\
			& \leq \frac{4(\cte - 1)}{k} + 2 \left( \frac{\norm{c}_1^2}{k} \right)^2 + \frac{10 - 12\cte}{k^2} +  O\left( \frac{\cte}{k^3}\right)\ , 
\end{align}
which is the result we set out to prove.
\end{proof}
This result is identical to that of~\cite{Seddon_2021}, except that we have defined $\cte$ as
\begin{equation}
\cte =  \norm{c}_1\sum_{\arr{w}} \abs{c_{\arr{w}}} \cdot \left| \sumvec{m}   \bra*{\psim}\ket*{\phimv{m}{w}}\right|^2 \ .
\end{equation}
The value of $\cte$ may be more difficult to compute than was the case for~\cite{Seddon_2021}. It may be possible to upper bound it by considering a tensor product of corresponding resource states. Nevertheless we can argue that $\cte/  \norm{c}^2_1 $ should not be large, since for non-trivial circuits we should expect $\ket*{\psim}$ to have small overlap with any Gaussian state.

\section{Sampling simulator technical details \label{sec:sampling_simulator_details}}

Here we detail the subroutines used in the simulator for sampling from non-Gaussian circuits laid out in~\cref{sec:algorithm}.
In~\cref{sec:convex_unitary_channel_simulation} we consider a circuit composed on convex-unitary channels. In~\cref{sec:modified_simulator_sparsifications} we modify the simulator to account for the use of adaptive channels in the circuit. 

\subsection{Convex-unitary channel simulation \label{sec:convex_unitary_channel_simulation}}

In this appendix, we give further technical details of classical simulation subroutines for near-Gaussian circuits in order to justify the statements about runtime given in the main text. We first outline an exact algorithm for bit-string sampling in~\cref{algo:samplingExact}. We then give the approximate algorithm, \cref{algo:samplingApprox}, whose runtime is reduced relative to the exact algorithm in two ways. First, sparsification reduces the effective rank of the state vector tracked by the algorithm. Then, the use of the \textsc{FastNorm} subroutine~\cite{bravyiComplexityQuantumImpurity2017a} reduces the runtime to estimate a Born rule probability from quadratic to linear in the channel decomposition cost. 

\begin{algorithm}[htbp]
\begin{algorithmic}[1]
\algrenewcommand\algorithmicrequire{\textbf{Input:}}
\algrenewcommand\algorithmicensure{\textbf{Output:}} 
\Require List of convex-unitary channels~$\{ \chan_1, \ldots, \chan_T \}$ defined on a system of~$n$ qubits and decomposable by an oracle~$\mathcal{D}$; integer~$w \leq n$.
\Ensure Bit string of length~$w$.
\For{$t \gets 1$ to~$T$}
\State \label{line:alg1_Ut}~$\mathbb{U}_t \gets \mathcal{D}(\chan_t)$
\Comment $\mathbb{U}_t$ is a~$\chi_t$-length list of coefficients~$c_{t,j}$ and (decriptions of) Gaussian unitaries~$G_{t,j}$ as in~\cref{eq:DUt}. \label{step:exactoracle}
\EndFor
\State \label{line:alg1_evolve}~$\ket{\Psi} \gets~$ \Call{EvolveCircuit}{$\{\mathbb{U}_1, \ldots,\mathbb{U}_T \}$,~$\ket{0^n}$ } \label{step:exactstate}

\Comment The state vector is stored as a length~$k$ list of Gaussian states~$\ket{g_i}$ and amplitudes~$\alpha_i$ including phase. 
\State $\ket{\Psi_0} \gets \ket{\Psi}$
\For{$b \gets 1$ to~$w$} \label{line:alg1_forstart}
\State $\Pr(x_b = 0) \gets  \norm{(I+Z_b)\ket{\Psi_{b-1}}/2}^2$ \label{step:exactnorm}
\Comment Computed by evaluating the inner product~$\bra{\Psi_{b-1}} (I+Z_b)/2 \ket{\Psi_{b-1}}$, expanded as a linear combination of Gaussian inner products.
\State $x_b \gets 0$,~$p_b \gets \Pr(x_b = 0)$ with probability~$\Pr(x_b = 0)$, otherwise~$x_b \gets 1$, $p_b\gets 1 - \Pr(x_b = 0)$. \label{step:exactsample}
\State $\ket{\Psi_b} \gets  \frac{I+(-1)^{x_b} Z}{2 \sqrt{p_b} } \ket{\Psi_{b-1}}$ \label{step:exactproject}
\EndFor \label{line:alg1_forend}
\State \Return~$x = (x_1,x_2,\ldots,x_w)$
\end{algorithmic}
\caption{Exact bit-string sampling algorithm}\label{algo:samplingExact}
\end{algorithm}

Both the exact and approximate variants of the algorithm make use of a procedure \linebreak~$\textsc{EvolveCircuit}$ (\cref{algo:evolveCircuit}) that updates a representation of the state vector under evolution by a sequence of gates described as a sum over Gaussian operations. This procedure can be built using the phase-sensitive Gaussian update method of either~\cite{reardonsmith2024improved} or~\cite{Dias_2024}. Here we assume the procedure described in~\cite{Dias_2024} as ``evolve'' (after mapping into the qubit picture via the Jordan-Wigner transformation). We refer to this subroutine as \textsc{GaussianEvolve}, to differentiate from \textsc{EvolveCircuit}: \textsc{GaussianEvolve} is used to evolve (Gaussian states) under Gaussian unitaries, and \textsc{EvolveCircuit} is used to evolve under possibly non-Gaussian unitaries. The procedure \textsc{GaussianEvolve} was shown in~\cite{Dias_2024} to have runtime~$O(n^3)$. If the input to \textsc{EvolveCircuit} is a list of~$T$ operators~$\{U_t\}_{t=1}^T$, then during the update associated with~$U_t$ the number of calls to  \textsc{GaussianEvolve} is the rank~$\chi$ of the state prior to the update, multiplied by the rank~$\chi_t$ of the decomposition of the current operator~$U_t$. The rank is non-decreasing, so the most expensive update will be the final step, involving~$\prod_{t=1}^T \chi_t$ calls to \textsc{GaussianEvolve}. (In the context of our bit-string sampling algorithm, we assume without loss of generality that the initial state input to \textsc{EvolveCircuit} has rank~$1$.) Therefore \textsc{EvolveCircuit} has runtime upper bounded by~$O(n^3 T\prod_{t=1}^T \chi_t)$. 

\begin{algorithm}[htbp]
\begin{algorithmic}[1]
\algrenewcommand\algorithmicrequire{\textbf{Input:}}
\algrenewcommand\algorithmicensure{\textbf{Output:}} 
\Require A description~$\mathbb{U}_t = \{ (c_{t,1}, R_{t,1}),\ldots,(c_{t,{\chi_t}},R_{t,{\chi_t}}) \}$ of a (possibly non-Gaussian) linear combination $U_t = \sum_{i=1}^{\chi_t} c_{t,i} G_{t,i}$ of Gaussian unitary operators~$G_{t,i}$ with~$c_{t,i}\in \mathbb{C}$, where each~$R_{t,i} \in O(2n)$ is the orthogonal matrix representing the Gaussian unitary~$G_{t,i}$ (see~\cref{eq:UR_R_relation}). The initial state~$\ket{\psi_0}=\sum_{j=1}^{\chi_0} \alpha_{0,j} \ket{\phi_{0,j}}$ is stored as a list~$d(\psi_0) = \{(\alpha_{0,1},d(\phi_{0,1})), \ldots, (\alpha_{0,\chi_0},d(\phi_{0,\chi_0}))\}$ of coefficients and suitable phase-sensitive descriptions~$\{d(\phi_{0,j})\}_j$ respectively for the Gaussian states~$\{\ket{\phi_j}\}_j$~\cite{Dias_2024,reardonsmith2024improved}. 
\Procedure{EvolveCircuit}{$\{\mathbb{U}_1, \ldots,\mathbb{U}_T\}, d(\psi_0)$}
\State $\chi \gets \chi_0$
\For{$t=1$ to~$T$} 
	\State $j' \gets 1$
	\For{$i = 1$ to~$\chi_t$}
		\For{$j = 1$ to~$\chi$}
			\State $\alpha_{t,j'} \gets c_{t,i} \cdot\alpha_{t-1,j}~$
			\State $d(\phi_{t,j'}) \gets~$\Call{GaussianEvolve}{$R_{t,i},d(\phi_{t-1,j})$}
			\State $j' \gets j' + 1$
		\EndFor
	\EndFor
	\State $\chi \gets \chi \cdot \chi_t$
\EndFor
\State \Return~$d(\psi_T) = \{(\alpha_{T,j}, d(\phi_{T,j})) \}_{j=1}^\chi$ 
\Comment The list~$d(\psi_T)$ represents the state~$\ket{\psi_T}$.
\EndProcedure
\end{algorithmic}
\caption{Subroutine for unitary evolution of pure states}\label{algo:evolveCircuit}
\end{algorithm}

We return to analysing the exact sampling algorithm, \cref{algo:samplingExact}. 
It is straightforward to check the correctness of the algorithm. In Line~\ref{step:exactoracle}, the oracle returns a description of some unitary operator~$U_{t,j}$ with probability~$p_{t,j}$ such that~$\chan_t = \sum_{j} p_{t,j} \curlU_{t,j}$, and so the sequence of~$T$ calls to the oracle returns a description of a sequence of~$T$ unitary channels~$\smash{U_{\arr{j}} = U_{t,j_t} \ldots U_{1,j_1}}$ with probability~$p_{\arr{j}} = \prod_{t=1}^T p_{t,j_t}$. Then, \textsc{EvolveCircuit} computes the exact evolution of the initial state under this sequence of unitary operators, so by Line~\ref{step:exactstate} the algorithm has prepared a description of the final state~$\smash{\ket*{\Psi_{\arr{j}}} = U_{\arr{j}}\ket{0^n}}$ with probability~$\smash{p_{\arr{j}}}$. The density operator representing the ensemble of states generated by the algorithm is then
\begin{align}
\rho & = \sum_{\arr{j}} p_{\arr{j}} \curlU_{\arr{j}}(\op{0^n}) \\
	  &  = \left(\left(\sum_{j_T} p_{t,j_t} \curlU_{t,j_t}\right)\circ \ldots  \circ \left(\sum_{j_1} p_{2,j_2} \curlU_{2,j_2}\right) \circ \left(\sum_{j_1} p_{1,j_1} \curlU_{1,j_1}\right) \right) (\op{0^n}) \\
	  & = \left( \chan_T \circ \ldots \circ \chan_2 \circ \chan_1 \right) (\op{0^n} ) = \chan  (\op{0^n} ),
\end{align}
which is exactly the final mixed state that we want to simulate, see~\cref{def:bitsampling_approx}. The final task of the simulator is to draw a bit-string with probability
\begin{align}
	p(x) = \Tr(\Pi_x \chan(\op{0^n}))= \sum_{\arr{j}} p_{\arr{j}}(x) \ ,
\end{align}
where~$\smash{p_{\arr{j}}(x) = \| \Pi_x \ket*{\Psi_{\arr{j}} }\|^2}$ is the probability of sampling the bit-string~$x$ having drawn the pure state~$\smash{\ket*{\Psi_{\arr{j}}}}$ from the ensemble. The conditional probability of sampling bit $x_b$ having first sampled bits $(x_1,\ldots,x_{b-1})$ is computed by taking a ratio of Born rule probabilities, 
\begin{align}
\Pr(X_{b} = x_{b}| X_1=x_1,\ldots,X_{b-1} = x_{b-1})  & = \frac{\|\Pi_{x_b,b} \ldots \Pi_{x_1,1} \ket{\Psi}\|^2}{\norm{\Pi_{x_{b-1},b-1} \ldots \Pi_{x_1,1} \ket{\Psi}}^2} \\
& = \| \Pi_{x_b,b} \ket{\Psi_{b-1}} \|^2 \ ,
\end{align}
where~$\Pi_{x_b,b} = (I + (-1)^{x_b} Z_b)/2$. By inspection we can see that Lines~\ref{step:exactnorm}, \ref{step:exactsample} and \ref{step:exactproject} of~\cref{algo:samplingExact} exactly reproduce the procedure of sampling bitwise from the target distribution. 

To analyse the runtime of~\cref{algo:samplingExact}, first note in Line \ref{line:alg1_Ut} that the oracle returns the description of the $t$-th unitary operator as a list of $\chi_t$ complex coefficients paired with $2n\times 2n$ orthogonal matrices, so we can assume that to store $T$ of these description takes time $O(T \chi n^2) $, where $\chi = \prod_{t=1}^T \chi_t$. Next, in Line \ref{line:alg1_evolve} there is a single call to \textsc{EvolveCircuit} which takes time $O(T \chi n^3)$. In the for-loop where bit-strings are counted (Lines \ref{line:alg1_forstart} to \ref{line:alg1_forend}), each norm is computed by evaluating $\chi^2$ Gaussian inner products, each taking time $O(n^3)$~\cite{Dias_2024,reardonsmith2024improved}, whereas the time for updating $\ket{\Psi_{b-1}}$ to $\ket{\Psi_b}$ by projection takes time $O(\chi n^3)$. Overall, then, for circuits with significant magic, the runtime for exact sampling is dominated by the quadratic scaling with the rank for computing $w$ norms, so will be  $O(w \chi^2 n^3)$. Note that the rank $\chi$ depends on the particular unitary trajectory sampled, so the average runtime will be $\smash{\sum_{\arr{j}} p_{\arr{j}} O(w \chi^2_{\arr{j}} n^3)}$.

Similarly to the approaches in previous work based on stabilizer~\cite{Bravyi_2016_fastnorm,Bravyi_2019,Seddon_2021} and Gaussian~\cite{bravyiComplexityQuantumImpurity2017a,Dias_2024,reardonsmith2024improved} rank/extent, we can improve the scaling with rank from quadratic to linear by replacing exact computation of probabilities with fast norm estimation~\cite{Bravyi_2016_fastnorm,bravyiComplexityQuantumImpurity2017a}, and reduce the effective rank by means of sparsification~\cite{Bravyi_2019,Seddon_2021,Dias_2024}. We propose the approximate bit-string sampling algorithm given in~\cref{algo:samplingApprox}. The procedure is essentially the same as that described for the stabilizer case for pure states in~\cite{Bravyi_2019}, mixed states in~\cite{Seddon_2021} and for convex-unitary channels in~\cite{seddonThesis}. The key differences are in the subroutines adapted for the fermionic setting, and that we can implement the oracle~$\mathcal{D}$ efficiently for some important cases based on the analytic decompositions we provide in~\cref{sec:decomp_2qbit_gates_mom,sec:noisydecomps}. In several cases these decompositions prove to be optimal.  Aside from the phase-sensitive subroutine \textsc{EvolveCircuit}, already described, we make use of~\cref{lem:fastnorm} and the associated $\textsc{FastNorm}$ procedure used to estimate norms of linear combinations of Gaussian states. 

\begin{algorithm}[htbp]
\begin{algorithmic}[1]
\algrenewcommand\algorithmicrequire{\textbf{Input:}}
\algrenewcommand\algorithmicensure{\textbf{Output:}} 
\Require List of convex-unitary channels~$\{ \chan_1, \ldots, \chan_T \}$ defined on a system of~$n$ qubits and decomposable by an oracle~$\mathcal{D}$; integer~$w \leq n$; $\delta$,~$\varepsilon >0$ and~$p \in (0,1)$.
\Ensure Bit-string of length~$w$.
\State $E \gets 1$ \Comment Initialize cost variable.
\For{$t \gets 1$ to~$T$}
\State $\mathbb{U}_t \gets \mathcal{D}(\chan_t)$ \Comment  $\mathbb{U}_t$ is a~$\chi$-length list of coefficients $c_{t,j}$ and (descriptions of) Gaussian unitaries~$G_{t,j}$ as in~\cref{eq:DUt}.
\State $E \gets E \cdot \norm{c_t}^2_1~$  
\EndFor
\State $k \gets \left\lceil 4 \frac{E}{\delta} \right\rceil$ \Comment This assumes~$\delta$ is above a critical precision~\cite{Seddon_2021}, see~\cref{sec:SPARSE}. 
\State $\mathbb{V} \gets~$ \Call{SparsifyCircuit}{$k$,~$\{\mathbb{U}_1,\ldots \mathbb{U}_T\}$} \Comment~$\mathbb{V}$ is a list of~$k$ (descriptions of) unitary Gaussian circuits~$\mathbb{G}_i$ incorporating a phase and normalization factor.

\Comment Each~$\mathbb{G}_i$ is a~$T$-length list of (descriptions of) Gaussian unitary gates~$G_{i,t}$.
\State $\ket{\Psi} \gets~$ \Call{EvolveCircuit}{\{$\mathbb{V}$\},~$\ket{0^n}$ }

\Comment The state vector is stored as a length~$k$ list of Gaussian states~$\ket{g_i}$ with suitable phase-sensitive description and amplitudes~$\alpha_i$ including phase. 
\State $N \gets~$ \Call{FastNorm}{$\ket{\Psi}$,~$\varepsilon$,~$p$}
\State $\ket{\Psi_0} \gets \ket{\Psi}/N$
\For{$b \gets 1$ to~$w$}
\State $\Pr(x_b = 0) \gets~$ \Call{FastNorm}{$(I+Z_b)\ket{\Psi_{b-1}}/2$,~$\varepsilon$,~$p$}
\If{$\Pr(x_b = 0) >\frac{1}{2}$} 
	\State $\Pr(x_b = 0) \gets~$~$1-$\Call{FastNorm}{$(I-Z_b)\ket{\Psi_{b-1}}/2$,~$\varepsilon$,~$p$}
	
\EndIf
\State $x_b \gets 0$,~$p_b \gets \Pr(x_b = 0)$ with probability~$\Pr(x_b = 0)$, otherwise~$x_b \gets 1$, $p_b\gets 1 - \Pr(x_b = 0)$.
\State $\ket{\Psi_b} \gets  \frac{I+(-1)^{x_b} Z}{2\sqrt{p_b}} \ket{\Psi_{b-1}}$
\EndFor
\State \Return~$x = (x_1,x_2,\ldots,x_w)$
\end{algorithmic}
\caption{Approximate bit-string sampling algorithm}\label{algo:samplingApprox}
\end{algorithm}
\begin{algorithm}[htbp]
\begin{algorithmic}[1]
\Require Integer~$k>0$; list of unitary operations described as a linear combination of Gaussian unitary operators, as per the input to~\cref{algo:evolveCircuit}.
\Ensure List of~$k$ Gaussian circuits described as a sequence of Gaussian gates with phase and normalization factor.
\algrenewcommand\algorithmicrequire{\textbf{Input:}}
\algrenewcommand\algorithmicensure{\textbf{Output:}} 
\Procedure{SparsifyCircuit}{$k,\{\mathbb{U}_1,\ldots,\mathbb{U}_T\}$}
\For{$i \gets 1$ to~$k$}
	\State Initialise empty circuit~$\mathbb{G}_i$.
	\State $\alpha_i \gets 1$
	\For{$t \gets 1$ to~$T$}
		\Comment Recall~$\mathbb{U}_t = \{ (c_{t,1}, R_{t,1}),\ldots,(c_{t,{\chi_t}},R_{t,{\chi_t}}) \}$,~$\norm{c_t}_1 = \sum_j \abs{c_{t,j}}$
		\State Select the~$j$-th term in the decomposition with probability~$\abs{c_{t,j}}/\norm{c_t}_1$, append~$R_{t,j}$ to~$\mathbb{G}_i$.
		\State $\alpha_i \gets \alpha_i \cdot c_{t,j} / \abs{c_{t,j}}$
	\EndFor
	\State $\alpha_i \gets \frac{\prod_{t=1}^T \norm{c_t}_1}{k} \cdot \alpha_i$
\EndFor
\State \Return~$\mathbb{V} = \{ (\alpha_1,\mathbb{G}_1),\ldots, (\alpha_k,\mathbb{G}_k) \}$
\EndProcedure
\end{algorithmic}
\caption{Subroutine for sparsifying a unitary circuit~\cite{Bravyi_2019}\label{algo:sparsifyCircuit}}
\end{algorithm}

To be concrete, we give the specific implementation of \textsc{SparsifyCircuit} we propose in~\cref{algo:sparsifyCircuit}, following the approach used in the sum-over-Cliffords algorithm from~\cite{Bravyi_2019}. Following the arguments used to prove the ensemble sampling lemmas in~\cite[Theorem 16, Lemmas 17 and 18]{Seddon_2021} and~\cite[Appendix E]{seddonThesis} it can be shown that randomly generating a~$k$-term sparsified state~$\ket{\Omega}$ by following the procedure \textsc{SparsifyCircuit} to sparsify a circuit~$U$, evolving a Gaussian initial state~$\ket{g_0}$ under the sparsified circuit, and normalizing the final state, we are in effect drawing from an ensemble~$\rho_1 = \sum_{\Omega} \Pr(\Omega) \op{\Omega}/\Tr[\Omega]$ close in trace-norm to the exact final state,
\begin{equation}
\norm{\rho_1 - \curlU(\op{g_0})}_1 \leq 2\frac{\norm{c}_1^2}{k}  + \sqrt{\Var(\Tr[\Omega])},
\end{equation}
where~$\norm{c}_1 = \prod_{t=1}^T \norm{c_t}_1$ is the product of the~$L^1$-norm for each decomposition of a non-Gaussian gate in the circuit. The variance in the normalization of $\Omega$ can be bounded  as (see Appendix~\ref{sec:sparseTechnical})
\begin{equation}
 \Var(\Tr[\Omega]) \leq \frac{4(\cte - 1)}{k} + 2 \left( \frac{\norm{c}_1^2}{k} \right)^2 + \frac{10 - 12\cte}{k^2} +  O\left( \frac{\cte}{k^3}\right).
 \end{equation}
Here~$\cte$ is a constant for the effective decomposition of the final state, 
\begin{equation}
\cte = \norm{c}_1 \sum_j \abs{c_j} \cdot \abs*{ \bra{g_0}U^\dagger G_j \ket{g_0}}^2, \quad U = \sum_j c_j G_j.
\end{equation}
In Appendix~\ref{sec:sparseTechnical} we proved a more general result for circuits involving intermediate Gaussian measurement and feed-forward, which reduces to the setting described here by removing non-unitary operations. This is analogous to the result for unitary circuits given in~\cite{seddonThesis}, except that stabilizer states and operations are replaced with Gaussian ones.

Having introduced each subroutine, we can now argue for the runtime and accuracy guarantees given in~\cref{thm:bitsamplingthm}. Once again, these can be derived by following the arguments of~\cite{bravyiComplexityQuantumImpurity2017a,Bravyi_2019,Seddon_2021,seddonThesis} and substituting stabilizer states/operations for Gaussian ones as necessary, so for brevity we omit a full technical proof and instead sketch the reasoning.
\begin{enumerate}[1)]
\item First define an algorithm \textsc{SparseExact}. It computes probabilities exactly, then it draws bit-strings as per the exact bit-string sampling procedure in~\cref{algo:samplingExact}, but for a sparsified and normalized final state~$\ket{\Omega}/\norm{\ket{\Omega}}$ output from \textsc{SparsifyCircuit}~(\cref{algo:sparsifyCircuit}). Call the output distribution from this algorithm~$p_{SE}$.
\item Use the ensemble sampling lemma to show that since~$\norm{\rho_1 - \curlU(\op{g_0})}_1 \leq \delta$, for some~$\delta$, we have~$\norm{p_{SE} - p}_1 \leq \delta$, where~$p$ is the true distribution. Note that~$\delta$ can be made arbitarily small by increasing the number of terms in the sparsification.
\item  \textsc{SparseExact} is identical to the approximate bit-string sampling algorithm~\cref{algo:samplingApprox}, except that the exact computations of bitwise probabilities are replaced by calls to \textsc{FastNorm}, so that the bit-string~$x$ is returned with probability~$p_{sim}$. We set the parameters of \textsc{FastNorm} such that~$\norm{p_{sim} - p_{SE}}_1 \leq \epsilon'~$ with probability at least~$1-p_{\text{fail}}$.
\item Then~$\norm{p_{sim} - p}_1 \leq \delta + \epsilon'$ with probability at least~$1-p_{\text{fail}}$.
\end{enumerate}

The runtime of~\cref{algo:samplingApprox} can be dominated by the call to \textsc{EvolveCircuit} or by the sequence of calls to \textsc{FastNorm}, depending on the number of qubits~$w$ measured, the depth of the circuit~$T$ and the size~$n$ of the system. From our earlier arguments, for a given run of the algorithm, the call to \textsc{EvolveCircuit} takes time~$\mathcal{O}(T k n^3)$ where~$k$ is the number of terms in the sparsified circuit decomposition. It can be shown~\cite{Seddon_2021} that to make~$2w$ calls to \textsc{FastNorm} achieving estimation error~$\epsilon'$ with total failure probability no larger than~$p_{\text{fail}}$, we must set the parameters for \textsc{FastNorm} to~$\epsilon' \rightarrow \epsilon/3w$ and~$p_f \rightarrow p_{\text{fail}}/2w$. Then by~\cref{lem:fastnorm} the total runtime for~$2w$ calls to \textsc{FastNorm} for a~$k$-term superposition is~$O(w^3 n^{7/2} \epsilon^{-2} \log(w/p_\text{fail}) k)$. Now recall that~$k$ is a function of the cost~$E$ and the desired precision~$\delta$ of the sparsification. For simplicity assume the favourable case where~$\delta~$ is above a critical precision (see~\cref{sec:SPARSE}) where we set~$k \geq 4 E\delta$. The total cost~$E=\norm{c}^2_1$ is the~$L^1$-norm squared for the decomposed unitary circuit returned by the oracle~$\mathcal{D}$. Recall that the oracle is probabilistic, and has a cost function~$E(\mathcal{D},\chan_t)$ defined over the set of channels involved in the target circuit. Then the average number of terms in the sparsification is~$\mathbb{E}(k) = \delta^{-1} \prod_{t=1}^T E(\mathcal{D},\chan_t)~$. Therefore the total runtime~$\tau$ for a generating a single bit-string of length~$w$ has expected value
\begin{equation}
\mathbb{E}(\tau) = {O}\left(\left(T n^3 + w^3 n^{7/2} \epsilon^{-2} \log(w/p_f^{-1})\right) \delta^{-1}  \prod_{t=1}^T E(\mathcal{D},\chan_t) \right).
\end{equation}
It can be shown~\cite{Seddon_2021} that given total error budget~$\delta' = \delta + \epsilon$, where~$\delta$ is the sparsification error and~$\epsilon$ is the error in fast norm estimation, we achieve optimal runtime by setting~$\delta =\delta'/3$ and~$\epsilon=2\delta'/3$. If the oracle is equimagical over the set of channels composed in the circuit, then the runtime becomes deterministic. If the oracle is optimal, then~$\mathbb{E}(k) = \delta^{-2} \prod_{t=1}^T \Xi(\mathcal{D},\chan_t)$. Finally we note that arbitrary precision (beyond the critical precision described in~\cref{sec:SPARSE})  can be achieved at a cost of the~$\delta^{-1}$ factor being replaced by~$\delta^{-2}$ (albeit with often favourable constants). This justifies the statement in~\cref{thm:bitsamplingthm}.

In~\cref{sec:Udecomp} and~\cref{sec:noisydecomps} of the main text, we showed that an optimal and efficient oracle exists for many important two-qubit gates and channels by giving simple analytic decompositions. In particular, in~\cref{sec:noisydecomps}  we saw that noise often reduces the overhead of simulating a circuit in this context. For certain noisy gates where the optimal decomposition is not known, we give decompositions with cost reduced relative to the noiseless case.

\subsection{Modifying the simulator for adaptive channels \label{sec:modified_simulator_sparsifications}}

For certain noisy channels we were not able to find an improved decomposition in terms of convex unitary channels alone. Instead we obtain improved average cost by augmenting the set of convex-unitary channels with adaptive Gaussian channels. In particular we saw this for $ZZ$-rotations subject to single-qubit~$Z$ noise in~\cref{sec:noisydecomps}. 
More generally, it may be of interest to simulate circuits that involve intermediate measurement and feed-forward as part of the quantum algorithm.  

In these cases we must contend with the fact that the mid-circuit measurement probabilities are state-dependent and not known in advance, and must be estimated on the fly by the simulator. We have already shown in~\cref{sec:sparseTechnical} that the ensemble sampling sparsification method can be extended to channels involving Gaussian adaptive channels interleaved with non-Gaussian unitary operations. Here we sketch how the sampling simulator already described can be modified for this setting.

We assume without loss of generality the target circuit with~$N_c$ circuit elements has the form
\begin{equation}
\chan_{\mathrm{Tot}} = \chan^{(N_c)} \circ \chan^{(N_c -1)} \circ \ldots \circ \chan^{(1)} \label{eq:noisyAdaptiveCircuit}
\end{equation}
where each circuit element has a known decomposition
\begin{equation}
\chan^{(j)} = (1-p^{(j)}) \curlU^{(j)} + p^{(j)}\mathcal{A}^{(j)}
\end{equation}
where each~$\curlU^{(j)}$ is a non-Gaussian unitary map with known decomposition as superposition of Gaussian gates, and each~$\mathcal{A}^{(j)}$ is an adaptive channel with free Gaussian Kraus operators~$\smash{K^{(j)}_0}$ and~$\smash{K^{(j)}_1}$ (for example, representing a single-qubit measurement followed by a Gaussian unitary conditioned on the outcome). Expanding out all elements, the total channel can be re-written as 
\begin{equation}
\chan_{\mathrm{Tot}} = \sum_{\arr{z}} p_{\arr{z}} \chan_{\arr{z}}
\end{equation}
where~$\arr{z}$ is a bit-string representing a trajectory through the circuit, such that~$z_j=0$ means the unitary path is taken at step~$j$, and~$z_j=1$ means the adaptive channel is chosen. Then $\{p_{\arr{z}}\}$ is a product probability distribution that is known at the outset and can be efficiently sampled from. Each~$\chan_{\arr{z}}$ is a properly normalized CPTP map, comprising a sequence of unitary or adaptive channels. We can write each of these in the form
\begin{equation}
\chan_{\arr{z}} =\curlU_T \circ \curlA_T \circ \ldots \curlU_1 \circ \curlA_1 \circ \curlU_0  
\end{equation}
where we have merged the gates in between each adaptive channel, and note that the number of adaptive channels~$T$ will depend on~$\arr{z}$. Notice that this interleaved circuit is the setting described in~\cref{sec:sparseTechnical}. For non-trivial noisy circuits, the probability of selecting the adaptive channel at a given step can be expected to be low, so that for typical trajectories~$T \ll N_C$.

 The simulation task is then to sample a bit-string~$x$ of length~$L$ from a distribution
\begin{equation}
\Pr(x) = \Tr[\Pi_x \chan_{\mathrm{Tot}} ( \op{\phi_0}) ] = \sum_{\arr{z}}  p_{\arr{z}} \Tr[ \Pi_x \chan_{\arr{z}} (\op{\psi_0})] 
\end{equation}
where~$\Pi_x$ is the projector onto the outcome~$x$ of~$Z$ measurements on some subset of~$L$ qubits. The high-level outline of an exact simulator is as follows.
\begin{enumerate}
\item Select a CPTP trajectory~$\chan_{\arr{z}} =\curlU_T \circ \curlA_T \circ \ldots \curlU_1 \circ \curlA_1 \circ \curlU_0$ with probability~$p_{\arr{z}}$.
\item For~$t = 1$ to T:
\begin{enumerate}
\item Update the description of the pure state to~$\ket{\psi'_t} = U_{t-1} \ket{\psi_{t-1}}$. Note that this causes a branching of the superposition, leading to the effective Gaussian rank $\chi_t$ increasing at each time step. This takes time~$O(\chi_t)$.
\item Compute probabilities~$p_{t,y_t}=\norm{K_{t,y_t}\ket{\psi'_t}}^2$ of outcomes~$y_t\in{0,1}$ for the adaptive channel~$\curlA_t$ acting on the state. The Kraus operators are assumed free, so for exact simulation this takes time~$O(\chi^2_t)$. (With fast norm estimation would be~$O(\chi_t)$.)
\item Sample outcome~$y_t$ with the probabilities computed in the previous step to select the update~$\ket{\psi_t} = K_{t,y_t} \ket{\psi'_t}/\sqrt{p_{t,y_t}}$.
\end{enumerate}
\item For~$i = 1$ to~$L$:
\begin{enumerate}
\item Compute the probability~$\Pr(x_i = 0 \mid \arr{x_{i-1}}, \psi_T)$, which is the probability of sampling zero for the~$i$-th bit~$x_i$, given the bits~$\arr{x}_{i-1}$ sampled so far (defining~$\arr{x}_{0}$ to be the empty string) and final circuit state~$\ket{\psi_T}$ (we absorb the final unitary~$U_T$ into~$\psi_T$). This can be done exactly in time~$O(\chi^2_T)$.
\item Select~$x_i=0$ or~$x_i=1$ given the distribution just computed, and project the state accordingly.
\end{enumerate}
\item Output~$x$.
\end{enumerate}
The probability of sampling~$x$ via this procedure is given by 
\begin{equation}
P_x = \sum_{\arr{z}}\sum_{\psi_T \mid \arr{z}}\prod_{i=1}^L \Pr(x_i \mid \arr{x}_{i-1},\psi_T)\Pr(\psi_T \mid \arr{z}) \cdot p_{\arr{z}}
\end{equation}
where~$\Pr(\psi_T \mid \arr{z})$ is the probability of having obtained final state~$\psi_T$ during circuit simulation given the trajectory~$\arr{z}$, and the summation is over all possible such trajectories, and all possible pure state outcomes on that trajectory~$\psi_T\mid\arr{z}$.  The probability of obtaining the state~$\ket{\psi_T}$ depends on the sequence of outcomes selected for the adaptive channels. Let~$K'_{t,y_t} = K_{t,y_t} U_{t-1}$, so that
\begin{equation}
\ket{\psi_t} =  \frac{K'_{t,y_t} \ket{\psi_{t-1}}}{\sqrt{p_{t,y_t} }}, \quad p_{t,y_t} = \norm{K'_{t,y_t} \ket{\psi_{t-1}}}^2 \ .
\end{equation}
Here $p_{t,y_t}$ are the conditional probabilities of obtaining the outcome $y_t$ from the $t$-th intermediate measurement, given normalized pre-measurement state $\ket{\psi_{t-1}}$, and in general depend on the preceding sequence $(y_1,\ldots,y_{t-1})$. Therefore the probability $\Pr(\psi_T \mid \arr{z})$ is equal to the probability $\Pr(y) = \prod_{t=1}^T p_{t,y_t}$ of obtaining some bit string $y$ encoding the outcomes of the sequence of intermediate measurements,
\begin{equation}
\Pr(y) = \prod_{t=1}^T \norm{K'_{t,y_t}\ket{\psi_{t-1}}}^2 =  \prod_{t=1}^T \frac{\norm{K'_{t,y_t}\ldots K'_{1,y_1} \ket{\phi_0}}^2}{\norm{K'_{t-1,y_{t-1}}\ldots K'_{1,y_1} \ket{\phi_0}}^2}= \norm{K'_{T,y_T}\ldots K'_{1,y_1} \ket{\phi_0}}^2 \ .
\end{equation}
Here $K'_{T,y_T}\ldots K'_{1,y_1} = K^{(\arr{z})}_{\arr{y}}$ in the notation of~\cref{sec:sparseTechnical}, where $\arr{y}$ indexes the Kraus operators of the channel $\chan_{\arr{z}}$. It follows that the expected final state 
\begin{equation}
\rho = \sum_{\arr{z}} p_{\arr{z}} \sum_{\arr{y}} K^{(\arr{z})}_{\arr{y}} \op{\phi_0} \left(K^{(\arr{z})}_{\arr{y}}\right)^\dagger
\end{equation}
is the correct density operator for the channel $\chan$, and hence the bit string $x$ will be drawn from the true distribution.

The runtime for producing one sample using the exact simulator sketched above scales as~$O(\chi^2_{\arr{z}})$ where the effective Gaussian rank $\chi_{\arr{z}}$ depends on the channel trajectory ${\arr{z}}$ chosen in step 1, but not on which Kraus operators were selected at each step $t$. Therefore the expected runtime will be~$O(\sum p_{\arr{z}} \chi^2_{\arr{z}})$. As in the convex-unitary case, to achieve runtime linear in the rank, we can use the fast norm method to estimate the probabilities $p_{t,y_t}$ at the cost of incurring some multiplicative estimation error in the usual way. Finally we can reduce the effective rank using an approximate sparsified description of each selected  channel $\chan_{\arr{z}}$ as explained in~\cref{sec:sparseTechnical}. In this case we should expect the average runtime to scale as $O(\prod_j \Xiaug(\cal{E}^{(j)}))$, where $\Xiaug(\cal{E}^{(j)})$ is the augmented extent for each channel in the original circuit~\cref{eq:noisyAdaptiveCircuit}. We leave a rigorous proof of runtime guarantees for the approximate algorithm for future work.

\section{Decomposition of a convex-unitary diagonal channel  \label{sec:DIAGONAL}}
Here we argue that the unitary terms in any convex-unitary decomposition of a diagonal channel must themselves be diagonal. This argument applies specifically to convex linear combinations of channels, as one can easily construct non-convex linear combinations of channels where non-diagonal elements cancel each other out. We first review some basic properties of linear maps. Any linear map~$T$ on a space of~$n$ qubits can be decomposed in the Pauli basis as 
\begin{equation}
T(A) = \sum_{j,k} c_{j,k} P_j A P_k, \label{eq:generalmap}
\end{equation} and as a transfer matrix in the form
\begin{equation}
\hat{T} = \sum_{j,k} c_{j,k} P_j \otimes \overline{P}_k
\end{equation}
where~$c_{j,k} \in \mathbb{C}$ and~$j,k$ index a pair of~$n$-qubit Pauli operators~$(P_j,P_k)$. The dual~$T^*$ of a map is defined as the unique map such that 
\begin{equation}
\Tr[A T(B)] = \Tr[T^*(A) B].
\end{equation}
A linear map is trace-preserving if and only if its dual satisfies~$T^*(\id) = \id$~\cite[Theorem 1.55]{wolf2022mathematical}. 
In the Pauli basis, the dual of the general map~\eqref{eq:generalmap} is  written
\begin{equation}
T^*(A)  = \sum_{j,k} c_{j,k} P_k A P_j.
\end{equation}
The trace-preserving condition then becomes
\begin{equation}
	\label{eq:tracepreserving_necessarycondition}
T^*(\id) = \sum_{j,k} c_{j,k} P_k P_j = \id \quad \Rightarrow \quad  2^n = \Tr[T^*(\id)] = \sum_{j} c_{j,j} 2^n
\end{equation}
where the second equation follows from the orthogonality of the Pauli operators,~$\Tr[P_j P_k] = \delta_{j,k} 2^n$. Therefore~$\sum_j c_{j,j} = 1$ is a necessary but not sufficient condition for a map to be trace-preserving.

Consider an~$n$-qubit unitary matrix~$U = \sum_j \alpha_j P_j$ decomposed in the Pauli basis, and the corresponding linear map~$\cal{U} (\cdot) = U (\cdot) U^\dagger$. The coefficients of the transfer matrix of the map~$\cal{U}$ are~$c_{j,k} = \alpha_j \overline{\alpha_k}~$, with the~$j=k$ elements satisfying 
\begin{equation}
c_{j,j} = \abs{\alpha_j}^2, \quad \sum_j c_{j,j} = 1.
\end{equation}
Now, consider the following twirling operation (see e.g.,~\cite{PhysRevLett.76.722,PhysRevA.54.3824}) which acts on the space of~$n$-qubit transfer matrices (that is, on~$4^n \times 4^n$ matrices),
\begin{equation}
\mathcal{T}_Z ( \cdot ) = \frac{1}{4^n} \sum_{x \in \{0,1\}^{2n}} Z(x) (\cdot) Z(x),
\end{equation}
where~$Z(x) = \bigotimes_{j=1}^{2n} Z_j^{x_j}$. This map projects any~$4^n \times 4^n$ matrix onto the commutant of the group~$\mathcal{G}$ of~$2n$-qubit~$Z$-strings,
\begin{equation}
\mathcal{G} = \langle Z_1, Z_2, \ldots , Z_{2n}\rangle.
\end{equation}
For unsigned Pauli strings~$P$, we have
\begin{equation}
\mathcal{T}_z (P)  = \begin{cases}
P  \quad \mathrm{if} \quad P \in \mathcal{G} \\
0 \quad \mathrm{otherwise}.
\end{cases}
\end{equation}
Since all diagonal matrices commute, for any diagonal unitary~$U_D$, the transfer matrix is invariant under the~$Z$-twirling map, i.e., 
\begin{equation}
\mathcal{T}_Z (\hat{\mathcal{U}}_D) = \mathcal{T}_Z ( U_D \otimes \overline{U}_D ) = U_D \otimes \overline{U}_D,
\end{equation}
and similarly for any convex combination of diagonal unitary channels.

On the other hand, consider a non-diagonal unitary operator~$V = \sum_j \beta_j P_j$, with transfer matrix
\begin{equation}
V \otimes \overline{V} = \sum_{j,k} \beta_j \overline{\beta_k} P_j \otimes \overline{P}_k.
\end{equation}
Before twirling we have~$\sum_{j} \abs{\beta_j}^2 = 1$, by the trace-preserving property. Define the index set
\begin{equation}
\mathcal{S} = \{ (j,k) : P_j \otimes P_k \in \mathcal{G} \}.
\end{equation}
Applying the~$Z$-twirling map, only~$Z$-strings survive, so 
\begin{equation}
\mathcal{T}_Z ( V \otimes \overline{V}) =\sum_{j,k} q_{j,k} P_j \otimes \overline{P}_k,
\end{equation}
where
\begin{equation}
q_{j,k} = \begin{cases} \beta_j \overline{\beta_k} \quad \mathrm{if} \quad (j,k) \in \mathcal{S} \\
 0 \quad \mathrm{otherwise}.
\end{cases}
\end{equation}
If~$V$ is not diagonal, there is at least one~$\beta_j \neq 0$ such that the Pauli string~$P_j$ contains an~$X$ or~$Y$. Then the Pauli decomposition of the transfer matrix~$V \otimes \overline{V}$ contains a term~$ \beta_j \overline{\beta_j} P_j \otimes \overline{P}_j$ proportional to a non-diagonal Pauli string. This will be zeroed out by the twirling operation, so~$q_{j,j} = 0~$. Since all the diagonal elements~$q_{j,j}$ of the twirled transfer matrix are either equal to~$\abs{\beta_j}^2$ or zero, and there is at least one~$q_{j,j} = 0$, we must have~$\sum_j q_{j,j}  < 1$. Let $T'$ be the linear map represented by the transfer matrix . Next we argue that twirling does not increase the operator norm of the dual applied to the identity, so $\norm{T'^*(\id)}\leq 1$.
Each $2n$-qubit Pauli string $Z(x) = Z'(r) \otimes Z'(s)  $ string can be rewritten as a product of two $n$-qubit strings $Z'(y) = \bigotimes_{j=1}^{n} Z_j^{y_j}$, for some $r$ and $s$. So the twirled transfer matrix can equivalently be represented as
\begin{align}
\mathcal{T}_Z ( V \otimes \overline{V})  &= \frac{1}{4^n}\sum_{r,s \in \{0,1\}^{n}} (Z'(r) \otimes Z'(s))( V \otimes \overline{V} )(Z'(r) \otimes Z'(s)) \\
	& = \frac{1}{4^n}\sum_{r,s \in \{0,1\}^{n}} (Z'(r) V Z'(r)) \otimes (Z'(s) \overline{V}  Z'(s)) \\
	& = \frac{1}{4^n} \sum_{r,s \in \{0,1\}^{n}} V_r \otimes \overline{V}_s
\end{align}
where each $V_y = Z'(y) V  Z'(y)$ is some unitary matrix. The linear map represented by $\mathcal{T}_Z ( V \otimes \overline{V}) $ can therefore be written 
\begin{equation}
T'(\cdots) = \frac{1}{4^n}\sum_{r,s \in \{0,1\}^{n}} V_r (\cdots)V^\dagger_s,
\end{equation}
and its dual applied to the identity gives
\begin{equation}
T'^*(\id) =  \frac{1}{4^n}\sum_{r,s \in \{0,1\}^{n}} V^\dagger_s V_r.
\end{equation}
Then by the triangle inequality and the unitarity of $ V^\dagger_s V_r$ we obtain
\begin{equation}
\norm{T'^*(\id)}_\infty \leq  \frac{1}{4^n}\sum_{r,s \in \{0,1\}^{n}} \norm{V^\dagger_s V_r}_\infty = 1.
\end{equation}
Finally we use H{\"o}lder's inequality to show that a $Z$-twirled non-diagonal unitary map must be strictly trace-decreasing. One can check by symmetry of the twirling map preserves the hermiticity-preserving property of $\mathcal{V} = V(\cdots)\overline{V}$. It follows that the trace of $T'(\rho)$ is real for any density matrix $\rho$. Then using H{\"o}lder's inequality,
\begin{equation}
\Tr[T'(\rho)] = \Tr[T'^*(\id) \rho] \leq \norm{T'^*(\id) \rho}_1 \leq \norm{T'^*(\id)}_{\infty} \norm{\rho}_1 \leq \Tr(\rho).
\end{equation}
In other words, the twirled map $T'$ is trace-non-increasing. But we saw earlier that when $V$ is non-diagonal, the twirled map $T'$ cannot be trace-preserving. Therefore it must be strictly trace-decreasing.

Now assume for contradiction that the convex-unitary channel~$\chan_D$ is diagonal, and has a convex decomposition
\begin{equation}
\hat{\chan}_D = \sum_i p_i V_i \otimes \overline{V}_i 
\quad\text{ with }\quad 0 \leq p_i \leq 1, \sum_{i}p_i = 1
\end{equation}
where some of the unitary matrices~$V_i$ are non-diagonal. Then applying the twirling map, we obtain
\begin{equation}
\mathcal{T}_Z( \hat{\chan}_D) = \sum_i p_i \mathcal{T}_Z(V_i \otimes \overline{V}_i) 
\end{equation}
On the left-hand side,~$\mathcal{T}_Z( \hat{\chan}_D) = \hat{\chan}_D$, so it remains a representation of a CPTP map. On the right-hand side, each~$V_i$ is either diagonal or it is not. If it is diagonal, then~$\mathcal{T}_Z(V_i \otimes \overline{V}_i) = V_i \otimes \overline{V}_i$ and that term preserves the trace. If it is not diagonal then~$\mathcal{T}_Z(V_i \otimes \overline{V}_i)$ is trace-decreasing. Since we have assumed that some unitary terms~$V_i$ are not diagonal, then the right-hand side is a convex combination of trace-preserving and trace-decreasing terms, so is trace-decreasing overall. But the left-hand side is trace-preserving, so this is a contradiction. Therefore all feasible decompositions of~$\chan_D$ as a convex combination of unitary channels must be termwise diagonal. Note that this result holds only for convex combinations of unitary maps. For quasiprobability distributions, where the coefficients $p_i$ can be negative, we can add and subtract non-diagonal terms that cancel each other out.

\end{document}